\documentclass[11pt,a4paper]{article}

\usepackage[a4paper,hmargin=20mm,vmargin=25mm]{geometry}
\usepackage[english]{babel}
\usepackage[mac]{inputenc}
\usepackage[T1]{fontenc}

\usepackage{hyperref}
\hypersetup{colorlinks=true, linkcolor=black,citecolor=black}

\usepackage{tikz}
\usepackage{appendix}
\usepackage{dsfont}

\usepackage{subfig}

\usepackage{amsmath}
\usepackage{amsbsy}
\usepackage{amsfonts}
\usepackage{amssymb}
\usepackage{amscd}
\usepackage{amsthm}
\usepackage{mathtools}

\theoremstyle{definition}
\newtheorem{remark}{Remark}
\newtheorem{definition}{Definition}[section]

\theoremstyle{plain}
\newtheorem{proposition}[definition]{Proposition}
\newtheorem{hyp}{Assumption}
\newtheorem{theorem}[definition]{Theorem}
\newtheorem{lemma}[definition]{Lemma}
\numberwithin{equation}{section}

\providecommand{\abs}[1]{\left\lvert#1\right\rvert}
\providecommand{\sabs}[1]{\lvert#1\rvert}
\providecommand{\babs}[1]{\bigg\lvert#1\bigg\rvert}
\providecommand{\norm}[1]{\left\lVert#1\right\rVert}
\providecommand{\bnorm}[1]{\bigg\lVert#1\bigg\rVert}
\providecommand{\snorm}[1]{\lVert#1\rVert}
\providecommand{\prs}[1]{\left\langle #1\right\rangle}

\providecommand{\bprs}[1]{\bigg\langle #1\bigg\rangle}

\providecommand{\ie}{\textit{i.e.} }
\providecommand{\PPP}{\R^J\times[-\pi,\pi[^J\times\R^{2J}}

\providecommand{\R}{\mathbb{R}}
\providecommand{\RR}{\mathbb{R}}

\renewcommand{\P}{\mathbb{P}}

\providecommand{\Z}{\mathbb{Z}}

\providecommand{\bV}{\mathcal{V}}
\providecommand{\BV}{\ensuremath{\boldsymbol{\bV}} }

\providecommand{\1}{\mathds{1}}

\providecommand{\Aa}{\ensuremath{\boldsymbol{a}} }
\renewcommand{\b}{\ensuremath{\boldsymbol{b}\vphantom{b}} }

\providecommand{\x}{\ensuremath{\boldsymbol{x}} }
\providecommand{\y}{\ensuremath{\boldsymbol{y}} }

\providecommand{\N}{\mathbb{N}}
\providecommand{\C}{\mathbb{C}}

\providecommand{\Sh}{\ensuremath{\R^{k\times 2}}}
\DeclareMathOperator{\tr}{tr}                   % Trace

\providecommand{\Y}{\mathbf{Y}}
\providecommand{\G}{\mathcal{G}}

\providecommand{\bh}{\boldsymbol{h}}
\providecommand{\f}{\boldsymbol{f}\vphantom{f}} 
\renewcommand{\S}{\mathbb{S}}

\providecommand{\A}{\mathbf{A}}
\renewcommand{\AA}{\mathcal{A}}

\newcommand{\one}{\ensuremath{\boldsymbol{1}} }
\newcommand{\bTheta}{\ensuremath{\boldsymbol{\Theta}} }

\providecommand{\balpha}{\ensuremath{\boldsymbol{\alpha}} }
\newcommand{\bzeta}{\ensuremath{\boldsymbol{\zeta}} }

\providecommand{\E}{\mathbb{E}}
\newcommand{\bSigma}{\ensuremath{\boldsymbol{\Sigma }}}

\DeclareMathOperator*{\argmin}{\mbox{argmin}}

\DeclareMathOperator*{\diag}{\mbox{Diag}}
\providecommand{\LL}{\mathcal{L}}
\makeatother

\title{Consistent estimation of a mean planar curve modulo similarities}

\author{
  {\em J\'er\'emie Bigot}\footnote{Institut Sup\'erieur de l'A\'eronautique et de l'Espace,
D\'epartement Math\'ematiques, Informatique, Automatique,
10 Avenue \'Edouard-Belin, BP 54032-31055, Toulouse CEDEX 4, France. Email: jeremie.bigot@isae.fr},  \quad DMIA/ISAE - Universit\'e de Toulouse \\
  {\em Benjamin.Charlier}\footnote{Institut de Math\'ematiques de Toulouse,
Universit\'e de Toulouse et CNRS (UMR 5219),
31062 Toulouse, Cedex 9, France. Email: benjamin.charlier@math.univ-toulouse.fr}, \quad IMT/CNRS (UMR 5219) - Universit\'e de Toulouse }

\begin{document}
\maketitle

\begin{abstract}
We consider the problem of estimating a mean planar curve from a set of $J$ random planar curves observed on a $k$-points deterministic design. We study the consistency of a smoothed Procrustean mean curve when the observations obey a deformable model including some nuisance parameters such as random translations, rotations and scaling. The main contribution of the paper is to analyze the influence of the dimension $k$ of the data and of the number $J$ of observed configurations on the convergence of the smoothed Procrustean estimator to the  mean curve of the model. Some numerical experiments illustrate these results.
\end{abstract}



\subsubsection*{Acknowledgements}

The authors acknowledge the support of the French Agence Nationale de la Recherche (ANR) under reference ANR-JCJC-SIMI1 DEMOS.

\section{Introduction}

In this paper, we are interested in the statistical treatment of random planar curves observed through rigid deformations of the plane. In many fields of interest, such data appear as, for instance, contours extracted from digital images or level sets of a real function defined on the plane. In handwriting recognition problems one typically compares curves that describe letters, digits or signatures %On real data, the variability of such processes are rather complex to model and can be onny well explain by non-rigid deformable model. 
and the acquisition process often create some ambiguity of location, size and orientation. %Obviously, the natural variability of handwriting is much more complex and a realistic statistical model should involve local deformations (\eg non-rigid deformations). Nevertheless, it is far beyond the scope of this paper to give a comprehensive model of handwriting process. 
Our aim is then to study an estimation procedure for a mean curve from a sample of $J$ noisy and discretized planar curves observed through translation, rotations and scaling. The group generated by this set of transformations is usually called \emph{similarity group} of the plane and in the sequel, an element of this group will be called a \emph{deformation}. 

% will designate a similarity of the plane. We assume also that the observations are random curves distributed around a mean curve. In this framework, we are interested in the estimation of this mean curve. 
%In many applications, the quantities of interest are planar curves. For example, in handwriting recognition one compares curves that describe letters, digits or signatures. Other examples of such data are contours extracted from digital images or level curves of real functions defined on $\R^2$. As data are often collected through a computer device (\eg a digital camera or an image scanner) there may be an  ambiguity of location, scale and translation in the observations. The group generated by this set of transforms is usually called \emph{similarity group} of the plane. In the sequel, an element of this group will be called a \emph{deformation}. % will designate a similarity of the plane. We assume also that the observations are random curves distributed around a mean curve. In this framework, we are interested in the estimation of this mean curve. We study the consistency of the proposed estimators with respect to a deformable statistical model. 

% observed through the action of the group generated by the rotations, scalings and translations in $\R^2$. 

\subsection{A deformable model coming from statistical shape analysis}\label{part.Intro.1}

%As data are often collected through a computer device (\eg a digital camera or an image scanner)
\paragraph{Deformable model} In many practical cases of interest, data  are collected through a computer device such as a digital camera or an image scanner. In order to take this into account in our model, we assume that observations at hand  are  discretized versions of continuous planar curves. Each observation $\Y$ is then given by a set of $k$ points of the plane called a \emph{configuration}. It can be written as a $k\times 2$ real matrix
\[
\Y = \begin{pmatrix} Y_{1}  \\ \vdots \\ Y_{k} \end{pmatrix} \in\R^{k\times 2},
\]
where $Y^{\vphantom{(1)}}_{\ell} = \big(Y^{(1)}_\ell,Y_\ell^{(2)} \big) \in \R^2$ with  $\ell=1,\ldots,k$. A degenerated configuration is a configuration composed of $k$ times the same point $b = (b^{(1)},b^{(2)}) \in \R^2$. In tensorial notation, such a configuration is written $\1_k \otimes b\in\Sh$ where $\1_k$ denotes the (column) vector of $\R^{k}$ with all entries equal to one and $\otimes$ denotes the tensor product. From now on, we assume that the observations satisfy the following functional regression model for  $j = 1,\ldots,J$,
\begin{equation}\label{eq:perturbmodel}
\Y_j = e^{a_{j}^*}(\f + \bzeta_{j} )R_{\alpha_{j}^*} +  \1_k \otimes b_{j}^*,  \quad \mbox{ with } R_{\alpha_{j}^*} = \begin{pmatrix} \cos(\alpha_{j}^*) &  -\sin(\alpha_{j}^*) \\   \sin(\alpha_{j}^*) & \cos(\alpha_{j}^*) \end{pmatrix},
\end{equation}
where the unknown mean pattern $\f \in \R^{k\times 2}$ has been obtained by sampling on an equi-spaced design a planar curve $f =(f^{(1)},f^{(2)}): [0,1] \longrightarrow \R^2$  satisfying $f(0) =f(1)$. It means that we have 
\[
%\f= \big(f(\tfrac{\ell}{k})\big)_{\ell=1}^{k}. %, \quad \ell=1,\ldots,k.
	\f = \big(f(\tfrac{\ell}{k})\big)_{\ell=1}^k := \big(f^{(1)}(\tfrac{\ell}{k}),f^{(2)}(\tfrac{\ell}{k})\big)_{\ell=1}^k \in \Sh.
\] 
%where the mean pattern  $\f = \big(f^{(1)}(\tfrac{\ell}{k}),f^{(2)}(\tfrac{\ell}{k})\big)_{\ell=1}^k \in \R^{k\times 2}$ is an unknown configuration of $k$ points of the plane coming from the discretization of a continuous %, closed 
%plane curve $f=(f^{(1)},f^{(2)}):[0,1] \longrightarrow \R^2$ satisfying $f(0) =f(1)$.
The error terms $\bzeta_{j} \in\R^{k\times 2},\ j=1,\ldots,J$ are independent copies of a random perturbation $\bzeta$ in $\R^{k\times 2}$ with zero expectation. For  $j=1,\ldots,J$, the scaling, rotation and translation parameters $(a_j^*,\alpha_j^*,b_j^*) \in  \R \times [-\pi,\pi[ \times \R^2$ are independent and identically distributed (i.i.d)  random variables independent of the random perturbations $\bzeta_{j}$.

	Model \eqref{eq:perturbmodel} is a deformable model in the sens of \cite{BC11} and our aim is to estimate the underlying mean curve $f:[0,1] \longrightarrow \R^2$ from $\Y_1,\ldots,\Y_J\in \Sh$. More precisely, we study the influence on the estimation procedure of the number $J$ of configurations at hand and the number $k$ of discretization points composing the configurations $\Y_j$'s. Note that the model \eqref{eq:perturbmodel} has been introduced by Goodall in \cite{MR1108330} but with a fixed number $k$ of points in each configuration and non-random nuisance parameters $(a^*_j,\alpha^*_j,b^*_j)$'s. This framework has been highly popular in the statistical shape community, %and it was already considered by various authors but with a fixed number of point $k$ in each observation
see \cite{MR1436569,MR1618880,MR1225013}. Note that the configuration $\f$ was called a population mean in \cite{MR1108330} or a perturbation mean in \cite{HuckemannSJS}. 

\paragraph{Shape of a configuration.} Since the seminal work of Kendall \cite{MR737237}, one considers that the shape of a configuration $\Y\in\Sh$ is ``what remains when translations, rotations and scaling are filtered out''%$Y_{\vphantom{1}}^{(\ell)} = \big(Y^{(\ell)}_1,Y^{(\ell)}_2 \big) \in \R^2,\ \ell=1,\ldots,k$  and modulo  the effects of translation, rotation and scaling
. More precisely, two configurations $\Y_1,\Y_2 \in \R^{k\times 2}$ are said to have the same shape if there exists a vector $(a,\alpha,b)\in \R\times[-\pi,\pi[\times\R^2$ such that 
\begin{equation}\label{eq.deform}
\Y_2 = e^{a}\Y_1 R_{\alpha} + \1_k \otimes b,\quad \mbox{ with } R_{\alpha_{}} = \begin{pmatrix} \cos(\alpha_{}) &  -\sin(\alpha_{}) \\   \sin(\alpha_{}) & \hphantom{-}\cos(\alpha_{}) \end{pmatrix},
\end{equation}
The Kendall's Shape Space is the quotient space modulo this equivalent relation. It is usually denoted by $\bSigma_2^k$ and defined as the set of normalized (\ie centered and scaled to size one) configurations quotiented by the rotation of the plane, see part \ref{part:ShapeSpace} for further details. 

The above definition can be trivially extended to the case of planar curves by replacing the configuration $\f\in\Sh$ by the planar curve $f:[0,1] \longrightarrow \R^2$. Thence, we call \emph{shape} of $f$ the set of %continuous closed 
planar curves  that can be written $ e^a f R_\alpha + \1 \otimes b$ for some $(a,\alpha,b)\in\R\times[-\pi,\pi[\times \R^2$.

\subsection{Estimation of the mean curve and of the shape of the mean curve}

\paragraph{Consistency in the Shape Space.} In shape analysis, an important issue is the computation of a sample mean shape from a set of $J$ random planar configurations $\Y_{1},\ldots,\Y_{J}\in \Sh$ satisfying model \eqref{eq:perturbmodel} and the study of its consistency as the number of samples  $J$ goes to infinity ($k$ remaining fixed).  %A statistical model of shapes must include some nuisance parameters associated to the ambiguity of location, rotation and scaling. In \cite{MR1108330}, consistent estimation of a mean shape is therefore considered in the following deformable model:  
According to Goodall \cite{MR1108330} a sample mean pattern $\hat \f\in\Sh$ computed from $\Y_{1},\ldots,\Y_{J}$  is said to be consistent if, as $J\to\infty$, it has asymptotically the same shape than the mean pattern $\f$. %Since Goodall's proposal, the deformable model \eqref{eq:perturbmodel} has been highly popular in the statistical shape community, and 
In this framework, the deformations parameters are considered as \emph{nuisance parameters} that  contain no informations. That is why the data are first normalized (\ie centered and scale to unit size) without changes in the statistical analysis. The study of consistent procedures to estimate the shape of the mean pattern $\f$ using this approach has been considered by various authors \cite{MR1436569,MR1618880,MR1225013,HuckemannSJS}. In this setting,  sample mean patterns obtained by a Procrustes procedure have received a special attention. In particular, it is shown in \cite{MR1436569,MR1618880} that, in the very specific case of isotropic perturbations  $\bzeta$, the so-called full and partial Procrustes sample means are  consistent estimators of the shape of $\f$. Nevertheless, these estimators can be inconsistent for non-isotropic  perturbations. Therefore, it is generally the belief that consistent statistical inference based on Procrustes analysis is  restricted to very limited assumptions on the distribution of the data, see also \cite{MR1646114,MR1891212,HuckemannSJS} for further discussions. 

%The aim of this paper is to show that a Procrustes sample mean can be considered as a consistent procedure even in the case of non-isotropic  perturbations.  To this purpose, we propose to exhibit the relation that exists between the dimensionality $k$ of the data and the consistency of smoothed Procrustes sample means in the perturbation model \eqref{eq:perturbmodel}. Our main result (see Theorem \ref{theo:main} below) is that when the dimensionality $k$ is high and the mean pattern $\f$  is the discretization of a sufficiently smooth plane curve, it is possible to build a consistent estimator of the shape of $\f$  in model \eqref{eq:perturbmodel} under general assumptions on the perturbation $\bzeta_{j}$. The problem of analyzing high-dimensional 2D configurations (i.e.\ when the number of landmarks $k$ is high) arises in the statistical study of a set of random  points in $\R^2$ that have been sampled from planar curves. A typical example is the analysis of contours extracted from digital  images.

\paragraph{Estimation of the mean curve.} Our approach differ from the one developed in Shape analysis in two main aspects.  First, we assume that the unknown mean pattern $\f \in \R^{k\times 2}$ in model \eqref{eq:perturbmodel} has been obtained by  sampling a planar curve $f : [0,1] \longrightarrow \R^2$ on an equi-spaced design. This allows us to study the influence of the number $k$ of points composing each configuration on the estimation of $f$. Note that if $f$ is sufficently regular and $k$ is large, to estimate $\f\in\Sh$ (\ie the values of $f$ on the design) roughly amounts to estimate $f$. % we make no difference between the estimation of $f$ or $\f=f(\tfrac{\ell}{k})$. % consider the quantity $\tfrac{1}{k}\snorm{\f - \f}_{\Sh}$ rather than the continuous version $\snorm{\hat f - f}_{L^2}$. 
The second difference concerns the randomness of the deformation parameters $(a^*_j,\alpha^*_j,b^*_j)\in\R\times[-\pi,\pi[\times \R^2$, $j=1,\ldots,J$, that we want to estimate rather than to consider as nuisance parameters. Indeed, the values of the deformations parameters can be informative in some cases : assume that the size of the data are very similar except for one. This difference in size may be due to a (possibly) relevant factor and if the data are normalized as in Shape analysis this %(possibly relevant) 
	information is lost.  %This problem can be transpose to more complicated deformation and.  
%Indeed, it can be interesting to compare observations not only with shape but also with the amplitude of the deformations parameters, see \cite{}. 

	Under a suitable smoothness assumption on $f$, we are able to estimate consistently,  with an asymptotic in $k$ only, the curve  $f_j := e^{a_j^*} f R_{\alpha_j^*} + \1\otimes b_j^* $ from $\Y_j$ for each $j=1,\dots,J$. %These estimates are based on a low-pass filter method and it is worth noticing 
	Note that, by definition,  the $f_j$'s  have the same shape as $f$. When $J$ is fixed, this smoothing step allows us to estimate consistently the deformations parameters $(a^*_j,\alpha_j^*,b_j^*)_{j=1}^J$ with a Procrustes matching step. Now, if we assume that the deformations parameters $(a^*_j,\alpha_j^*,b_j^*) $ are centered random variables, we show that it is possible to recover the \emph{true} mean shape $f$ when $\min\{J,k\}\to \infty$. These results are rigorously stated in the next section, see Theorem \ref{theo:main} below. 

%	\bigskip	

%	In this paper, we give conditions that ensure the consistency of the Procrustes type procedure in a double asymptotic framework in $J$ and $k$. The condition on the error term $\bzeta$ are given by Assumption \ref{ass.zeta} and allows consistency even with a non-isotropic noise  but at the price of a convenient smoothing step.  % of curves at hand goes to infinity. %, an increasing of $k$ allows to reconstruct a function with the same shape of $f$    %of randomness of the $ (a_j^*,\alpha_j^*,b_j^*)$'s is that 

%un and was motivated by \dots consideration. 

\subsection{Main contribution}
 
Our estimating procedure is composed of two steps. First, we perform  a dimension reduction step by projecting the data into a low-dimensional space of $\R^{k\times 2}$ to eliminate the influence of the random perturbations  $\bzeta_{j}$.  Then, in a second step, we apply  Procrustes analysis in this low-dimensional space  to obtain a consistent estimator of  $\f\in\Sh$.

The reduction dimension step is based on an appropriate smoothness assumptions on $f$. Let $L>0$, $s > 0$ and define the Sobolev ball of radius $L$ and degree $s$ as
\begin{equation}\label{eq:H}
H_s(L) = \Big\{f = (f^{(1)},f^{(2)})\in L^2([0,1],\R^2), \ \sum_{m\in\Z} (1 + \sabs{m}^2)^s (\sabs{ c_m(f^{(1)})}^2 +\sabs{ c_m(f^{(2)})}^2  ) <L \ )  \Big\}
\end{equation}
where $c_m(f)= (c_m(f^{(1)}), c_m(f^{(2)}) ) = \Big(\int_{0}^1 f^{(1)}(t) e^{-i2\pi mt}dt,\int_{0}^1 f^{(2)}(t) e^{-i2\pi mt}dt \Big)\in\C^2$ is the $m$-th Fourier coefficient of $f=(f^{(1)},f^{(2)})\in L^2([0,1],\R^2)$, for $m\in\Z$.%where $c_m(f_{p}) = \int_{0}^1 f_{p}(t) e^{-i2\pi mt}dt$ is the $m$-th Fourier coefficient of $f_{p}\in L^2([0,1],\R)$, for $m\in\Z$ and $p=1,2$. 
\begin{hyp}\label{ass.f}
	The curve $f:[0,1]\longrightarrow \R^2$ is closed, \ie $f(0)=f(1)\in\R^2$, and belongs to $H_s(L)$ for some $L>0$ and $s> \frac{3}{2}$. Moreover, the $k\times2$ matrix $\f= \big(f(\tfrac{\ell}{k})\big)_{\ell=1}^{k}$ is of rank two, \ie the configuration $\f$ is not degenerated.
\end{hyp}
\noindent Assumption \ref{ass.f} implies that $f$ is not reduced to a point, continuously differentiable and is equal to its Fourier series. We introduce the following  $k \times k$  matrix
\begin{equation}\label{eq.Alambda}
\A^{\lambda} = \bigg( \frac{1}{k} \sum_{0 \leq |m| \leq \lambda}   e^{i2\pi m\tfrac{\ell-\ell'}{k}} \bigg)_{\ell,\ell'=1}^{k}. 
\end{equation}
The matrix  $\A^{\lambda}$ is  the smoothing matrix corresponding to a discrete Fourier low pass filter with frequency cutoff $\lambda\in\N$. It is a projection matrix in a sub-space $\BV^\lambda$ of $\R^{k}$ of dimension $2 \lambda+1$. Then, we project the data on $\BV^{\lambda}\times\BV^{\lambda}\subset \R^{k \times 2}$, and we estimate the scaling, rotation and translation parameters in model  \eqref{eq:perturbmodel} using M-estimation as follows: denote the scaling parameters by $\Aa = (a_1,\ldots,a_J)\in\R^J$, the rotation parameters by $\balpha = (\alpha_1,\ldots,\alpha_J)\in\R^J$ and the translation parameters by $\b= (b_1,\ldots,b_J) \in\R^{2J}$, and introduce the functional,
\begin{equation}\label{def.Mlambda}
M^{\lambda}(\Aa,\balpha,\b) =\frac{1}{Jk} \sum_{j=1}^J \bigg \|  e^{-a_{j}}\A^{\lambda} (\Y_{j}- \1_k\otimes b_{j} ) R_{-\alpha_{j}}- \frac{1}{J} \sum_{j'=1}^Je^{-a_{j'}}\A^{\lambda} (  \Y_{j'} - \1_k\otimes b_{j'} ) R_{- \alpha_{j'}}\bigg\|^2_{\Sh},
\end{equation}
where  $\snorm{\cdot}_{\Sh}$ is the standard Euclidean norm in $\Sh$. An M-estimator of
$$
\big(\Aa^*,\balpha^*,\b^*\big) = \big( a^{*}_1,\ldots,a^{*}_J, \alpha^{*}_1,\ldots,\alpha^{*}_J,b^{*}_1 ,\ldots, b^{*}_J \big) \in\R^{J}\times[-\pi,\pi[^{J}\times \R^{2J}
$$
is  given by
\begin{equation}\label{def.estimateurs}
(\hat \Aa^{\lambda} , \hat \balpha^{\lambda}, \hat \b\vphantom{\b}^{\lambda}) \in\argmin_{(\Aa,\balpha,\b)\in \bTheta_0 } M^{\lambda}(\Aa,\balpha,\b),
\end{equation}
where $\big(\hat \Aa^{\lambda} , \hat \balpha^{\lambda}, \hat \b\vphantom{\b}^{\lambda}\big) = \big(\hat a^{\lambda}_1,\ldots,\hat a^{\lambda}_J, \hat \alpha^{\lambda}_1,\ldots,\hat \alpha^{\lambda}_J,\hat b\vphantom{\b}^{\lambda}_1 ,\ldots,\hat b\vphantom{\b}^{\lambda}_J \big)\in\R^{J}\times[-\pi,\pi[^{J}\times \R^{2J} $  
and  
\begin{equation}\label{eq:Theta0}
\bTheta_0 = \bigg\{ (\Aa,\balpha,\b) \in [-A,A]^{J} \times [-\AA,\AA]^{J}\times \R^{2J} : \sum_{j=1}^{J} a_{j} =  0,\ \sum_{j=1}^{J} \alpha_{j} = 0 \mbox{ and } \sum_{j=1}^{J} b_{j} = 0  \bigg\},
\end{equation}
with $A,\AA>0$  being parameters whose values will be discussed   below. %Indeed, $ \big(\hat \Aa^{\lambda} , \hat \balpha^{\lambda}, \hat \b\vphantom{\b}^{\lambda}\big)$ are the deformations parameters that best match the smoothed data.

Finally, the mean pattern $\f$ is estimated by the following smoothed Procrustes mean
\begin{equation}\label{eq.defflambda}
\hat \f\vphantom{\f}^{\lambda} = \frac{1}{J} \sum_{j=1}^{J} e^{-\hat a^{\lambda}_{j}} \left(\A^{\lambda}\Y_{j} - \1_k\otimes \hat b_{j}^{\lambda} \right)R_{-\hat \alpha_{j}^{\lambda}}.
\end{equation}

To analyze the convergence of the estimator $\hat \f\vphantom{\f}^{\lambda}$ to the mean pattern $\f$, let us introduce some regularity conditions on 
%the planar curve $f$ and on 
the covariance structure of the random variable $\bzeta$ in $\Sh$. Let $\tilde \bzeta = (\bzeta^{(1)},\bzeta^{(2)}) = (\zeta^{(1)}_1,\ldots,\zeta_k^{(1)},\zeta^{(2)}_1 ,\ldots,\zeta^{(2)}_k ) \in\R^{2k}$ be the vectorized version  of $\bzeta=(\zeta^{(1)}_\ell,\linebreak[1]\zeta^{(2)}_\ell)_{\ell=1}^k\in\Sh$. 
\begin{hyp}\label{ass.zeta}
 The random variable $\tilde \bzeta$ is a centered Gaussian vector  in $\R^{2k}$ with covariance matrix $\bSigma_k\in\R^{2k\times 2k}$. %, and such that $\bzeta^{(1)}$ and $\bzeta^{(2)}$ are i.d.d.\ vectors in $\RR^{k}$.
 Let $\gamma_{\max}(k)$ be the largest eigenvalue of $\bSigma_k$. Then,
$$
\lim_{k\to\infty} \gamma_{\max}(k) k^{-\frac{2s}{2s+1}} =0,
$$
where $s$ is the smoothness parameter  defined in Assumption \ref{ass.f}.
\end{hyp}
%Note that neither isotropy nor invariance conditions are required on the covariance structure of $\bzeta$. 
For example, the isotropic Gaussian error model corresponds to $\bSigma_k = Id_k$ and $\gamma_{\max}(k) = 1$ satisfies Assumption \ref{ass.zeta}. If there exists correlations terms (\ie non-zero off-diagonal entries in $\bSigma_k$), then the level of  the perturbation $\tilde \bzeta$  has to be sufficiently small. A simple model is the case where  $\bSigma_k = \left[S(\abs{\ell - \ell'})\right]_{\ell,\ell'=1}^k$ for some function $S:\R\longrightarrow \R$ satisfying $\int_\R \abs{S(t)}dt <+\infty$ implying that $\gamma_{\max}(k)\leq \sum_{\ell\in\Z} |S(\ell)|$ (see Lemma 4.11 in \cite{Toep}) and thus Assumption \ref{ass.zeta} is again satisfied. %Gaussian vectors $\bzeta$ that satisfy Assumption \ref{ass.zeta} are given . 
The following theorem is the main result of the paper.

\begin{theorem} \label{theo:main}
Consider model \eqref{eq:perturbmodel}  and suppose that Assumptions \ref{ass.f} and \ref{ass.zeta} hold. Suppose also that the random variables $(\Aa^*,\balpha^*,\b^*)$ are bounded  and belong to $[-\tfrac{A}{2},\tfrac{A}{2}]^J\times [-\tfrac{\AA}{2},\tfrac{\AA}{2}]^J \times [-B,B]^{2J}$ for some $0<A,B$ and $0<\AA<\frac{\pi}{4}$ and that $ \lambda(k)= k^{\frac{1}{2s+1}} $. 
\begin{itemize}
	\item 	For any $J\geq 2$ there exists $\f_{\bTheta_0} = e^{a_0} \f R_{\alpha_0} + \1_k \otimes b_0$ for some $(a_0,\alpha_0,b_0)\in\R\times [-\pi,\pi[ \times \R^2$ and a function $V_1(k,x)$ such that for any $x>0$,
			\begin{align} \label{eq.thmain1}
				\P\left(  \frac{1}{k}\snorm{\hat \f\vphantom{\f}^\lambda -  \f_{\bTheta_0}}^2_{\Sh} \geq V_1(k,x) \right) \leq e^{-x}, %\text{ and } V(J,k,x)  %\ &\xrightarrow[\phantom{J}k \to \infty\phantom{,}]{} 0, \qquad \text{ in probability, } 
			\end{align}
			with $V_1(k,x) \to 0$ when $k\to\infty$ and $x$ remains fixed.
		\item Suppose, in addition, that the random variables $(\Aa^*,\balpha^*,\b^*)$ have zero expectation  in $[-\tfrac{A}{2},\tfrac{A}{2}]^J\times [-\tfrac{\AA}{2},\tfrac{\AA}{2}]^J \times [-B,B]^{2J}$ with $A,\AA<0.1$. Then, there exists  a function $V_2(J,k,x)$ such that for any $x>0$,
			\begin{align}\label{eq.thmain2}
				\P \left(\frac{1}{k}\snorm{\hat \f\vphantom{\f}^\lambda-\f}^2_{\Sh}\geq V_2(J,k,x) \right) \leq e^{-x}, %\xrightarrow[k,J \to \infty]{} 0,\qquad \text{ in probability.} 
			\end{align}
			where $V_2(J,k,x) \to 0$ when $\min\{J,k\} \to \infty$ and $x$ remains fixed. 
	\end{itemize}
%
%Suppose that the random variables $(\Aa^*,\balpha^*,\b^*)$ are bounded  and belong to $[-A,A]^J\times [-\AA,\AA]^J \times [-B,B]^{2J}$ for some $0<A,B$ and $0<\AA<\frac{\pi}{4}$. Suppose that Assumptions \ref{ass.f} and \ref{ass.zeta} hold. If $ \lambda = \lambda(k)= \lfloor k^{\frac{1}{2s+1}} \rfloor  $, then for any $J\geq 2$ there exists $(a_0,\alpha_0,b_0)\in[-A,A]\times [-\AA,\AA] \times [-B,B]$ such that 
%\begin{align}
%\frac{1}{k}\snorm{\hat \f\vphantom{\f}^\lambda -  \f_{\bTheta_0}}^2_{\Sh} &\xrightarrow[\phantom{J}k \to \infty\phantom{,}]{} 0, \qquad \text{ in probability, } \label{eq.thmain1} 
%\end{align}
%where  $\f_{\bTheta_0} = e^{a_0} \f R_{\alpha_0} + \1_k \otimes b_0$. Suppose, in addition, that the random variables $(\Aa^*,\balpha^*,\b^*)$ have zero expectation  in $[-A,A]^J\times [-\AA,\AA]^J \times [-B,B]^{2J}$ with $A,\AA<0.1$. Then,
%\begin{align}
%\frac{1}{k}\snorm{\hat \f\vphantom{\f}^\lambda-\f}^2_{\Sh} \xrightarrow[k,J \to \infty]{} 0,\qquad \text{ in probability.} \label{eq.thmain2}
%\end{align}
\end{theorem}

Statement \eqref{eq.thmain1} means that, under mild assumptions on the covariance structure of the error terms $\bzeta_j$, it is possible to consistently estimate the shape of the mean curve $f$ when the number of observations $J$ is fixed and the number $k$ of discretization points  increases. Note  that $(a_0,\alpha_0,b_0)$  depends on $J$ and is given by formula \eqref{eq.g0} in Section \ref{part.ident}. The function $V_1(k,x)$ is explicitly given  in Section \ref{part.consfunc}. To obtain statement \eqref{eq.thmain2}, we assume the condition $A,\AA<0.1$ which  means  that the random scaling and rotations in model \eqref{eq:perturbmodel} are not too large. Also, it is assumed that random scaling, rotations and translations have zero expectation, meaning that the deformations parameters in model \eqref{eq:perturbmodel} are centered around the identity. Then, under such assumptions, statement \eqref{eq.thmain2} shows that one can consistently estimate the true mean curve $f$ when both the sample size $J$ and the number of landmarks $k$ go to infinity. Again, the function $V_2(J,k,x)$ is explicitly given  in Section \ref{part.consfunc}. These results are consistent with those obtained in \cite{BC11}, where we have studied the consistency of Fréchet means in deformable models for signal and image processing.

\subsection{Organization of the paper}

In Section \ref{part.defgroup}, we recall some properties on the similarity group of the plane, and we describe its action on the mean pattern $\f$. Then, we discuss General Procrustes Analysis (GPA) and we compare it to our approach. In Section  \ref{part.ident} we discuss some identifiability issues in model \eqref{eq:perturbmodel}. The estimating procedure is described in detail in Section \ref{part.EstimProc}. Consistency results are given in Section \ref{part.consistency}. Some  experiments in Section \ref{part.num} illustrate the numerical performances of our approach. All the proofs are gathered in a technical appendix.

\section{Group structure and Generalized Procrustes Analysis}
 \label{part.defgroup}

\subsection{The similarity group}
 \label{part.SIM}

 \paragraph{Group action} First let us introduce some  notations and definitions that will be useful throughout the paper. The similarity group of the plane is the group $(\G,.)$ generated by isotropic scaling, rotations and translations.  The identity element in $\G$ is denoted $e$ and the inverse of $g\in\G$ is denoted by $g^{-1}$. We parametrize the group $\G$ by a scaling parameter $a \in \R$, an angle $\alpha \in [-\pi,\pi[$ and a translation $b \in \R^{2}$, and  we make no difference between $g\in\G$  and its parametrization $(a,\alpha,b)\in\R \times [-\pi,\pi[ \times \R^2$. For all $g_1 = (a_1,\alpha_1,b_1),g_2 = (a_2,\alpha_2,b_2) \in \R\times [-\pi,\pi[ \times \R^2 $ we have
\begin{align}
g_1.g_2  = (a_1,\alpha_1,b_1).(a_2,\alpha_2,b_2) & = (a_1+a_2,\alpha_1+\alpha_2,e^{a_1} b_2 R_{\alpha_1} +b_1), \nonumber\\
g_1^{-1}  = (a_1,\alpha_1,b_1)^{-1} & =   (-a_1,-\alpha_1, - e^{-a_1} b_1 R_{-\alpha_1} ), \label{eq.loicomp}\\
 e & =  (0,0,0).\nonumber
\end{align}
The action of $\G$ onto $\R^{k\times 2}$ is given by the mapping $(g,\x) \longmapsto  g.\x := e^{a} \x R_{\alpha} + \1_{k} \otimes b,$ for  $g = (a,\alpha,b) \in \G$ and $\x \in \R^{k \times 2}$.  Note that we use the same symbol ``$.$'' for the composition law of $\G$ and its action on $\R^{k\times 2}$. This action can also be defined on $L^2([0,1],\R^2)$ by replacing $\f$ by $f\in L^2([0,1],\R^2)$. Coming back to the final dimensional case,  let 
$$
\one_{k\times 2} = \1_{k} \otimes \R^2
$$
 be the two dimensional linear subspace of $\R^{k \times 2}$ consisting of degenerated configurations, \ie configurations composed of $k$ times the same landmarks.  The orthogonal subspace $\one^{\perp}_{k\times 2}$ is the set of centered configurations. We have the orthogonal decomposition
$
\R^{k\times 2} = \one^\perp_{k\times 2} \oplus \one^{\phantom{\perp}}_{k\times 2},
$
and for any configuration $\x\in\Sh$ we write $\x = \x_0 + \bar\x \in \one^\perp_{k\times 2} \oplus \one^{\phantom{\perp}}_{k\times 2}$. We call $\x_0$ the centered configuration of $\x$ and $\bar \x = \1_k\otimes\Big(\frac{1}{k}\sum_{\ell=1}^k x^{(1)}_\ell, \frac{1}{k}\sum_{\ell=1}^k x^{(2)}_\ell\Big)$ the degenerated configuration associated to $\x$, see Figure \ref{fig.orbitG} for an illustration.

\paragraph{Orbit, stabilizer and section}%\begin{definition}
Given a configuration $\x$ in $\R^{k\times 2}$, the \emph{orbit} of $\x$ is defined as the set  
$$
\G.\x := \{g.\x, \ g\in \G\} \subset \R^{k\times 2}.
$$
This set is also called the shape of $\x$.
%\end{definition} 
%Consider now a degenerated configuration $\x\in \one_{k\times 2}$. Its
The orbit of any degenerated configuration is the entire subspace $\one_{k\times 2}$.  Note also that the linear subspace $\one_{k\times 2}$ is stable by the action of $\G$, and that the action of $\G$ on $\one_{k\times 2}$ is not free, meaning that for any $\bar\x\in\one_{k\times 2}$ the equality $g_1.\bar\x = g_2.\bar\x$ does not imply that $g_1 = g_2$. Now, if $\x \in \Sh \setminus \one_{k\times 2}$ is a non-degenerated configuration of $k$ landmarks, its orbit $\G.\x$ is a sub-manifold  of  $\R^{k\times 2}\setminus\one_{k\times 2}$ of dimension $\dim(\G)=4$.  

%\begin{definition}
Given a configuration $\x\in\R^{k\times2}$, the \emph{stabilizer} $I(\x)$ is the closed subgroup of $\G$ which leaves $\x$ invariant, namely
$$
I(\x) = \{g\in \G :  \ g.\x=\x \}. 
$$
%\end{definition}
If $\bar\x = \1_k \otimes (\bar x^{(1)}, \bar x^{(2)}) \in\one_{k\times 2}$ is a degenerated configuration, its stabilizer is non trivial and is equal to $I(\bar\x) =\{ (a, \alpha, (\bar x^{(1)},\bar x^{(2)}) - e^{a}(\bar x^{(1)}, \bar x^{(2)}) R_{\alpha}),\ a\in\R,\ \alpha\in[-\pi,\pi[ \}$.  If $\x\in \Sh \setminus \one_{k\times 2}$ is a non-degenerated configuration, its stabilizer $I(\x)$ is trivial, \ie is reduced to the identity $\{e\}$. The action of $\G$ is said  free if the stabilizer of any point is reduced to the identity. Hence, the action of $\G$ is free on the set of non-degenerated configurations of $k$-ads in $\R^2$.

	A \emph{section} of the orbits of $\G$ is a subset of $\Sh$ containing a unique element of each orbit.
%Two well-known examples of sections for the similarity group acting on $\Sh$ are the so-called Bookstein's and  Kendall's coordinates (see e.g.\ \cite{MR1646114} for a precise definition).
A well-known example of section for the similarity group acting on $\Sh\setminus \one_{k\times 2}$ is the so-called Bookstein's coordinates system (see e.g.\ \cite{MR1646114} p. 27).

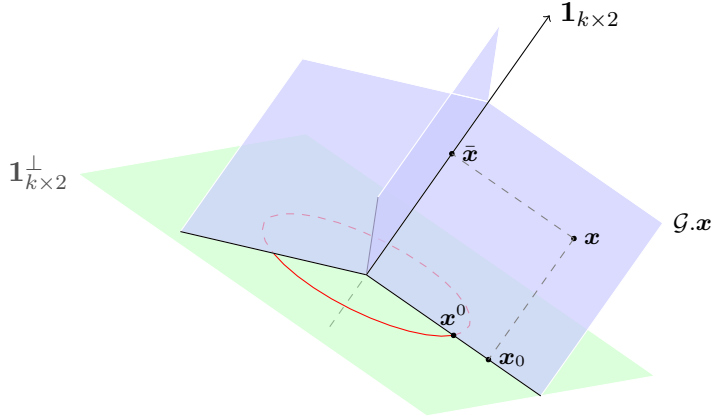
\begin{figure}
\begin{center}
\begin{tikzpicture}[scale=2.8,rotate=-35]
  
  \draw[dashed] (0,-.3,0) --(0,0,0);
  
  \coordinate (a) at (-1, 0,1 ); 
  \coordinate (b) at(1,0,1) ;
  \coordinate (c) at (1,0,-1);
  \coordinate (d) at (-1,0,-1);
  \draw[draw=white,fill=green!20,fill opacity=.7]  (a)node[left]{$\one_{k\times 2}^\perp$} -- (b)-- (c)-- (d)--cycle;
  
  \coordinate (a) at (-.5,1,0.866 ); 
  \coordinate (b) at(-.5,0,0.866 ) ;
  %\coordinate (c) at (-.5*1,1,-0.866*1 );
  \coordinate (c) at (-.5,1,-0.866 );
  \coordinate (d) at (-.5,0,-0.866 );

  \draw [dashed,red]   (  0.5000     ,     0  ,        0)--
  (  0.4879    ,      0 ,   -0.1093)--
  (  0.4522    ,      0 ,   -0.2134)--
  (  0.3946    ,      0 ,   -0.3071)--
  (  0.3179    ,      0  ,  -0.3860)--
  (  0.2258    ,      0 ,   -0.4461)--
  (  0.1227    ,      0 ,   -0.4847)--
  (  0.0138    ,      0 ,   -0.4998)--
  ( -0.0959    ,      0 ,   -0.4907)--
  ( -0.2008    ,      0 ,   -0.4579)--
  ( -0.2961    ,      0 ,   -0.4029)--
  ( -0.3771    ,      0 ,   -0.3284)--
  ( -0.4397   ,       0  ,  -0.2380)--
  ( -0.4811    ,      0  ,  -0.1361)--
  ( -0.4992   ,       0  ,  -0.0275)--
  ( -0.4932    ,      0 ,    0.0823)--
  ( -0.4632    ,      0  ,   0.1882)--
  ( -0.4109   ,       0  ,   0.2849)--
  ( -0.3386   ,       0  ,   0.3679)--
  ( -0.2500   ,       0  ,   0.4330);

\draw [red] (     0.5000     ,    0    ,     0)--
  (  0.4910     ,    0  ,  0.0946)--
  (  0.4642    ,     0  ,  0.1858)--
  (  0.4206   ,      0  ,  0.2703)--
  (  0.3619  ,       0  ,  0.3450)--
  (  0.2900        , 0  ,  0.4073)--
  (  0.2077       ,  0  ,  0.4548)--
  (  0.1179      ,   0  ,  0.4859)--
  (  0.0238     ,    0  ,  0.4994)--
  ( -0.0712    ,     0   , 0.4949)--
  ( -0.1635   ,      0  ,  0.4725)--
   (-0.2500  ,       0   , 0.4330);
 %\draw[red!80] (.3,.5) arc  ( 0:240:.3 and .2);
   
   \draw[draw=white,fill=blue!20,fill opacity=.7] (0,0)--(0,1)--(c)--(d)--cycle;\draw[thin]  (0,0,0)--(d);
   \draw[draw=white,fill=blue!20,fill opacity=.7] (0,0)--(0,1)--(a)--(b)--cycle; \draw[thin]  (0,0,0)--(b);
   \draw[draw=white,fill=blue!20,fill opacity=.7] (0,1)--(1,1)--(1,0)--(0,0)--cycle; \draw[thin](0,0,0)--(1,0,0);

\path (1,1) node[right]{{\footnotesize $\G.\x$}}; 
 
 \draw[white,thin] (c) -- (d);
  \draw[white,thin] (0,1,0) --(1,1,0);
  \draw[white,thin] (0,1,0) --(-.5,1,0.866 ); 
  \draw[white,thin] (0,1,0) --(c);
  \draw[->] (0,0,0) --(0,1.5,0) node[right]{$\one_{k\times 2}$};   

%   \draw[red!80] (.3,0) arc  ( 0:-130:.35 and .2);
%   \draw[red!80,dashed] (.3,.5) arc  ( 0:240:.3 and .2);

\draw[fill=black] (.7,.7 , 0 ) node[right]{{\footnotesize $\x$}} circle (.3pt);
\draw[fill=black] (0,.7 , 0 ) node[right]{{\footnotesize $\bar\x$}} circle (.3pt);
\draw[fill=black] (.7,0, 0 ) node[right]{{\footnotesize $\x_0$}} circle (.3pt);
\draw[fill=black] (.5,0, 0 ) node[above]{{\footnotesize $\x^0$}} circle (.3pt);
\draw[gray,dashed] (.7,.7 , 0 ) --(.7,0 , 0 );
\draw[gray,dashed] (.7,.7 , 0 ) --(0,.7 , 0 );

\end{tikzpicture}
\end{center}
\caption{ Three orbits of the action of the similarity group $\G$ are represented in blue. The space of centered configurations $\one_{k\times 2}^\perp$ is the green plane. The preshape sphere $\S_2^k$ is the red circle. For a particular $\x\in\Sh$, the centered version is $\x_0$ and the centered and normalized version is $\x^0$. The degenerated configuration associated to $\x$ is $\bar \x$. }\label{fig.orbitG}
\end{figure}

\subsection{Kendall's shape space and Generalized Procrustes analysis}%distance between configurations}
\label{part:ShapeSpace}

\paragraph{Shape space}% and distances.}

Let $\x \in \R^{k\times2} \setminus \one_{k\times 2}$ be a non-degenerated  configuration.  Let $H =Id_k - \frac{1}{k}\1_k\1_k'$ be a centering matrix. The effect of translation can be eliminated by centering the configuration $\x$ using the matrix $H$  (see \cite{MR1646114} for other centering methods), while the effect of isotropic scaling is removed by projecting the centered configuration on a unit sphere, which yields to the so-called pre-shape $\x^0$ of  $\x$ defined as
$$
\x^0 = \frac{H \x}{\| H \x \|}_{\Sh} \in \R^{k \times 2}.
$$
Consider now the pre-shape sphere defined by  
$
\S^{k}_{2} := \left\{\x^0, \ \x \in \R^{k \times 2} \setminus \one_{k\times 2} \right\}
$ and see Figure \ref{fig.orbitG} for an illustration. Note that this normalization of the planar configurations amounts to choose a section for the action of the group generated by the translation and scaling in the plane.  The Kendall's shape space  is then defined as the quotient of $\S^{k}_{2}$ by the group $\mathbf{SO}(2)$ of rotations of the plane, namely 
$$
\bSigma_{2}^{k} := \S^{k}_{2} / \mathbf{SO}(2) = \left\{ [\x^{0}] : \x^{0} \in \S^{k}_{2} \right\} \mbox{ with } [\x^{0}] = \left\{ \x^{0} R_{\alpha}, \alpha \in [-\pi,\pi[ \right\}.
$$
The space $\bSigma_{2}^{k}$ can be endowed with a Riemannian structure and we refer to \cite{MR1891212} for a detail discussion on its geometric properties. 
%The partial Procrustes distance is defined on the pre-shape sphere $\S_{2}^{k}$ as
%$$
%d_P(\x^0,\y^0) =  \inf_{\alpha\in[-\pi,\pi[} \snorm{ \x^0 - \y^0 R_\alpha}_{\Sh}, \qquad \x^0,\y^0\in\S_{2}^{k}.
%$$

Let us briefly recall the definition of the so-called partial and full Procrustes distances on $\S^{k}_{2}$. The partial Procrustes distance is defined on the pre-shape sphere $\S_{2}^{k}$ as
$$
d_P^2(\x^0,\y^0) =  \inf_{\alpha\in[-\pi,\pi[} \snorm{ \x^0 - \y^0 R_\alpha}^2_{\Sh}, \qquad \x,\y\in\S_{2}^{k}.
$$
Hence, it is the (Euclidean) distance between the orbits $[\x^0] = \mathbf{SO}(2).\x^0$ and $ [y^0] = \mathbf{SO}(2).\y^0$ with $\x^0,\y^0\in\S_{2}^{k} $. Let now $\mathcal{H}$ be the group of transformations of the plane generated by scaling and rotations. The action of $h\in\mathcal{H}$ on the centered configuration $\x^0$ is defined as
$
h.\x^{0} := e^{a}  \x^{0} R_{\alpha}
$ where $h =(a,\alpha)\in\R\times [-\pi,\pi[$. The Full Procrustes distance is then defined as
$$
d_F^2(\x^0,\y^0) =  \inf_{h\in\mathcal H} \snorm{ \x^0 - h.\y^0 }^2_{\Sh}, \qquad \x,\y\in\S_{2}^{k}.
$$

\paragraph{Generalized Procrustes analysis} 
The full Procrustes sample mean $\hat{\Y}_{F}$ of $\Y_{1},\ldots,\Y_{J}$ (see e.g.\  \cite{MR1108330,MR1646114}) is defined by 
\begin{equation}\label{eq:Procmean}
	\hat{\Y}_{F} = \argmin_{\x^0\in\S_{2}^{k}} \sum_{j=1}^J d_F^2(\Y^0_j,\x^0).
\end{equation}
The partial Procrustes mean $\hat \Y_P$ is defined in the same way by replacing $d_F$ by $d_P$ in \eqref{eq:Procmean}. Thence, this two Procrustes means are Fréchet mean either on $(\S^k_2,d_F)$  or $(\S^k_2,d_P)$ endowed with the empirical measure $\mu_J = \sum_{j=1}^J \delta_{\Y_j}$.

In practice, there are several way to compute the full Procrustes mean. In \cite{MR1175661} the author used complex coordinates and expressed the full Procrustes mean as the biggest eigenvalue of a symmetrical positive definite complex matrix, see \cite{MR1646114} result 3.2 and \cite{bhat}. The full Procrustes mean $\hat\Y_F$ can also be approximated by General Procrustes procedure which amounts to use the following identity
\begin{equation}
\hat{\Y}_{F} = \frac{1}{J} \sum_{j=1}^{J} \hat{h}_{j} \cdot \Y_{j}^{0} \label{def:fullmean}
\end{equation}
where $\hat{h}_{1},\ldots,\hat{h}_{J}$ are the  argmins of the functional %$\in \argmin_{(h_{1},\ldots,h_{J}) \in \mathcal H^{J}} 
$M(h_1,\ldots,h_J) = \frac{1}{J} \sum_{j=1}^{J}   \big\| h_{j}.\Y_{j}^{0} - \frac{1}{J} \sum_{j'=1}^{J}   h_{j'} . \Y_{j'}^{0} \big\|^{2}_{\Sh}$ subject to the constraint $\| \frac{1}{J} \sum_{j=1}^{J}  h_{j}.\Y_{j}^{0} \|_{\Sh}^{2} = 1$. The configurations $\hat h.\Y_j$'s are known as Full Procrustes fits and the $\hat h$ can be explicitly computed by using a singular value decomposition. %are given by the following Procrustes procedure
%\begin{equation}
% \begin{cases}
%(\hat{h}_{1},\ldots,\hat{h}_{J}) \in \argmin_{(h_{1},\ldots,h_{J}) \in \mathcal H^{J}}   \frac{1}{J} \sum_{j=1}^{J}   \left\| h_{j}.\Y_{j}^{0} - \frac{1}{J} \sum_{j'=1}^{J}   h_{j'} . \Y_{j'}^{0} \right\|^{2}_{\Sh} \\ %\label{eq:fullProcrustes}\\
% \mbox{subject to } \Big\| \frac{1}{J} \sum_{j=1}^{J}  h_{j}.\Y_{j}^{0} \Big\|_{\Sh}^{2} = 1. 
%\end{cases}
%\end{equation}
%This expression is very similar to \eqref{}. Although it can be explicitly solve (see in \cite{MR1646114}), the iterative Generalized Procrustes Algorithm allow fast numerical approxiamtion (see \cite{MR1646114} p. ).
%Obviously, the mean shape $\f$ in model \eqref{eq:perturbmodel}  does not necessarily belong to $\S^{k}_{2}$, and will generally not have the same orientation than $\hat{\Y}_{F}$. Therefore, using the Euclidean distance in $\R^{k \times 2}$ to compare  $\hat{\Y}_{F}$ and $\f$  is  not meaningful. 
%The matching criterion $M$ is invariant under the action of rotations, and to ensure uniqueness of the argmins, one can add a constraint similar to those given in part \ref{part.ident}. % meaning that $\hat{\Y}_{F} R_{\alpha}$ is a minimizer of  \eqref{eq:fullProcrustes} for any $\alpha \in [-\pi,\pi[$.
%Up to  uniqueness problems due to the invariance of $M$ with respect to rotations, this formulation is similar to the procedure given by equations \eqref{def.estimateurs} and \eqref{eq.defflambda}.
In practice, one can use the iterative General Procrustes algorithm to compute  $\hat\Y$, see \cite{MR1646114} pages 90-91. 

\subsection{Discussion on the double asymptotic setting and comparison with GPA}

%In the Statistical Shape Analysis . 
\paragraph{Asymptotic settings} When random planar curves (such as digits or letters for instance) are observed, a natural framework for statistical inference is an asymptotic setting in the number $J$ of curves. This setting means that increasing the number of curves  at hand should help to compute a more accurate empirical mean curve. Unfortunately, consistency results of Procrustes type procedures are reduced to the very specific case of an isotropic perturbation. In this paper, we show that increasing the number $k$ of discretization points will ensure a consistent estimation of a mean shape in more general cases. 

Consider model \eqref{eq:perturbmodel} where $k>2$ is fixed and %the covariance matrix of 
the random perturbation $\bzeta$ is isotropic, (see Proposition 1 in \cite{MR1618880} for a precise definition of isotropy for random variables belonging to $\R^{k \times 2}$). In this framework, it has been proved in \cite{MR1436569} that the functional $M_J(\x^0) = \frac{1}{J} \sum_{j=1}^J d_F^2(\Y_j,\x^0)$ defined on $\Sh$ converge uniformly in probability to the functional $M(\x^0) = \E(d^2_F(\Y_1,\x^0))$ which admits a unique minimum at $[\f^{0}]$. These two facts imply that $[\hat{\Y}_{F}]$ converges almost surely to $[\f^{0}]$ as $J \to + \infty$. In \cite{MR1618880} the author used a Fréchet mean approach to show the consistency of $\hat\Y_F$ with a  slightly more general kind of noise. Finally, note that the estimator $\hat \Y_P$ defined above is also studied in \cite{MR1436569,MR1618880} with similar consistency results.

When the random perturbation $\bzeta$ in model \eqref{eq:perturbmodel} is non-isotropic, it has been argued in \cite{MR1436569} that the Procrustes estimator $\hat\Y_F$ can be arbitrarily inconsistent when the signal-to-noise ratio decrease. The heuristic presented by the authors suggests that the main phenomenon that prevent the Procrustes estimator to be consistent is the  fact %not the non-uniqueness of the minimizer of functional $M$ but the 
 that the functional $M$ do not attain its minimum at $[\f^0]$. In section 4.2 of \cite{HuckemannSJS} the author makes this remark clear as he gives an explicit example : given a mean pattern $\f^0$, the idea is to increase the level of noise $\bzeta$ until the argmin of the functional $M$ which was initially equals to $[\f_0]$ jumps abruptly to another point. This phenomenon seems to be linked to the geometry of the sphere and properties of the Fréchet mean.%, see \cite{HuckemannSJS,Charlier} for details. 

\paragraph{Comparison with GPA} Hence, it is commonly the belief that Procrustes sample means can be inconsistent  when considering convergence in $\bSigma_{2}^{k}$ and  the asymptotic setting $J \to + \infty$. Nevertheless, the above discussion suggests that a sufficient condition to ensure the consistency of Procrutres type estimators is to control the level of non-isotropic noise. That is why we introduced a pre-smoothing step that takes advantage of increasing the number $k$ of points composing each configuration in order to ensure more general consistency results. Therefore, our approach and GPA share some similarities. They are both based on the estimation of scaling, rotation and translation parameters by a Procrustean procedure which leads to the M-estimators \eqref{def.estimateurs} which is related to \eqref{eq:Procmean}. To compute a sample mean shape, this M-estimation step is then followed by a standard empirical mean in $\RR^{k \times 2}$ of the aligned  data using these estimated deformation parameters, see equations \eqref{eq.defflambda} and  \eqref{def:fullmean}.

However, one of the main differences between the approach developed in this paper and GPA is the choice of the normalization of the data. In GPA, the deformation parameters $\hat{h}_{1},\ldots,\hat{h}_{J} $ are computed so that the full Procrustes sample mean $\hat{\Y}_{F}$ belongs to  the  pre-shape sphere  $\S^{k}_{2}$, see the   constraint appearing in \eqref{eq:Procmean}. Therefore, the computation of $ \hat{h}_{1},\ldots,\hat{h}_{J} $ is somewhat independent of any assumption on the true parameters $(\Aa^*,\balpha^*,\b^*)$ in model \eqref{eq:perturbmodel}. In this paper, to ensure the well-posedness of the problem \eqref{def.estimateurs}, we chose to compute the estimator $( \hat\Aa^\lambda,\hat\balpha^\lambda,\hat\b\vphantom{\b}^\lambda)$ by minimizing the matching criterion \eqref{def.Mlambda} on the constrained set $\bTheta_0$. The choice of the constraints in $\bTheta_0$ is motivated by the hypothesis that the true deformation parameters $(\Aa^*,\balpha^*,\b^*)$   in \eqref{eq:perturbmodel} are centered around the identity. %Another main difference in our approach is the smoothing of the data before applying a Procrustean procedure.

\section{Identifiability conditions}
\label{part.ident}

Recall that in model \eqref{eq:perturbmodel}, the random deformations acting on the the mean pattern $\f$ are parametrized by a vector $(\Aa^*,\balpha^*,\b^*) = (a^*_1,\ldots,a^*_J,\alpha^*_1,\ldots,\alpha^*_J,b^*_1,\ldots,b^*_J)$ in $\R^J\times[-\pi,\pi[^J\times\R^{2J}$. 
\begin{hyp}\label{hyp.Theta}
Let $0<A,B$ and $0<\AA<\pi$  be three real numbers. The deformation parameters $(\Aa^*_j,\balpha^*_j,\b^*_j)$, are i.i.d random variables with zero expectation and and taking their values in 
$$
\Theta^*   = \left[-\frac{A}{2},\frac{A}{2}\right] \times \left[-\frac{\AA}{2},\frac{\AA}{2}\right] \times [-B,B]^{2}.
$$
\end{hyp}
Let
$
\bTheta^* = \left[-\frac{A}{2},\frac{A}{2}\right]^J\times \left[-\frac{\AA}{2},\frac{\AA}{2}\right]^J \times [-B,B]^{2J}.
$
Under Assumption \ref{hyp.Theta}, we have $(\Aa^*,\balpha^*,\b^*)\in\bTheta^*$. Note that the compactness of $\Theta^*$ (and thus of $\bTheta^*$) is  an essential condition to ensure the consistency of our procedure. Indeed, the estimation of the deformation parameters $(\Aa^*,\balpha^*,\b^*)$ and the mean pattern $\f$  is based on the minimization of the criterion \eqref{def.Mlambda}. If there were no restriction on the amplitude of the scaling parameter, the degenerate solution $a_j = -\infty$ for all $j=1,\ldots,J$ is always a minimizer of \eqref{def.Mlambda}. Therefore, the minimization has to be performed under additional compact constraints.

\subsection{The deterministic criterion $D$}

Let $(\Aa,\balpha,\b)\in\PPP$ and consider  the following criterion,
\begin{align}\label{eq.critD}
D (\Aa,\balpha,\b)
& = \frac{1}{Jk}\sum_{j=1}^J \bigg\| (g_{j}^{-1} .g_{j}^*).\f  -  \frac{1}{J}\sum_{j'=1}^J (g_{j'}^{-1} .g_{j'}^*).\f \bigg\|^2_{\Sh}.
%\\ &= \frac{1}{J}\sum_{j=1}^J \bigg\| e^{a^*_{j}-a_{j}}(\f +  \1_k\otimes (b^*_j-b_j))R_{\alpha_{j}^*-\alpha_{j}} -  \frac{1}{J}\sum_{j'=1}^J  e^{a_{j'}^*-a_{j'}}(\f + \1_k\otimes  (b_{j'}^*-b_{j'}))R_{\alpha_{j'}^*-\alpha_{j'}} \bigg\|^2.
\end{align}
where $g_j=(a_j,\alpha_j,b_j) $ and $g^*_j= (a_j^*,\alpha^*_j,b_j^*)$ for all $j = 1,\ldots,J$. The criterion $D$ is a version without noise of the criterion $M^{\lambda}$ defined at \eqref{def.Mlambda}. The estimation procedure described in Section \ref{part.Intro.1} is based on the convergence of the argmins of $M^{\lambda}$ toward the argmin of $D$ when $k$ %and $\lambda=\lambda(k)$
 goes to infinity. As a consequence, choosing identifiability conditions amounts  to fix a subset $\bTheta_0$ of $\PPP$ on which $D$ has a unique argmin. In the rest of this section, we determine the zeros of $D$, and then we fix a convenient constraint set $\bTheta_0$ that contains a unique point at which $D$ vanishes.
 
%%% tricks pour aligner les equivalents... 
  \newlength{\gnat}
  \settowidth{\gnat}{$\quad  g^{-1}_{j}. g_{j}^* = g^{-1}_{j'}. g_{j'}^*,$}
%%%   
The criterion $D$ clearly vanishes at $(\Aa^*,\balpha^*,\b^*)\in\PPP$. This minimum is not unique since easy algebra implies that 
\begin{alignat*}{1}
D(\Aa,\balpha,\b) = 0  \qquad \Longleftrightarrow &\qquad  (g_{j}^{-1} .g_{j}^*).\f =  (g_{j'}^{-1} .g_{j'}^*).\f,\quad \text{ for all } j,j' = 1,\ldots,J.
\intertext{Suppose now that  $\f\notin\one_{2\times k} $ is a non-degenerated planar configuration. In Section \ref{part.SIM}, we have seen that the action of $\G$ on $\f$ is free, that is, the stabilizer $I(\f)$ is reduced to the identity. Thus, we obtain, %can ``simplify by $\f$'' in the last formula and we obtain, % that is, when (see part ) then we obtain%it is possible to  and we get
}
D(\Aa,\balpha,\b)= 0\qquad  \Longleftrightarrow &\qquad  g^{-1}_{j}. g_{j}^* = g^{-1}_{j'}. g_{j'}^*, \hspace{-\gnat} \hphantom{\quad  (g_{j}^{-1} .g_{j}^*).\f =  (g_{j'}^{-1} .g_{j'}^*).\f,}\quad \text{ for all } j,j' = 1,\ldots,J  \\
\phantom{D(\Aa,\balpha,\b)= 0 }\qquad  \Longleftrightarrow  &\qquad  \begin{cases} a_{j}^{*}-a_{j} = a_{j'}^{*}-a_{j'} , \\  \alpha_{j}^{*} -\alpha_{j}=\alpha_{j'}^{*} -\alpha_{j'},\\ e^{-a_{j}} (b_{j}^* - b_{j} )R_{-\alpha_{j}} = e^{-a_{j'}} (b_{j'}^* - b_{j'} )R_{-\alpha_{j'}}  \end{cases}  \\  & \hphantom{\qquad  (g_{j}^{-1} .g_{j}^*).\f =  (g_{j'}^{-1} .g_{j'}^*).\f,}\quad \text{ for all } j,j' = 1,\ldots,J
\end{alignat*}
 We have proved the folowing result,
 \begin{lemma}\label{lemma.zeros}
Let $\f\in\R^{k\times 2}$ be a non-degenerated configuration of $k$-ads in the plane, \ie $\f \notin \one_{2\times k}$. Then, $D(\Aa,\balpha,\b)=0$  if and only if $(\Aa,\balpha,\b)$ belongs to the set %$(\Aa^*,\balpha^*,\b^*) * \Theta$ given by
\begin{align*}
(\Aa^*,\balpha^*,\b^*) * \G 
= \Big\{ (\Aa^*,\balpha^*,\b^*)*(a_0,\alpha_0,b_0),\ (a_0,\alpha_0,b_0)\in \R\times[-\pi,\pi[\times\R^2\Big\},
%\btheta^* * \Theta = \{ (\btheta_1^*.\btheta_0, \ldots, \btheta^*_J.\btheta_0 ),\ \btheta_0 \in \Theta \}
%g^* * SIM  = \btheta^* * \Theta =  \{ (\btheta_1^*.\btheta_0, \ldots, \btheta^*_J.\btheta_0 ),\ \btheta_0 \in\Theta \}.
\end{align*}
where $(\Aa^*,\balpha^*,\b^*)* (a_0,\alpha_0,b_0) = (a_{1}^{*}+a_{0},\ldots,a_{J}^{*}+a_{0} ,\alpha_{1}^{*} +\alpha_{0},\ldots,\alpha_{J}^{*} +\alpha_{0} ,e^{a_{1}^*}b_{0}R_{\alpha_{1}^*} +b_1^* ,\ldots ,e^{a_{J}^*}b_{0}R_{\alpha_{J}^*} +b_J^*  ) \in\PPP$.
%is the subset of $\Theta^J$ where $D$ vanishes. 
\end{lemma}
\begin{remark} \label{rem.orbit}
Lemma \ref{lemma.zeros} is simpler than it appears. By reordering the entries of the vector  $(\Aa^*,\balpha^*,\b^*)$ there is an obvious correspondence between $(\Aa^*,\balpha^*,\b^*)\in\bTheta^*$ and $(g_1^*,\ldots,g_J^*)\in\G^J$ \textit{via} the parametri\-{}zation of the similarity group defined in Section \ref{part.SIM}. Hence, Lemma \ref{lemma.zeros} tells us that the criterion $D$ vanishes for all the vectors $(\Aa,\balpha,\b)\in\PPP$ corresponding to   the subset of the group $\G^J$ given by 
$$
(g_1^*,\ldots,g_J^*) * \G = \{ (g_1^*.g_0, \ldots, g^*_J.g_0 ),\ g_0 \in \G \} \subset \G^J.
$$
The ``$*$'' notation is nothing else than the right composition by a same $g_0\in\G$ of all the entries of a $(g_1,\ldots,g_J)\in\G^J$. Hence the subset $(g_1^*,\ldots,g_J^*) * \G$ can be interpreted as the orbit of $(g_1^*,\ldots,g_J^*)\in\G^J$ under the (right) action of $\G$. Indeed, $\G$ acts naturally by (right) composition on the all the coordinates of an element of $\G^J$. 
\end{remark}

\subsection[The constraint set]{The constraint set $\bTheta_0$}

By Lemma \ref{lemma.zeros}, the set $\bTheta_0$ must intersect at a unique point,  say $(\Aa^*_{\bTheta_0},\balpha^*_{\bTheta_0},\b^*_{\bTheta_0})$, each set $(\Aa^*,\balpha^*,\b^*) * \G$.  It is convenient to choose $\bTheta_0$ to be of the form $\bTheta_0 = \PPP \cap \LL_0$ where $\LL_0$ is a linear space of  $\R^{4J}$. The linear space  $\LL_0$ must be chosen so that for any $(\Aa^*,\balpha^*,\b^*)$ in $\bTheta^*$, there exists a unique point $(\Aa^*_{\bTheta_0},\balpha^*_{\bTheta_0},\b^*_{\bTheta_0})$ in $\bTheta_0$ that can be written as   $(\Aa^*_{\bTheta_0},\balpha^*_{\bTheta_0},\b^*_{\bTheta_0}) = (\Aa^*,\balpha^*,\b^*)* (a_0,\alpha_0,b_0)$  for some $(a_0,\alpha_0,b_0)\in\R\times[-\pi,\pi[\times\R^2$. 

\begin{remark} %\label{rem.section}
As we have seen in Remark \ref{rem.orbit}, the set $(\Aa^*,\balpha^*,\b^*) * \G$ can be interpreted as an orbit of the action of $\G$ on $\G^J$. In this terminology, the set $\bTheta_0$ can be viewed as a  section of the orbits. Indeed, the section is the set of  representatives $(\Aa^*_{\bTheta_0},\balpha^*_{\bTheta_0},\b^*_{\bTheta_0})$ of each orbit. See Figure \ref{fig:identif2} for an illustration.
\end{remark}

Let us consider a choice of $\bTheta_0$ motivated by the fact that, under Assumption \ref{hyp.Theta}, the random deformation parameters have zero expectation. In this setting, it is natural to impose that  the estimated deformation parameters sum up to zero by choosing $\LL_0 = \one_{4J}^\perp$, which is the orthogonal of the linear  space $\one_{4J}=\1_{4J}.\R \subset\R^{4J}$. Such a choice leads to the set $ \bTheta_0$ defined equation \eqref{eq:Theta0} i.e.\
$$
\bTheta_0 = \{(\Aa,\balpha,\b)\in\Theta^J, \  (a_1+ \ldots +a_J,\alpha_1+ \ldots + \alpha_J,b_1+ \ldots +b_J)=0 \}.
$$
Now, let us show that for any $(\Aa^*,\balpha^*,\b^*)\in\bTheta^*$ there exists a unique $(\Aa^*_{\bTheta_0},\balpha^*_{\bTheta_0},\b^*_{\bTheta_0}) =(\Aa^*,\balpha^*,\b^*)* (a_0,\alpha_0,b_0)\in \bTheta_0$ for some $(a_0,\alpha_0,b_0)\in\R\times[-\pi,\pi[\times\R^2$. This amounts to solve the following equations 
\begin{equation}\label{eq.centre}
\left\{ \begin{array}{cccccc}
a_1^*+a_0 &+ &\ldots &+&a_J^*+a_0 &=0,\\
\alpha_1^*+\alpha_0 & +& \ldots &+& \alpha_J^*+\alpha_0 &=0,\\
e^{a_{1}^*}b_0R_{\alpha_{1}^*} + b_{1}^* &+&\ldots &+&e^{a_{J}^*}b_0 R_{\alpha_{J}^*}+ b_{J}^*   &= 0.
\end{array} \right.
%\btheta^*_1.\btheta_0 + \ldots +  \btheta^*_J.\btheta_0 = 0.
\end{equation}
After some computations, we obtain that equations \eqref{eq.centre} are satisfied if and only if
\begin{equation}\label{eq.g0}
(a_0,\alpha_0,b_0) = (- \bar\Aa^*, -\bar\balpha^*, - \bar\b^* (\overline{e^{\Aa^*} R_{\balpha^*}} )^{-1} ),
\end{equation}
where $\bar\Aa^* =   \frac{1}{J}\sum_{j=1}^J a_j^*\in\R$, $\bar\balpha^* =   \frac{1}{J}\sum_{j=1}^J \alpha_j^*\in\R$, $\bar\b^* =  \frac{1}{J}\sum_{j=1}^J b_j^*\in\R^2$ and $\overline{e^{\Aa^*} R_{\balpha^*}} = \frac{1}{J}\sum_{j=1}^J e^{a^*_j}R_{\alpha_j^*}$ is a $2\times 2$ invertible matrix.  Therefore, $(\Aa^*_{\bTheta_0},\balpha^*_{\bTheta_0},\b^*_{\bTheta_0})$ is uniquely given by
$$
([\Aa^*_{\bTheta_0}]_j,[\balpha^*_{\bTheta_0}]_j,[\b^*_{\bTheta_0}]_j) = ( a^*_j - \bar \Aa^*, \alpha^*_j -\bar \balpha^*, b^*_j -e^{a_j^*} (\bar \b^* (\overline{e^{\Aa^*} R_{\balpha^*}} )^{-1} )R_{\alpha^*_j}),
$$
for $j=1,\ldots,J$.

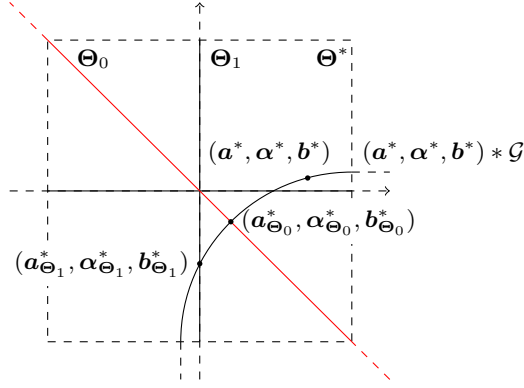
\begin{figure}[]
\begin{center}
\begin{tikzpicture}[scale=0.5]
\draw[dashed] (-4,-4) rectangle (4,4);
\draw[->,dashed] (-5,0) -- (5,0);
\draw[-] (-4,0) -- (4,0);
\draw[->,dashed] (0,-5) -- (0,5);
\draw[-] (0,-4) -- (0,4);
\path (0,3.5) node[right]{{\footnotesize $\bTheta_1$}};

\path (3.5,3.5) node[]{{\footnotesize $\bTheta_{\vphantom{1}}^*$}};

\fill[color=black] (2.839,0.3425)circle (2pt);
\path (1.839,1) node[]{{\footnotesize $(\Aa^*,\balpha^*,\b^*)$}};
 \path  (0,-1.93) node[left,fill=white]{{\footnotesize $(\Aa^*_{\bTheta_1},\balpha_{\bTheta_1}^*,\b_{\bTheta_1}^*)$}};
\path (.818,-.818) node[right,fill=white]{{\footnotesize $(\Aa^*_{\bTheta_0},\balpha_{\bTheta_0}^*,\b_{\bTheta_0}^*)$}};

\draw[red] (4,-4) --(-4,4);
\path (-3.5,3.5) node[right,black]{{\footnotesize $\bTheta_0$}}; 
\draw[red,dashed] (5,-5) -- (4,-4);
\draw[red,dashed] (-4,4) -- (-5,5);
\draw[dashed] (4,0.5) -- (5,.5);
\draw[dashed] (-.5,-4) -- (-.5,-5);

\draw[color=black] (4,0.5) arc (90:180:4.5);
\path (4,1)node[right]{{\footnotesize $(\Aa^*,\balpha^*,\b^*) * \G$}} ;

\fill[color=black] (0,-1.93) circle (2pt);
\fill (.818,-.818)circle (2pt);
\end{tikzpicture}
\end{center}
\caption{Choice of identifiability conditions when $J=2$.  }\label{fig:identif2}
\end{figure}
\begin{remark} Another possible approach is to fix, say the first observation as a reference, meaning that the criterion $D$ could  be optimized on the following subspace of $\PPP$ 
$$
\bTheta_{1}= \{(\Aa,\balpha,\b)\in [-A,A]^J\times [-\AA,\AA]^{2J}\times \R^2, \  (a_1,\alpha_1,b_1) = 0\}. %,\ldots, \btheta_J\in \Theta \}.
$$
%Here, it is easy to see that the minimum of $D$ on $\bTheta_{1}$ is 0. Indeed,  
%The  linear subset $\bTheta_{1}$ is the pre image by $\phi$ of $\mathfrak{S}_1 = \{ (e,g_2,\ldots,g_J), \ g_2,\ldots,g_J\in SIM \} \subset SIM^J$
With such a choice, for any $(\Aa^*,\balpha^*,\b^*)\in \Theta^J$, the $j$-th coordinate of  $(\Aa^*_{\bTheta_1},\balpha^*_{\bTheta_1},\b^*_{\bTheta_1}) = (\Aa^*,\balpha^*,\b^*)* (a_0,\alpha_0,b_0)$ is given by 
$$
([\Aa^*_{\bTheta_1}]_j,[\balpha^*_{\bTheta_1}]_j,[\b^*_{\bTheta_1}]_j)= (a_j^* - a^*_1, \alpha_j^* - \alpha_1^* , b_j^* -e^{a_j^*- a^*_1}b^*_1 R_{\alpha_j^*-\alpha_1^*}),
$$
where $(a_0,\alpha_0,b_0) = (a_1^*,\alpha_1^*,b_1^*)^{-1} = ( -a_1^*, -\alpha_1^*, -e^{-a_1^*}b_1^*R_{-\alpha_1^*}) $. A graphical illustration of the choice of identifiability conditions for $J=2$ is given in Figure \ref{fig:identif2}.
%If we are only interested in recovering the the equivalent class of $f^*$, the choice of the representative we will reconstruct do not matter.
%This method  works well if we do not have any information on the distribution of the deformation or if the choice of the representative in the orbit of $\f^*$ do not matter. Indeed, it remains to take an arbitrary observation as reference.
\end{remark}

\section{The estimating procedure}
\label{part.EstimProc}
\subsection{A dimension reduction step}
\label{part:RedDim}

We use Fourier filtering to project the data into a low-dimensional space as follows. Assume for convenience that $k$ is odd. For $\x = (\x_{1}',\ldots,\x_{k}')' \in \R^{k\times 2}$   and  $m= - \frac{k-1}{2},\ldots,\frac{k-1}{2}$, let 
\begin{align*}
c_m(\x)  =   \sum_{\ell=1}^{k} \x_{\ell}  e^{- i 2\pi m \tfrac{\ell}{k} }
=  \left( \sum_{\ell=1}^{k} x^{(1)}_{\ell} e^{- i 2\pi m \tfrac{\ell}{k} },  \sum_{\ell=1}^{k} x^{(2)}_{\ell} e^{- i 2\pi m \tfrac{\ell}{k} } \right) \in \C^2 \mbox{ with } \x_\ell = (x_\ell^{(1)}, x^{(2)}_{\ell}),
\end{align*}
be the $m$-th (discrete) Fourier coefficient of $\x$. Let $\lambda \in \{1,\ldots, \frac{k-1}{2}\}$ be a smoothing parameter, and define for each $\Y_{j}$ the smoothed shapes
$$
\hat \f\vphantom{\f}^{\lambda}_{j} =  \bigg( \frac{1}{k} \sum_{0 \leq |m| \leq \lambda}  c_m(\Y_{j}) e^{i2\pi m \frac{\ell}{k} } \bigg)_{\ell=1}^{k} = \A^{\lambda} \Y_{j} \in \R^{k \times 2}.  
$$
 In Section \ref{part.SIM}, we have shown that the similarity group is not free on the subset $\one_{k\times2}$ of degenerated configurations composed of $k$ identical landmarks, see Section \ref{part.SIM}. That is why we are going to treat separately the subspace $\one_{k\times2}$ and $\one^{\perp}_{k\times2}$ by considering the matrices
\begin{equation}
\bar A = \frac{1}{k} \1_k \1_k' \qquad\mbox{ and } \qquad A^{\lambda}_0 = \bigg( \frac{1}{k} \sum_{0 < |m| \leq \lambda}   e^{i2\pi m\tfrac{\ell-\ell'}{k}} \bigg)_{\ell,\ell'=1}^{k}. \label{eq:Alambda}
\end{equation}
Remark that $\bar A$ is a projection matrix on the one dimensional sub-space $\bar \bV :=\1_k.\R= \{ c \1_k : c \in \R\}$ of $\R^{k}$.   The matrix $A^{\lambda}_0$ is a projection matrix in a (trigonometric) sub-space $\bV^{\lambda}_0$ of dimension $2 \lambda$. Note that it is included in the linear space $\bar\bV^{\perp} = \{ x \in \R^{k} : x'\1_k = 0 \}$. Hence, $\bV^{\lambda}_0\times\bV^{\lambda}_0 $ is a linear subspace of $\one_{k\times 2}^\perp$ which is the space of the centered configurations. We have,
\begin{equation}\label{eq:decomp}
\A^{\lambda} =A^{\lambda}_{0} + \bar A \qquad \text{ and } \qquad \BV^{\lambda} = \bV^{\lambda}_0 \oplus \bar \bV.
\end{equation}
Thus, we can write the smoothed shape $\hat \f^{\lambda}_{j}$ as 
$$
\hat \f^{\lambda}_{j} =\A^{\lambda}\Y_{j} =  A^{\lambda}_{0}\Y_{j} + \bar A \Y_{j} \in (\bar \bV\times \bar \bV) \oplus (\bV^{\lambda}_0\times \bV^{\lambda}_0)
$$
where $\bar\bV\times \bar \bV = \one_{k\times 2}$ and  $\bV^{\lambda}_0\times \bV^{\lambda}_0\subset  \one_{k\times 2}^{\perp}$. In other words,  $  A^{\lambda}_{0}\Y_{j}$ is the smoothed centered configuration associated to $\Y_j$ and $\bar A \Y_{j}$ is the degenerated configuration given by the Euclidean mean of the $k$ landmarks composing $\Y_j$. Finally, remark that the low pass filter and the action of similarity group commute, that is, we have for all $g\in\G$ and $\f\in\Sh$
$$
 g. (\A^\lambda \f) = e^{a} \A^\lambda \f R_{\alpha} + \1_k\otimes b =  \A^\lambda (e^{a} \f R_{\alpha}+ \1_k\otimes b) = \A^\lambda  (g. \f).
$$

\subsection{Estimation of the deformation parameters}

Recall that  the estimator $(\hat\Aa^{\lambda},\hat\balpha^{\lambda},\hat\b\vphantom{\b}^{\lambda})$ of $(\Aa^*,\balpha^*,\b^*)$ is defined by the optimization problem \eqref{def.estimateurs}. %as the argmin of the functional $M^{\lambda}$ given by formula \eqref{def.Mlambda}. 
Nevertheless, as it is suggested by the discussion of Sections \ref{part.SIM} and \ref{part:RedDim}, one can carry out the estimation process in two steps. First, we estimate the rotation and scaling parameters on the space $\bV^{\lambda}_0 \times \bV^{\lambda}_0 \subset \one_{k\times 2}^\perp$ of the centered configurations. We then use these estimators to estimate the translation parameters which act on $\bar\bV \times\bar\bV = \one_{k\times2}$. Note that this procedure is equivalent to the optimization problem \eqref{def.estimateurs} as shown by Lemma \ref{lemme.argmin} below.

\paragraph{Estimation of rotations and scaling.} Define 
$$
M^{\lambda}_0(\Aa,\balpha) = \frac{1}{Jk} \sum_{j=1}^J \bnorm{  e^{-a_{j}}A^{\lambda}_0 \Y_{j} R_{-\alpha_{j}}- \frac{1}{J} \sum_{j'=1}^J e^{-a_{j'}}A^{\lambda}_0 \Y_{j'}  R_{- \alpha_{j'}}}^2_{\Sh}.
$$
Let $\1_J^{\perp} = \{(a_1,\ldots,a_J)\in\R^J, a_1+\ldots,a_J=0 \} $ and consider the  space
$
\bTheta_{0}^{\Aa,\balpha}  =  ([-A,A]^J \cap\1_J^{\perp}) \times ([-\AA,\AA]^J \cap \1_J^{\perp}).
$
Then, estimators of the rotation and scaling parameters are given by
\begin{equation}\label{eq.estimRotScal}
(\hat \Aa^{\lambda},\hat \balpha^{\lambda}) \in\argmin_{(\Aa,\balpha)\in\bTheta_{0}^{\Aa,\balpha} } M^{\lambda}_0(\Aa,\balpha).
\end{equation}

\paragraph{Estimation of translations.} Now that we have computed  estimators of the rotation and scaling parameters, let us define the  criterion,
$$
\bar M(\Aa,\balpha,\b) =\frac{1}{Jk} \sum_{j=1}^J \bnorm{  e^{- a_{j}}\bar A (\Y_{j} - \1_k\otimes b_j  )R_{-\alpha_{j}}- \frac{1}{J} \sum_{j'=1}^J e^{- a_{j'}}\bar A (\Y_{j'} - \1_k\otimes b_{j'}  )  R_{- \alpha_{j'}}}^2_{\Sh}
$$
and the  space
$
\bTheta_{0}^{\b} = \{(b_1^{(1)},\ldots,b_J^{(1)},b_1^{(2)},\ldots,b^{(2)}_J)\in\R^{2J}, b_1^{(1)}+\ldots+b_J^{(1)} = b_1^{(2)}+\ldots+b_J^{(2)} = 0 \}.%\bTheta_{0}^{\b} =[ -B,B] \cap\1_J^{\perp}.
$
The estimator of the translation parameters is then given by,
\begin{equation}\label{eq.defhatb}
\hat\b\vphantom{\b}^\lambda = \argmin_{\b\in\bTheta_0^{\b}} \bar M (\hat \Aa^{\lambda},\hat \balpha^{\lambda} ,\b).
\end{equation}
We emphasis that the estimators of the translation parameters depend on the estimated rotation and scaling parameters. It is shown in the proof of Lemma \ref{lemme.argmin} below that we have an explicit expression of $\hat\b\vphantom{\b}^\lambda$ given by 
\begin{equation}\label{eq.blambda}
\hat\b\vphantom{\b}^\lambda = (\bar\Y_1-e^{a_1} \bar\Y (\overline{e^{\Aa} R_{\balpha}})^{-1} R_{\alpha_1},\ldots,\bar\Y_J-e^{a_J} \bar\Y (\overline{e^{\Aa} R_{\balpha}})^{-1} R_{\alpha_J})
\end{equation}
where  $\overline{e^{\Aa} R_{\balpha}} = \frac{1}{J} \sum_{j=1}^J e^{a_j}R_{\alpha_j} \in\R^{2\times 2}$,  $\bar \Y_j =  \frac{1}{k}\sum_{\ell=1}^k [\Y_j]_\ell \in \R^2$, that is $\bar A \Y_j = \1_{k} \otimes \bar\Y_j \in\R^{k\times 2}$ is a degenerated configuration, and $\bar \Y = \frac{1}{J}\sum_{j=1}^J \bar\Y_j \in\R^{2}$.

\bigskip

This two steps procedure is equivalent to the optimization problem \eqref{def.estimateurs} as we have the following decompositions $M^{\lambda}(\Aa,\balpha,\b) = M^{\lambda}_0(\Aa,\balpha) + \bar M(\Aa,\balpha,\b)$ and $ \bTheta_0 = \bTheta_{0}^{\Aa,\balpha} \times \bTheta_{0}^{\b}$, implying   following result (see the Appendix for a detailed proof)
\begin{lemma}\label{lemme.argmin}
%We have, 
$
{ (\hat \Aa^{\lambda},\hat \balpha^{\lambda},\hat\b\vphantom{\b}^{\lambda}) \in \argmin_{( \Aa,\balpha,\b)\in\bTheta_0}\limits M^{\lambda}(\Aa,\balpha,\b)}
\Longleftrightarrow\begin{cases}
(\hat \Aa^{\lambda},\hat \balpha^{\lambda}) &\in\argmin_{(\Aa,\balpha)\in\bTheta_{0}^{\Aa,\balpha} }\limits M^{\lambda}_0(\Aa,\balpha)  \\ \hfill \hat\b\vphantom{\b}^\lambda  & = \argmin_{\b\in\bTheta_0^{\b}}\limits \bar M (\hat \Aa^{\lambda},\hat \balpha^{\lambda} ,\b).\end{cases}
$
\end{lemma}

\section{Consistency results}
\label{part.consistency}
 In what follows, $C,C_0,C_1$ denote positive constants whose value may change from line to line. The notation $C( \cdot)$ specifies the dependency of $C$ on some quantities. 

\subsection{Consistent estimation of the deformation parameters}
\paragraph{Rotation and scaling.} Recall that the rotation and scaling parameters are estimated separately on the smoothed and centered observations. We have the following result,
\begin{theorem}\label{th:rotscal}
Consider model \eqref{eq:perturbmodel} and suppose that Assumptions \ref{ass.f} and \ref{ass.zeta}  hold and that  Assumption  \ref{hyp.Theta} is verified with $\max\{A,\AA\}<0.1$.  Consider $(\hat \Aa^{\lambda},\hat\balpha^\lambda)$ the estimators defined in equation \eqref{eq.estimRotScal}. If $\lambda = k^{\frac{1}{2s+1}}$ then for all $x>0$ we have
\begin{equation}\label{eq.conrotscal}
\P\left( \frac{1}{J}\snorm{(\hat \Aa^{\lambda}, \hat \balpha^{\lambda}) -(\Aa^*, \balpha^*)}^2_{\R^{2J}} \geq C_1(L,s,A,\AA,\f_0) A_1(k,J,x) + C_2(A,\AA) A_2(J,x) \right) \leq 4 e^{-x},
\end{equation}
where $C_1(L,s,A,\AA,\f_0), C_2(A,\AA)  >0 $ are positive constants independent of $k$ and $J$, $A_1(k,J,x)=F\big(k^{-\frac{2s}{2s+1}} \big) + F\big(V_1(k,J,x)\big)$ with $V_1 (k,J,x )= 3\gamma_{\max}(k) k^{-\frac{2s}{2s+1}}\left( 1 + \sqrt{2\frac{x}{Jk}} + 2\frac{x}{Jk} \right)$ and $A_2(J,x) = \linebreak[1] \Big(\tfrac{x}{3J} +\sqrt{\tfrac{2x}{J}} \Big)^2$, where $F(u) = u+\sqrt{u}$, for all $u\geq 0$.
\end{theorem}

Remark that a direct consequence of Theorem \ref{th:rotscal} is the consistency of $(\hat \Aa^{\lambda}, \hat \balpha^{\lambda}) $ to $(\Aa^*, \balpha^*)$ when $\min\{k,J\}\to\infty$. Indeed, we have  $\lim_{u \to 0} F(u) = 0$  and under Assumption \ref{ass.zeta} and for any fixed $J\geq 1$ and $x>0$, the term $V_1(k,J,x)$ tends to zero as $k$ goes to infinity. Hence for any $x>0$ and $J\geq 1$, we have 
$$
 \lim_{k\to\infty} A_1(k,J,x) = 0.
$$ 
Under the same hypothesis as in Theorem \ref{th:rotscal} but without the bounds on $A$ and $\AA$, Proposition \ref{prop.ConvAAlpha} then ensures the convergence of $(\hat \Aa^{\lambda}, \hat \balpha^{\lambda})$ to $(\Aa^*_{\bTheta_0}, \balpha^*_{\bTheta_0})$ as $J$ remains fixed and $k\to\infty$. Now, for all $x>0$, we have
$$
\lim_{J\to\infty}  A_2(J,x)= 0.
$$ 
Thus a double asymptotic $\min\{k,J\} \to\infty$ ensures that $\frac{1}{J}\snorm{(\hat \Aa^{\lambda}, \hat \balpha^{\lambda}) -(\Aa^*, \balpha^*)}^2_{\R^{2J}}$ tends to 0 in probability.

\paragraph{Translation parameters.} We have the following result,
\begin{theorem}\label{th:shift}
Consider the hypothesis and notations of Theorem \ref{th:rotscal} and the estimator $\hat \b\vphantom{\b}^{\lambda}$ given by formula \eqref{eq.defhatb}. Then we have for all $x>0$,
\begin{equation}
\label{eq.conshift}
\P\Big(\frac{1}{J} \snorm{\hat \b\vphantom{\b}^{\lambda} - \b^* }^2_{\R^{2J}} \geq C_3(L,s,A,\AA,B,\f) A_3(k,J,x) + C_4(A,\AA,B) A_2(J,x)  \Big) \leq 9e^{-x},
\end{equation}
where $C_3(L,s,A,\AA,B,\f), C_4(A,\AA,B)>0 $ are positive constants independent of $k$ and $J$, $A_3(k,J,x) =   F( k^{-\frac{2s}{2s+1}} ) + F(V_1(k,J,x))+V_2(k,J,x)$ with  %$ V_1 (k,J,x ) = 3\gamma_{\max}(k) k^{-\frac{2s}{2s+1}}\big( 1 + \sqrt{2\frac{x}{Jk}} + 2\frac{x}{Jk} \big)$ and 
$ V_2(k,J,x)=4\frac{\gamma_{\max}(k)}{k} \big(1 +  \sqrt{2\frac{x}{J}} + \frac{x}{J}\big)$.%, and $B_2(J,x) = \Big(\sqrt{\tfrac{2x}{J}} + \tfrac{x}{3J}\Big)^2 $. Again, $F(u) = u + \sqrt u$ for all $u\geq 0$.
\end{theorem}

Similar comments to those made after Theorem \ref{th:rotscal} are still valid here. For any $J\in\N$ and $x>0$, we have  
$$
\lim_{k\to\infty} A_3(k,J,x) = 0,
$$ 
since $V_2(k,J,x)$ tends to 0 as $k$ goes to infinity by Assumption \ref{ass.zeta}. Under the same hypothesis as in Theorem \ref{th:shift} but without the bounds on $A$ and $\AA$, Proposition \ref{prop.ConvShift} ensures the convergence in probability of $\b\vphantom{\b}^{\lambda} $ to $\b^*_{\bTheta_0}$  with only an asymptotic in $k$. In the double asymptotic setting $\min\{k,J\} \to \infty$, Theorem \ref{th:shift} ensure the consistency of $\hat\b\vphantom{\b}^\lambda$ to the true value $\b^*$ of the translation parameters.

\subsection{Consistent estimation of the mean shape}
\label{part.consfunc}

\begin{theorem}\label{th:shape}
Consider model \eqref{eq:perturbmodel} and suppose that Assumptions \ref{ass.f} and \ref{ass.zeta} hold and that Assumption \ref{hyp.Theta} is verified with $\max\{A,\AA\}<0.1$. Consider the estimator $\hat \f\vphantom{\f}^\lambda$ defined in \eqref{eq.defflambda} and let $ \lambda = k^{\frac{1}{2s+1}}$. Then, we have for all $x>0$,
$$
\P\bigg( \frac{1}{k}\snorm{\hat \f\vphantom{\f}^\lambda - \f}^2_{\Sh}\geq  C(L,s,A,\AA,B,\f)\Big( A_1(k,J,x) + A_3(k,J,x) + A_2(J,x) \Big) \bigg) \leq 14e^{-x},
$$
where $C(L,s,A,\AA,B,\f)>0$ is a constant independent of $k$ and $J$,    $ A_1(k,J,x)= F\big(k^{-\frac{2s}{2s+1}} \big) + F\big(V_1(k,J,x)\big)$ with $ V_1 (k,J,x )= 3\gamma_{\max}(k) k^{-\frac{2s}{2s+1}}\left( 1 + \sqrt{2\frac{x}{Jk}} + 2\frac{x}{Jk} \right)$, $A_3(k,J,x) =  V_2(k,J,x)+\linebreak[5] F( k^{-\frac{2s}{2s+1}} ) + F(V_1(k,J,x)) $ with $ V_2(k,J,x)=4\frac{\gamma_{\max}(k)}{k} \big(1 +  \sqrt{2\frac{x}{J}} + \frac{x}{J}\big) $ and $A_2(J,x) = \Big(\sqrt{\tfrac{2x}{J}} + \tfrac{x}{3J}\Big)^2$.
\end{theorem}

The terms $A_1(k,J,x),\ A_3(k,J,x)$ and $A_2(J,x) $ that appear in the statement of Theorem \ref{th:shape} are the same to those appearing in Theorems \ref{th:rotscal} and \ref{th:shift}. As a consequence, we have $\frac{1}{k}\snorm{\hat \f\vphantom{\f}^\lambda - \f}^2_{\Sh} \to 0$ in probability when $\min\{k,J\} \to \infty$ and Theorem \ref{th:shape} gives rates of convergences of $\hat \f\vphantom{\f}^\lambda $ to $ \f $ thanks to a concentration inequality.%  see the comments  after Theorems \ref{th:rotscal} and \ref{th:shift}. 

Theorem \ref{th:shape} is similar to Theorems \ref{th:rotscal} and \ref{th:shift} and there is no need to assume an extra bound on $A$ and $\AA$ to ensure the convergence in probability of $\hat \f\vphantom{\f}^\lambda $ %to $ \f $ 
when $J$ is fixed and $k\to\infty$, see  the discussion following Theorems \ref{th:rotscal} and \ref{th:shift}. Let 
\begin{equation}\label{eq.fg0}
\f_{\bTheta_0} = e^{\bar \Aa^*}(\f+\1_{k} \otimes \bar \b^* (\overline{e^{\Aa^*} R_{\balpha^*}})^{-1} )R_{\bar\balpha^*}
\end{equation}
where $\bar\Aa^* =   \frac{1}{J}\sum_{j=1}^J a_j^*\in\R$, $\bar\balpha^* =   \frac{1}{J}\sum_{j=1}^J \alpha_j^*\in\R$, $\bar\b^* =  \frac{1}{J}\sum_{j=1}^J \b_j^*\in\R^2$ and  $\overline{e^{\Aa^*} R_{\balpha^*}} = \frac{1}{J}\sum_{j=1}^J e^{a^*_j}R_{\alpha_j^*}$ is an invertible $2\times 2$ matrix, see also formula \eqref{eq.g0}. A slight modification of the proof of Theorem \ref{th:shape} gives the following inequality,
$$
\P\bigg( \frac{1}{k}\snorm{\hat \f\vphantom{\f}^\lambda - \f_{\bTheta_0}}^2\geq  C(L,s,A,\AA,B,\f)\Big( A_1(k,J,x) + A_3(k,J,x)\Big) \bigg) \leq 14e^{-x}.
$$
This yields to statements \eqref{eq.thmain1} %and \eqref{eq.thmain2} 
in Theorem \ref{theo:main}.

\medskip

%In this case, we were not able to derive a rate of convergence and it yields to statements \eqref{eq.thmain1} and \eqref{eq.thmain2} in Theorem \ref{theo:main}.

\section{Numerical experiments}\label{part.num}
\subsection{Description of the data}
\label{part.desData}

We make here some numerical simulations to show the effect of the dimension $k$ and the number $J$  of observations on the estimation of the deformation parameters and the mean pattern with data generated by model \eqref{eq:perturbmodel}. Different types of noise are considered. For all $t\in[0,1]$, let
$$
f(t)= (10 \sin^2(\pi t) + \cos(10\pi) +20, 2\sin( 6 \pi t) - 11 \sin^2(\pi t) + 12  \exp( - 25 ( t- 0.4)^2) +1).
$$
This curve is plotted in Figure \ref{fig.meanpattern}. The deformation parameters $(a^*_j,\alpha^*_j,b^*_j)$, $j=1,\ldots,J $, are i.i.d uniform random variables taking their values in $\Theta = [-\frac{1}{4},\frac{1}{4}]\times[-\frac{1}{2},\frac{1}{2}]\times [-1,1]^2$. The law of the deformation parameters is supposed to be unknown \textit{a priori} and the minimization is performed on the constraint set 
$$
 \bTheta_0 = \bigg\{ (\Aa,\balpha,\b) \in [-1,1]^{J} \times [-1,1]^{J}\times [-5,5]^{2J}, \  \sum_{j=1}^{J} a_{j} =  0,\ \sum_{j=1}^{J} \alpha_{j} = 0 \mbox{ and } \sum_{j=1}^{J} b_{j} = 0  \bigg\}.
$$ 
Recall our notations: the error term is denoted by $\bzeta = (\bzeta^{(1)},\bzeta^{(2)})\in\R^{k\times 2}$ and the vectorized version of $\bzeta$ is denoted by $\tilde \bzeta \in \R^{2k}$. The simulations were run with three different kinds of noise. 
\begin{figure}[t]
\begin{center}
\subfloat[]{\includegraphics[width=5.3cm]{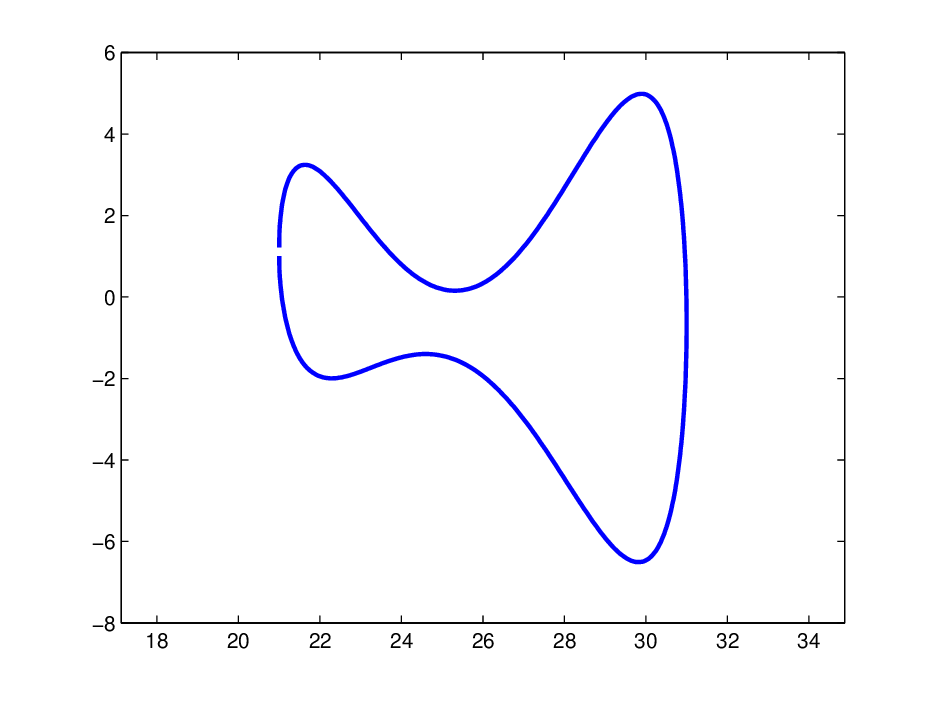}\label{fig.meanpattern0}}\hfill
\subfloat[]{\includegraphics[width=5.3cm]{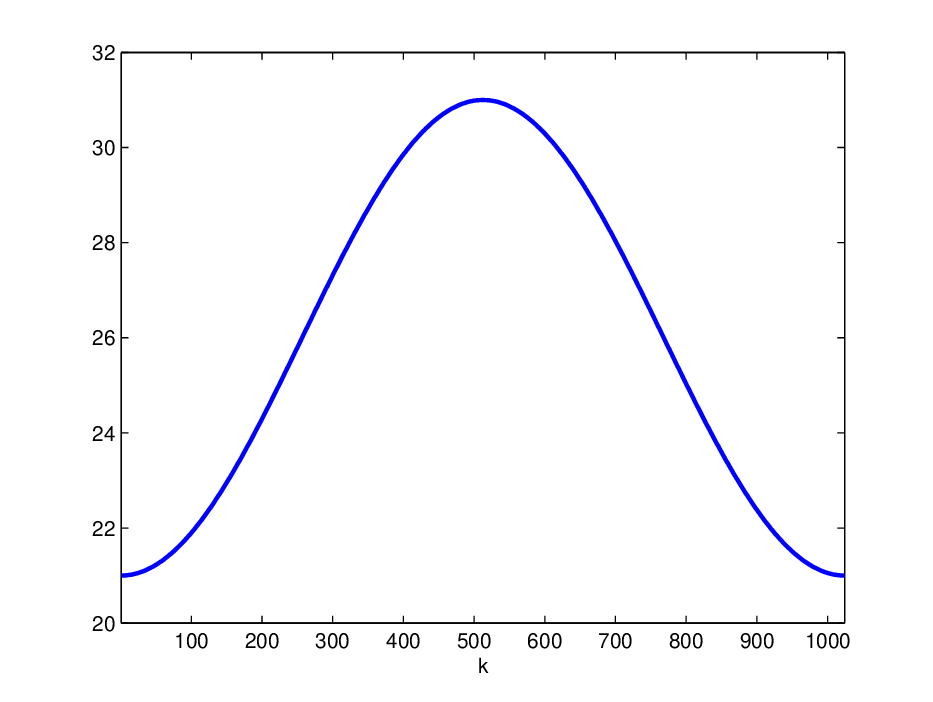}\label{fig.meanpattern1}}\hfill
\subfloat[]{\includegraphics[width=5.3cm]{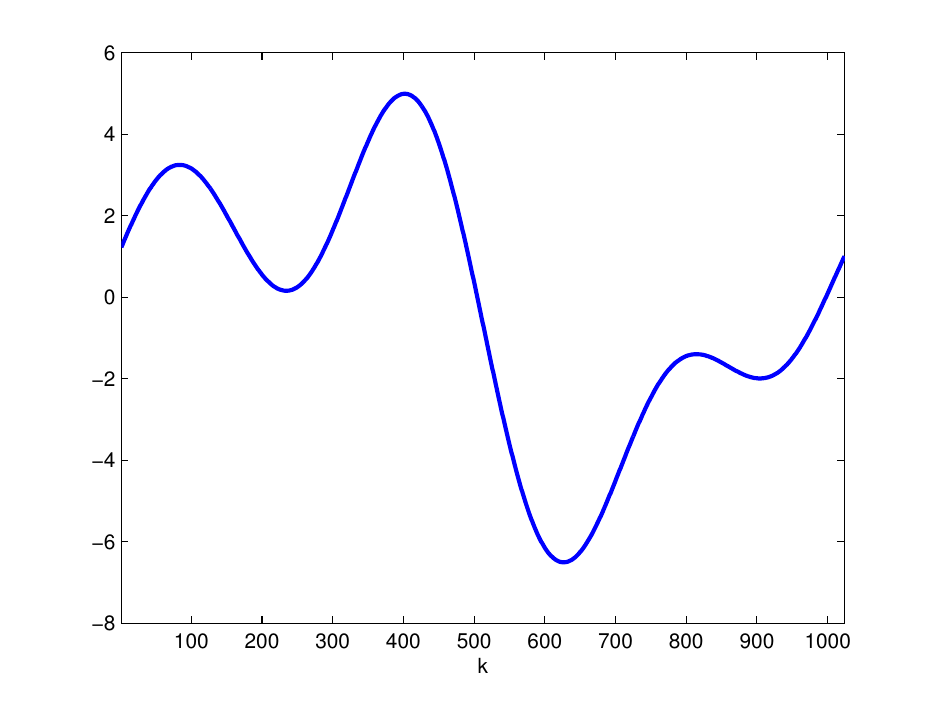}\label{fig.meanpattern2}}
\caption[f]{\subref{fig.meanpattern0} Plot of the mean pattern $\f = (\f^{(1)}, \f^{(2)})$ used in the simulations with $k=1024$. \subref{fig.meanpattern1} The first coordinates $\f^{(1)}$. \subref{fig.meanpattern2} The second coordinates $\f^{(2)}$.} \label{fig.meanpattern}
\end{center}
\end{figure} 
\begin{description}
\item[White noise :] the random variable $ \tilde \bzeta = (\zeta_1^{(1)},\zeta_2^{(1)}, \ldots ,\zeta_1^{(k)},\zeta_2^{(k)})'\in\R^{2k}$ is a centered Gaussian vector of variance 
$$
\bSigma_1= 4 Id_{2k}.
$$ 
We have $\gamma_{\max}(k) = 4$ and  this correspond to an isotropic Gaussian noise as in \cite{MR1436569,MR1618880}, see Figure \ref{fig.exempleBlanc}.%\ref{fig.noiseBlanc}, \ref{fig.noiseBlanc2} and \ref{fig.spectreblanc}.
\item[Weakly correlated noise :] the random variable $\tilde \bzeta$ is a centered Gaussian vector  of variance $\bSigma_2$ with 
$$
\bSigma^{\frac{1}{2}}_2 = Id_2 \otimes \left[\tfrac{1}{2}\exp\left(-\tfrac{\sabs{\ell-\ell'}}{100}^2\right)\right]_{\ell,\ell'=1}^k
$$
Hence, $\bSigma_2 $ is a Toeplitz matrix and it follows from classical matrix theory, see e.g. \cite{MR1084815}, that $\gamma_{\max}(k)$ is bounded (here $\gamma_{\max} \leq 80$). See Figure  \ref{fig.exempleBig}. %\ref{fig.noiseBig}, \ref{fig.noiseBig2} and \ref{fig.spectrebig}.
\item[Highly correlated noise :] the random variable  $\tilde \bzeta $  is a centered Gaussian vector of variance 
$$
\bSigma_3 =Id_2 \otimes  P \diag\left(\tfrac{\ell^2}{2k}, \ \ell=1,\ldots,k\right) P'
$$
where $P$ is an arbitrary matrix in $\mathbf{SO}(k)$. Hence, in this case $\gamma_{\max}(k) = \frac{k}{4}$ and the level of noise increase with $k$, see Figure \ref{fig.exempleBen}. % \ref{fig.noiseBen}, \ref{fig.spectreben} and \ref{fig.noiseBen2}.
\end{description}

\begin{figure}[]
\begin{center}
\subfloat[Observations $\Y$]{\includegraphics[width=5.3cm]{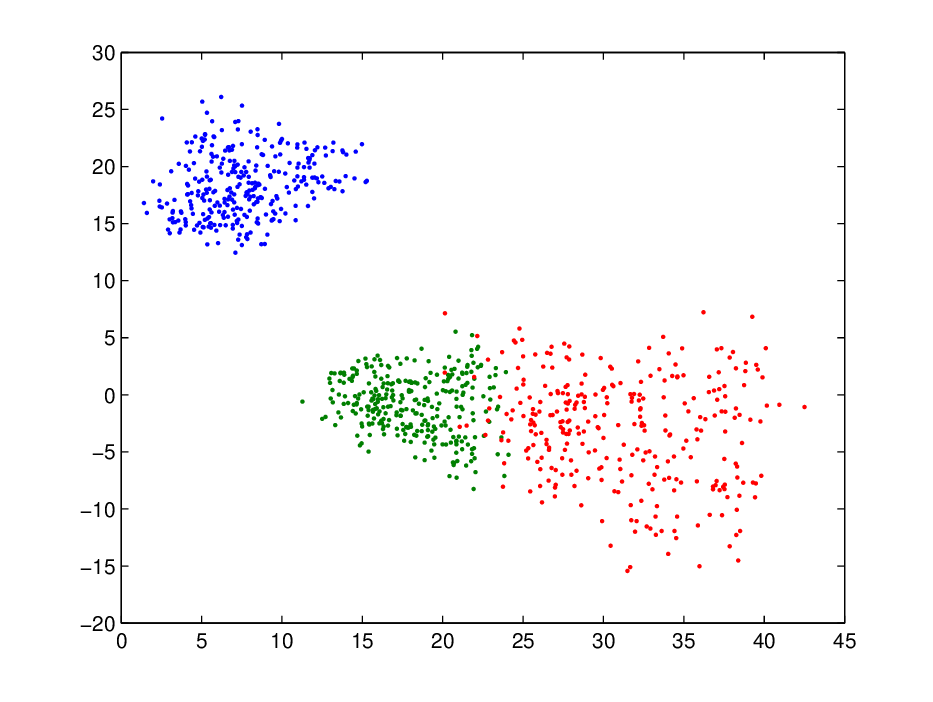}\label{fig.noiseBlanc}}\hfill
\subfloat[A realization of $\bzeta_{1}$]{\includegraphics[width=5.3cm]{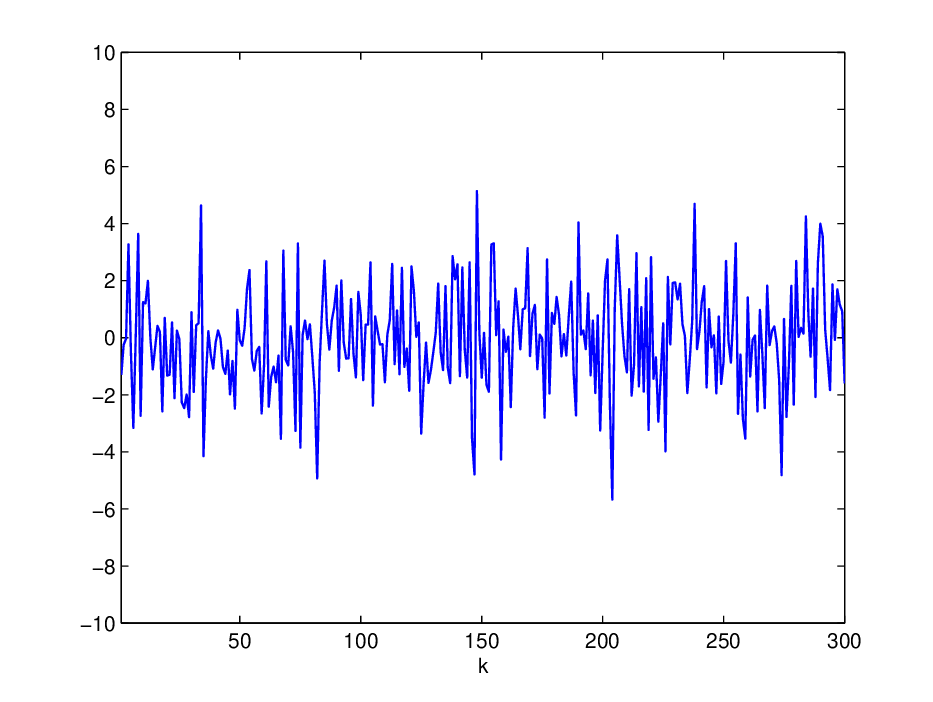}\label{fig.noiseBlanc2}}\hfill
\subfloat[Spectrum of $\bSigma_1$ ]{\includegraphics[width=5.3cm]{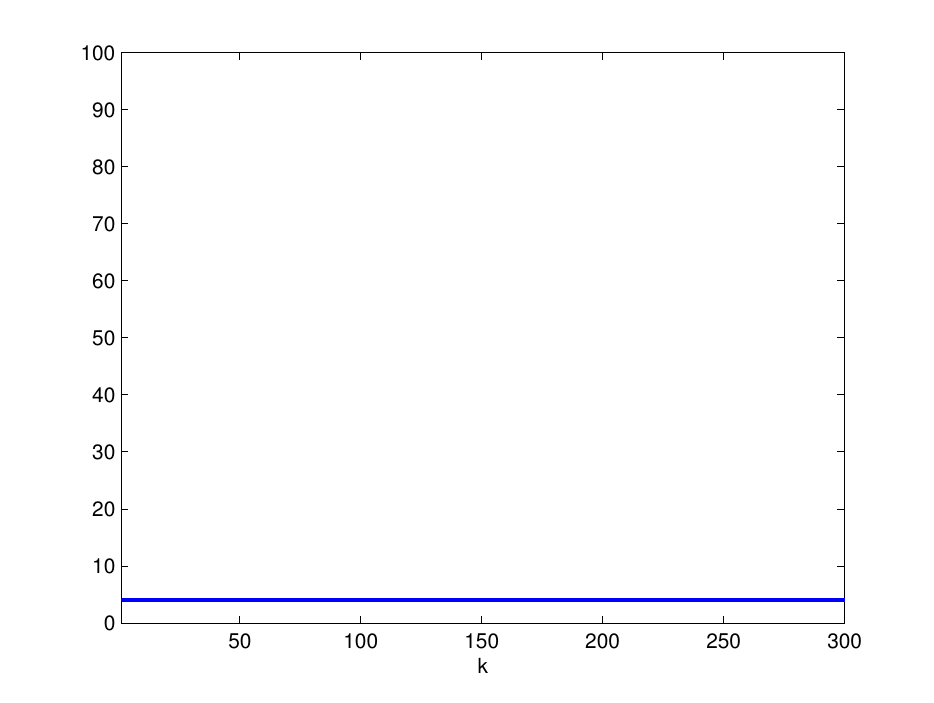}\label{fig.spectreblanc}}
\caption{Example of data generated by model \eqref{eq:perturbmodel} with white noise.} \label{fig.exempleBlanc}
\vspace*{2cm}
\subfloat[Observations $\Y$]{\includegraphics[width=5.3cm]{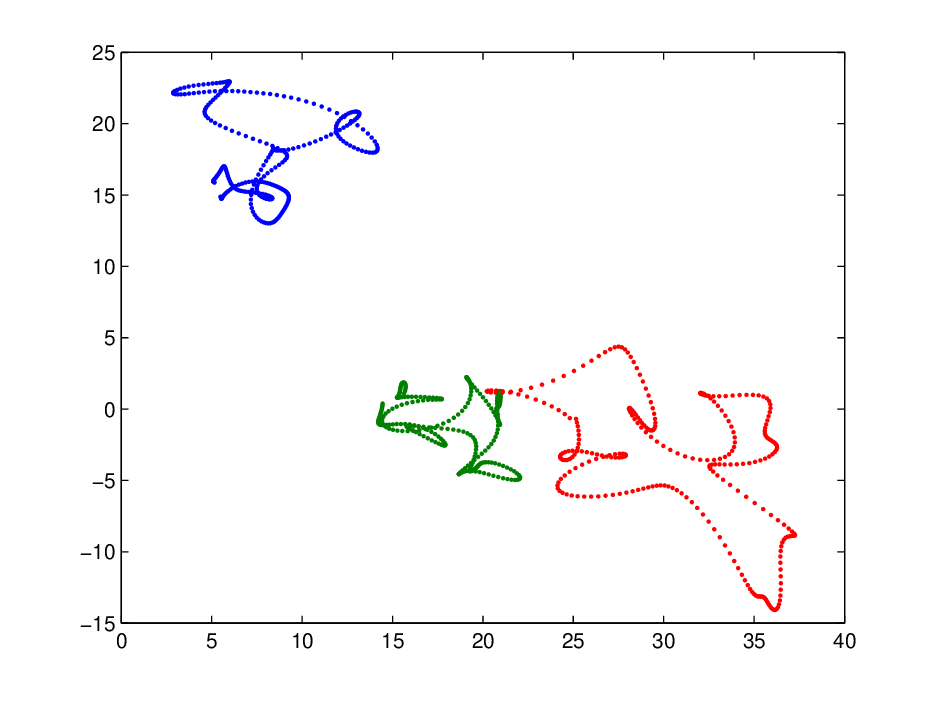}\label{fig.noiseBig}}\hfill
\subfloat[A realization of $\bzeta_{1}$]{\includegraphics[width=5.3cm]{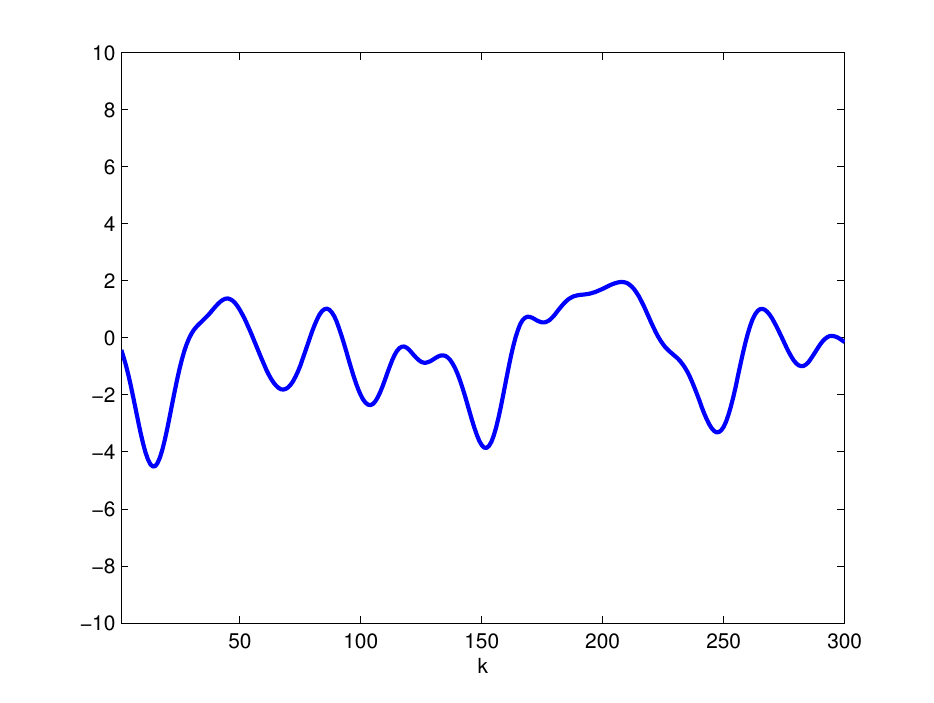}\label{fig.noiseBig2}}\hfill
\subfloat[Spectrum of $\bSigma_2$]{\includegraphics[width=5.3cm]{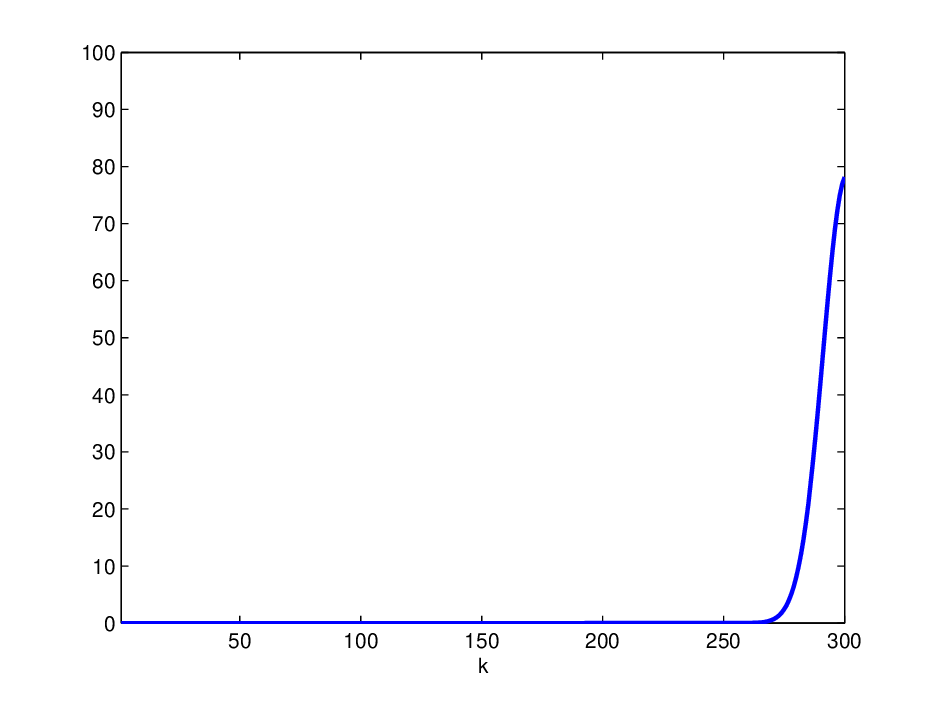}\label{fig.spectrebig}}
\caption{Example of data generated by model \eqref{eq:perturbmodel} with weakly correlated noise.} \label{fig.exempleBig}
\vspace*{2cm}
\subfloat[Observations $\Y$]{\includegraphics[width=5.3cm]{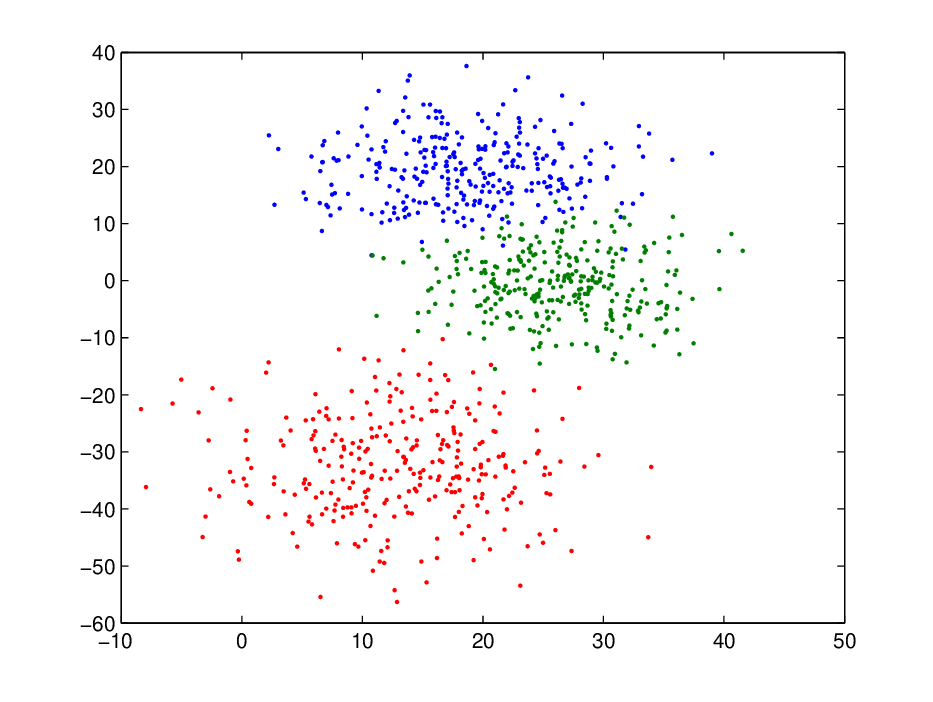}\label{fig.noiseBen}}\hfill
\subfloat[A realization of $\bzeta_{1}$]{\includegraphics[width=5.3cm]{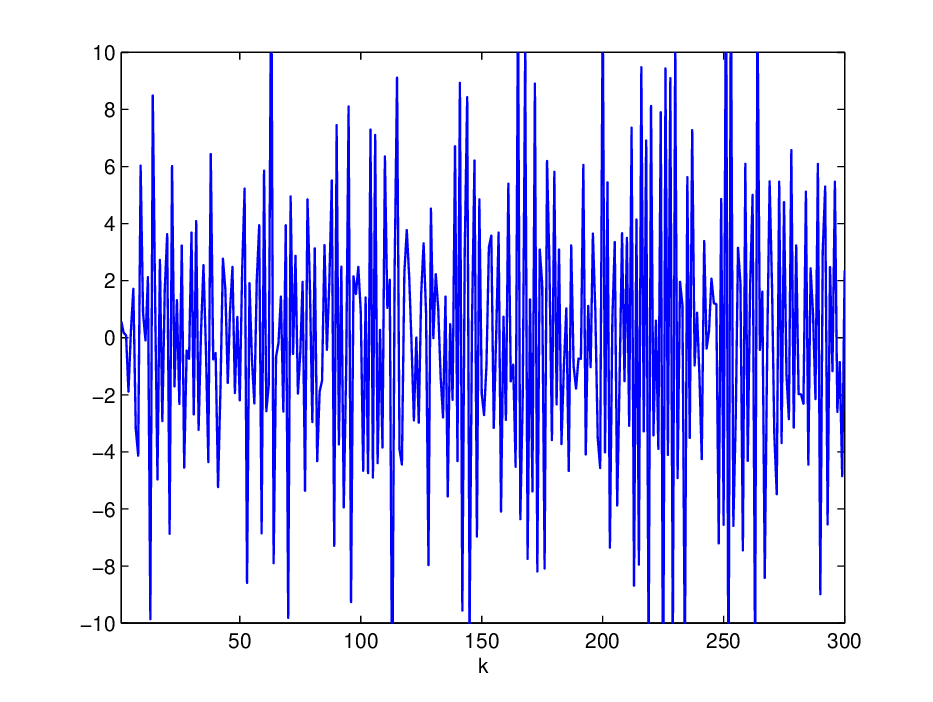}\label{fig.noiseBen2}}\hfill
\subfloat[Spectrum of $\bSigma_3$]{\includegraphics[width=5.3cm]{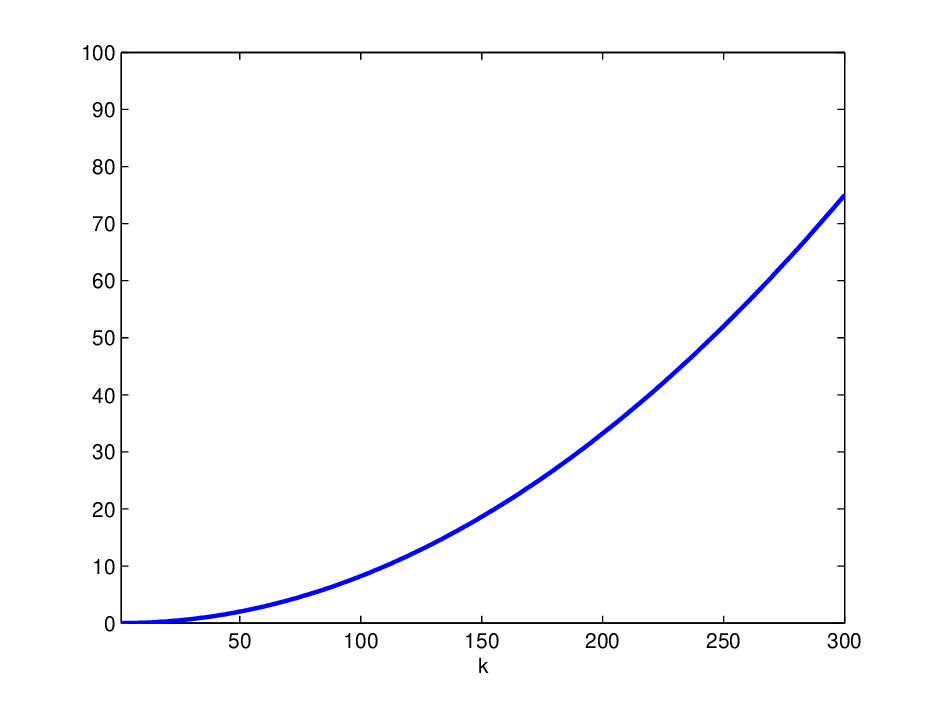}\label{fig.spectreben}}
\caption{Example of observations generated by model \eqref{eq:perturbmodel} with the highly correlated noise. %three different error terms $\bzeta$: the first column corresponds to white noise, the second column to the stationnary error term and the third column to t
%: \subref{fig.noiseBlanc}  the observations, \subref{fig.noiseBlanc2} a realization of $\bzeta$ and \subref{fig.spectreblanc}  the spectrum of the covariant matrix $\bSigma$. T 
} \label{fig.exempleBen}
\end{center}
\end{figure} 

\subsection{Description of the procedure}

The estimation procedure follows the guidelines described in Section \ref{part.EstimProc}. We are testing the effect of the number $J$ of observations and the number $k$ of landmarks on the estimation of the parameters of interest of model \eqref{eq:perturbmodel}. All the simulations are performed with $J=10,100,500$ and $k=20,50,100,1000,3000$ and for each combination of these two factors the simulations are performed with $M=30$ repetitions of model \eqref{eq:perturbmodel}. 

Moreover, estimations are done without and with the pre-smoothing step. In the former case we have $\lambda(k) = \frac{k}{2}$, that is, there is no reduction of the dimension. In the latter case, the smoothing parameter $\lambda(k)$ is fixed manually to ensure a proper  reconstruction of the mean pattern $f$. %Note that the Fourier coefficient of $f$ decrease exponentially and by equation \eqref{eq.BiasVar} we can set $\lambda(k)= c \sqrt{k}$ for some $c>0$. The value \ref{eq.BiasVar} is convenient for large $k$ but is obviously too small for small  values of $k$ as
Note that we need $\lambda>7$ to get correct results and we took $\lambda_{20} = \lambda_{50}  = \lambda_{100} = 7$, $\lambda_{1000} = 11 $ and $\lambda_{3000} = 25 $.

We use a quasi-Newton trust-region based algorithm to solve the optimization problems \eqref{eq.estimRotScal} and \eqref{eq.defhatb}. The formula for the gradient is given in \eqref{eq:GradD}. All the computations are performed with \emph{Matlab}.

\subsection{Results: estimation of the mean pattern}

%The framework is the same as in Section \ref{part.resultdef} for the estimation of the deformation parameters. Let $k\in\{20,50,100,1000,3000\}$ and $J\in\{10,100,500\}$  be fixed. 
For each of the $30$ repetitions of model  \eqref{eq:perturbmodel} with the possible values of $k$ and $J$, we compute the quantities $\frac{1}{k}\|\hat \f\vphantom{\f}^\lambda -\f^*_{\bTheta_0}\|^2_{\Sh} $ where $\hat\f\vphantom{\f}^\lambda  $ corresponds to the smoothed Procrustes mean of the observations defined in \eqref{eq.defflambda} and, $\frac{1}{k}\|\hat \f  -\f^*_{\bTheta_0}\|^2_{\Sh} $ where $\hat \f$ is the (non smoothed) Procrustes mean of the data. Recall that $\f_{\bTheta_0} $ is defined by formula \eqref{eq.fg0}.

%Boxplots of the results are given at Figures \ref{fig.diffFblanc}, \ref{fig.diffFbig} and \ref{fig.diffFben}. Again, the abscissa gives the different values of $k$ and the various colors of the boxplots represent the various values of the number  $J$ of observations. In red we have $J=10$, in green $J=100$ and in blue $J=500$. 
Boxplots of the results are given in Figures \ref{fig.diffFblanc}, \ref{fig.diffFbig} and  \ref{fig.diffFben} for the different kinds of error terms described in Section \ref{part.desData}. In the figures, the abscissa represents the different values of the number  $k$ of landmarks and boxplots in red correspond to $J=10$ observations,  in green to $J=100$ observations and in blue to $J=500$ observations.

The estimation of the mean pattern with the  white noise error term is given by Figure \ref{fig.diffFblanc}. In Figure \ref{fig.BruitBlan-diffF}, for a fixed $k$, the non-smoothed version $\frac{1}{k}\|\hat \f -\f^*\|^2_{\Sh} $ decreases when $J$ increases. Moreover, the values of $\frac{1}{k}\|\hat \f -\f^*\|^2_{\Sh} $  remain stable when $J$ remains fixed and $k$ increases. Recall that this framework corresponds to the isotropic Gaussian noise described in \cite{MR1436569}. The simulations seem to confirm their conclusions and show that in this framework the dimension $k$ is not preponderant. In Figure \ref{fig.BruitBlan-diffFlambda},  the smoothed version  $\frac{1}{k}\|\hat\f\vphantom{\f}^\lambda - \f\|^2_{\Sh}$  decreases   when $J$ and $k$ increase. The main difference with the non-smoothed estimation is the convergence to 0 of $\frac{1}{k}\|\hat\f\vphantom{\f}^\lambda - \f\|^2_{\Sh}$ when $J$ remains fixed and $k$ increases.
 
In Figure \ref{fig.diffFbig}, the results of the estimation of the mean pattern are plotted for the weakly correlated noise term.  Figure \ref{fig.BruitBig-diffF} shows us a similar behavior of the non-smoothed Procrustes mean but with non-decreasing values of $\frac{1}{k}\|\hat \f\vphantom{\f}^\lambda -\f^*\|^2_{\Sh}  $ when $k$ increases and $J$ remains fixed. In Figure \ref{fig.BruitBig-diffFlambda}, the smoothed Procrustes mean converges as $k$ goes to infinity and the bigger $J$ is the faster the convergence is.

The results of the estimations of the mean pattern with the highly correlated noise are presented Figure \ref{fig.diffFben}. The results that appear in Figure \ref{fig.BruitBen-diffF} are quite different compared to those presented Figures \ref{fig.BruitBlan-diffF} and \ref{fig.BruitBig-diffF}. The estimation seems to be worst when $k$ increases and $J$ remains fixed. The reason is that the level of noise, measured by $\gamma_{\max}(k)$ is increasing with $k$. %In all instances, 
The smoothing step is efficient and the estimations presented Figure \ref{fig.BruitBen-diffFlambda} have a similar behavior to  those given in  Figures \ref{fig.BruitBlan-diffFlambda} and \ref{fig.BruitBig-diffFlambda}.

\begin{figure}[]
\begin{center}
\subfloat[f][Boxplot of $\frac{1}{k}\|\hat \f -\f^*\|^2_{\Sh}$ (estimation without smoothing).]{\includegraphics[height=5cm]{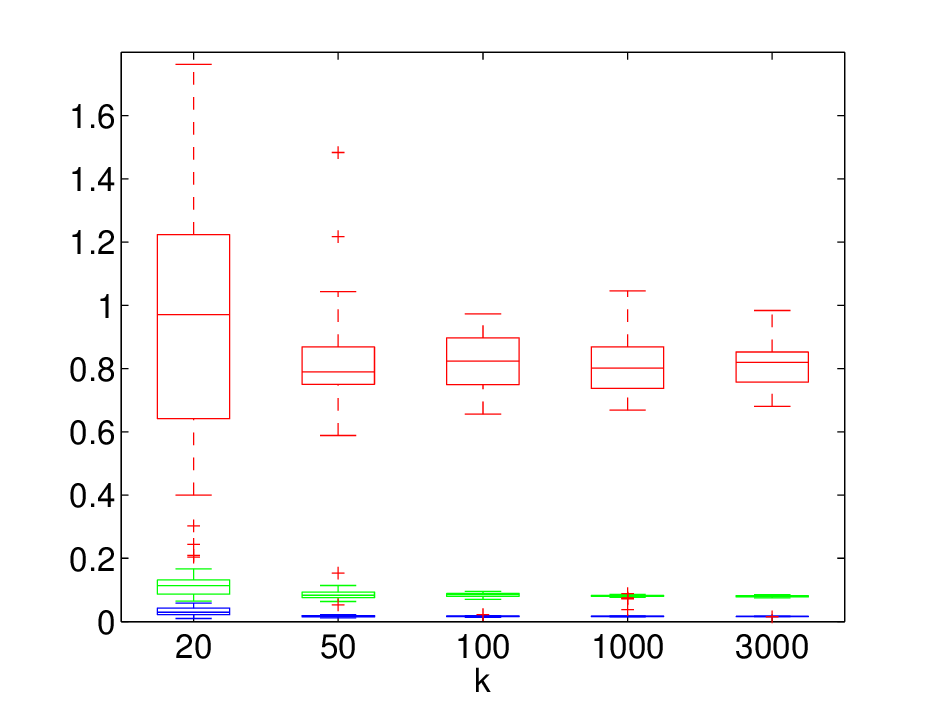}\label{fig.BruitBlan-diffF}}\hfill
\subfloat[f][ Boxplot of $\frac{1}{k}\|\hat \f\vphantom{\f}^\lambda -\f^*\|^2_{\Sh}$ (estimation with smoothing). The missing boxplots for $k=20$ belong to the range $ \lbrack 0.7,1.9\rbrack$.]{\includegraphics[height=5cm]{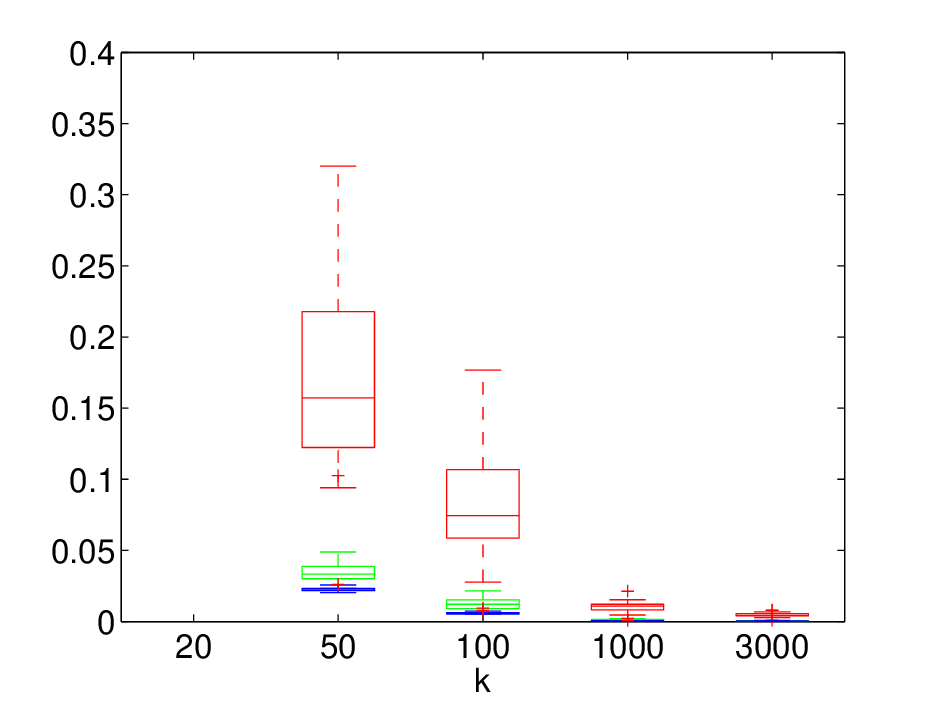}\label{fig.BruitBlan-diffFlambda}}

\caption{ Estimation of the mean pattern with the white noise. Boxplot in red correspond to $J=10$, in green to $J=100$ and in blue to $J=500$.}\label{fig.diffFblanc}

\subfloat[f][Boxplot of $\frac{1}{k}\|\hat \f -\f^*\|^2_{\Sh}$ (estimation without smoothing).]{\includegraphics[height=5cm]{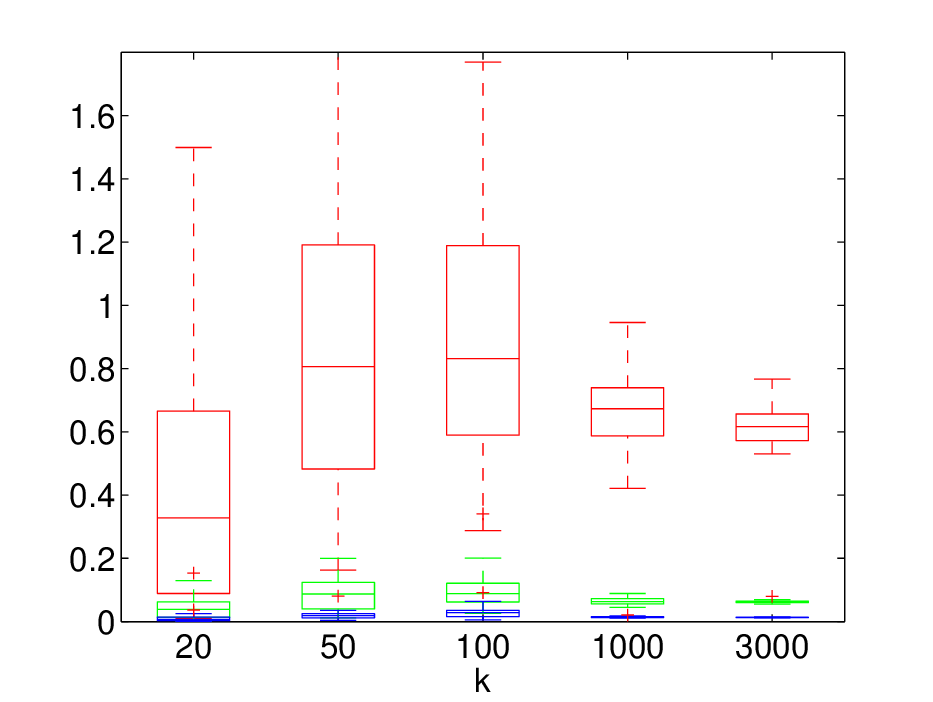}\label{fig.BruitBig-diffF}}\hfill
\subfloat[f][ Boxplot of $\frac{1}{k}\|\hat \f\vphantom{\f}^\lambda -\f^*\|^2_{\Sh}$ (estimation with smoothing).]{\includegraphics[height=5cm]{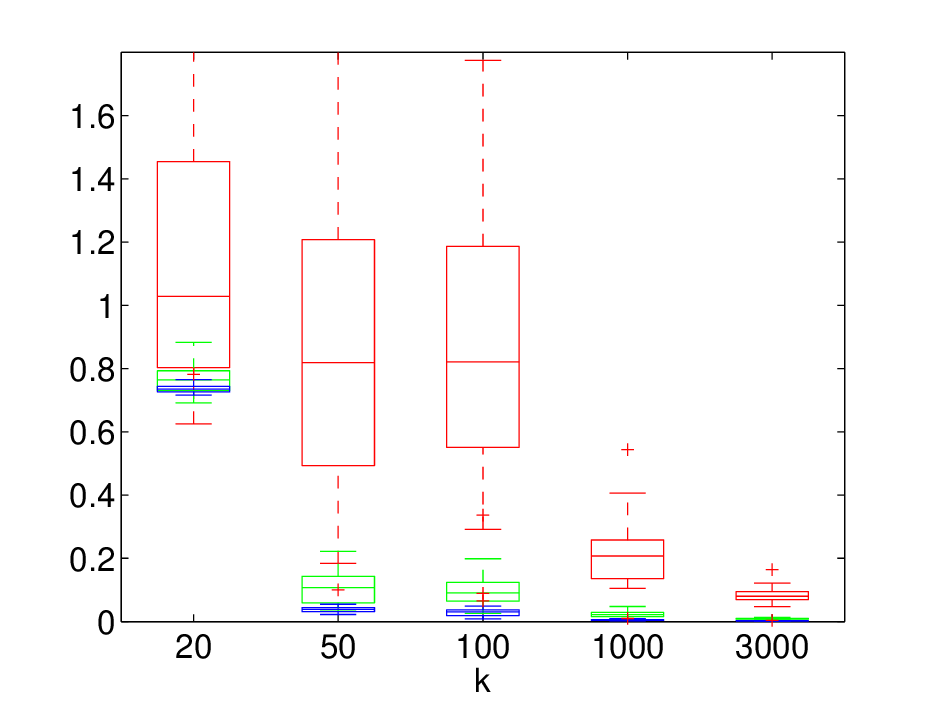}\label{fig.BruitBig-diffFlambda}}
%\end{center}
\caption{ Estimation of the mean pattern  with the weakly correlated noise.  Boxplot in red correspond to $J=10$, in green to $J=100$ and in blue to $J=500$.}\label{fig.diffFbig}

\subfloat[f][Boxplot of $\frac{1}{k}\|\hat \f -\f^*\|^2_{\Sh}$ (estimation without smoothing). The missing red boxplot for $k=20$ and $J=10$ belongs to the range $ \lbrack 70,110\rbrack$.]{\includegraphics[height=5cm]{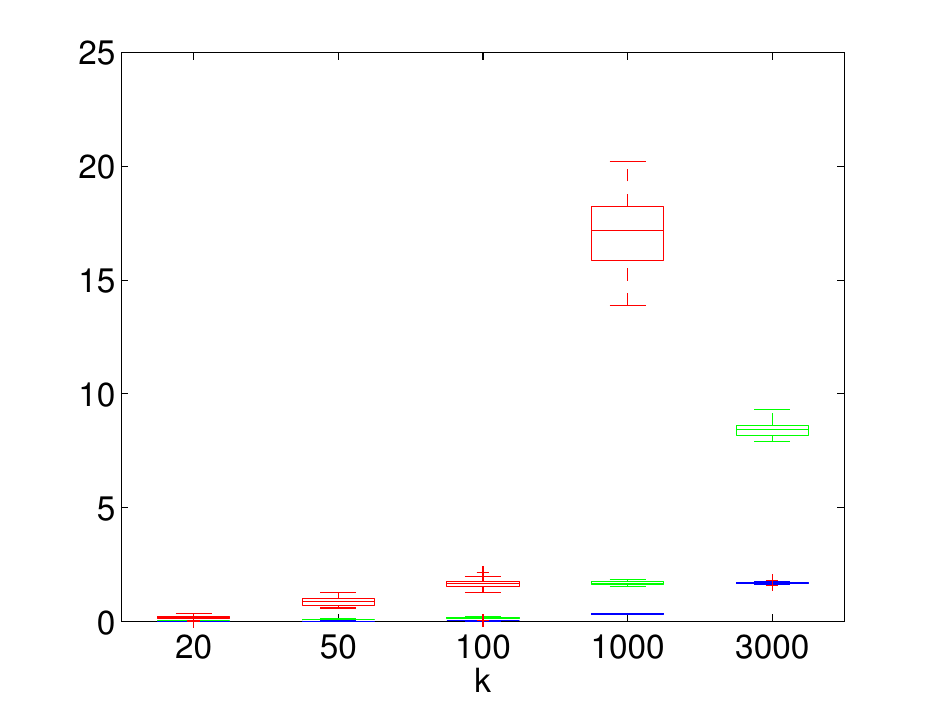}\label{fig.BruitBen-diffF}}\hfill
\subfloat[f][ Boxplot of $\frac{1}{k}\|\hat \f\vphantom{\f}^\lambda -\f^*\|^2_{\Sh}$ (estimation with smoothing). The missing boxplots for $k=20$ belong to the range $ \lbrack 0.7,1.9\rbrack$. ]{\includegraphics[height=5cm]{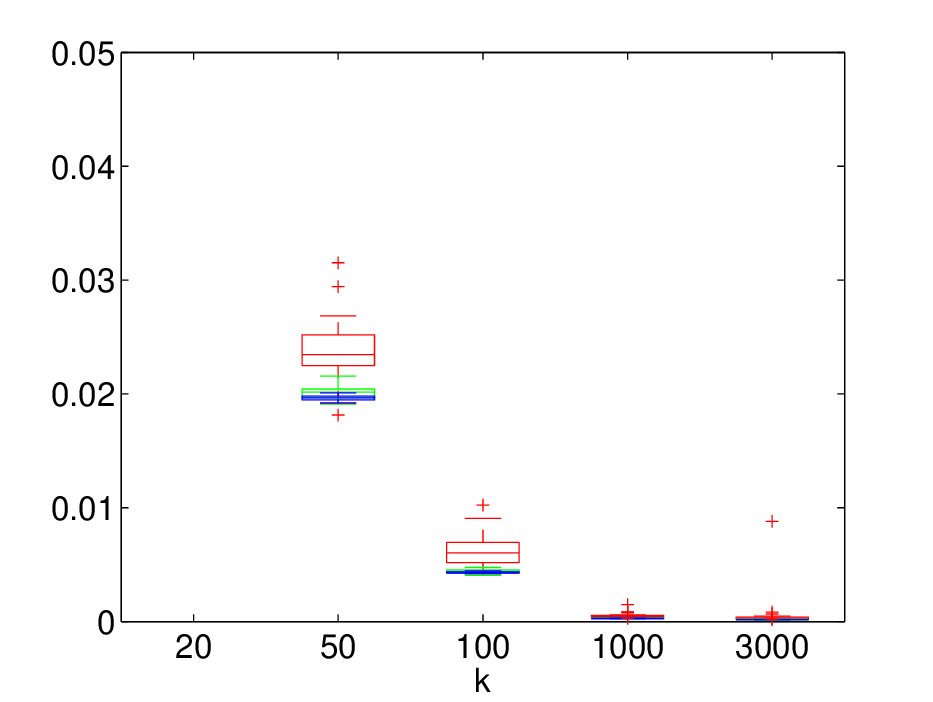}\label{fig.BruitBen-diffFlambda}}
\end{center}
\caption{ Estimation of the mean pattern with the highly correlated noise. Boxplot in red correspond to $J=10$, in green to $J=100$ and in blue to $J=500$. }\label{fig.diffFben}
\end{figure}

\appendix

%\section*{Appendix}

\section{Proofs}

\subsection{Proof of Lemma \ref{lemme.argmin}}

Using the decomposition \eqref{eq:decomp}, we obtain the following identity 
 \begin{equation}\label{eq.decomp2}
M^{\lambda}(\Aa,\balpha,\b) = M^{\lambda}_0(\Aa,\balpha)+\bar M(\Aa,\balpha,\b).
\end{equation}
Note that we have used the fact that the subspaces $\bV^\lambda_0$ and $\bar\bV$ are orthogonal and that $\bar A \1_k\otimes b_j = \1_k\otimes b_j$, for all $b_j\in\R^2$. Let us also introduce the notation $\bar \Y_j =  \frac{1}{k}\sum_{\ell=1}^k [\Y_j]_\ell \in \R^2$, that is $\bar A \Y_j = \1_{k} \otimes \bar\Y_j \in\R^{k\times 2}$ is a degenerated configuration, and $\bar \Y = \frac{1}{J}\sum_{j=1}^J \bar\Y_j \in\R^{2}$. For a fixed $(\Aa,\balpha)\in\R^J\times[-\pi,\pi[^J$, the functional $\b \longmapsto \bar M(\Aa,\balpha,\b)$ vanishes if and only if there exists a $b_0\in\R^2$ such that $e^{-a_j} (\bar A \Y_j- \1_k\otimes b_j)R_{-\alpha_j} = \1_k\otimes b_0 $ for all $ j=1,\ldots,J$. Therefore,  for this fixed $(\Aa,\balpha)$, there is a unique point $\b =\b(\Aa,\balpha) :=(\bar\Y_1-e^{a_1} \bar\Y (\overline{e^{\Aa} R_{\balpha}})^{-1} R_{\alpha_1},\ldots,\bar\Y_J-e^{a_J} \bar\Y (\overline{e^{\Aa} R_{\balpha}})^{-1} R_{\alpha_J})$  with  $\overline{e^{\Aa} R_{\balpha}} = \frac{1}{J} \sum_{j=1}^J e^{a_j}R_{\alpha_j} \in\R^{2\times 2}$ and  %in $\bTheta_{0}^{\b}$
 which satisfies,
$$
\b(\Aa,\balpha) =  \argmin_{\b\in\bTheta^{\b}_0} \bar M(\Aa,\balpha,\b)
$$
%That is 
%$$
%\b(\Aa,\balpha) = .
%$$
 Thence, thanks to the decomposition \eqref{eq.decomp2} and the fact that $\bar M(\Aa,\balpha,\b(\Aa,\balpha)) = 0 $, we have
$$
\argmin_{( \Aa,\balpha,\b)\in\bTheta_0} M^{\lambda}(\Aa,\balpha,\b) = \bigg (\argmin_{(\Aa,\balpha)\in\bTheta^{\Aa,\balpha}_0} M^{\lambda}_0(\Aa,\balpha) \ ,\ \b\Big(\argmin_{(\Aa,\balpha)\in\bTheta^{\Aa,\balpha}_0} M^{\lambda}_0(\Aa,\balpha) \Big)\bigg),
$$
and the claim is proved.
\hfill \qed

\subsection{Proof of Theorem \ref{th:rotscal}}

For all $(\hat \Aa^{\lambda} ,\hat \balpha^{\lambda})\in\bTheta^{\Aa,\balpha}_0$ and $(\Aa^*, \balpha^*)\in[-A,A]\times[-\AA,\AA]$, we have the following inequality 
$$
 \frac{1}{J}\snorm{(\hat \Aa^{\lambda}, \hat \balpha^{\lambda}) -(\Aa^*, \balpha^*)}^2_{\R^{2J}} \leq \frac{2}{J}\snorm{(\hat \Aa^{\lambda}, \hat \balpha^{\lambda}) -(\Aa^*_{\bTheta_0}, \balpha^*_{\bTheta_0})}^2_{\R^{2J}}  + \frac{2}{J}\snorm{(\Aa^*_{\bTheta_0}, \balpha^*_{\bTheta_0})-(\Aa^*, \balpha^*)}^2_{\R^{2J}}.
$$
The proof of Theorem \ref{th:rotscal} is a direct consequence of Proposition \ref{prop.ConvAAlpha} and Lemma \ref{lemme.bernstein} below which control the convergence in probability of the two terms in the right hand side in the preceding inequality.

\begin{proposition}\label{prop.ConvAAlpha}
Consider model \eqref{eq:perturbmodel} and suppose that Assumptions \ref{ass.f} and \ref{ass.zeta} hold and that Assumption \ref{hyp.Theta} is verified with $A,\AA<0.1$. If $\lambda=\lambda(k)= k^{\frac{1}{2s+1}}$ then there exists a constant $C(L,s,A,\AA,\f_0)$ such that for all $x>0$
$$
\P\left( \frac{1}{J}\snorm{(\hat \Aa^{\lambda}, \hat \balpha^{\lambda}) -(\Aa^*_{\bTheta_0}, \balpha^*_{\bTheta_0})}^2_{\R^{2J}} \geq  C(L,s,A,\AA,\f_0)\left(F\big(k^{-\frac{2s}{2s+1}} \big) + F\big(V_1(k,J,x)\big) \right)\right) \leq 2e^{-x},
$$
where $V_1 (k,J,x )= 3\gamma_{\max}(k) k^{-\frac{2s}{2s+1}}\left( 1 + \sqrt{2\frac{x}{Jk}} + 2\frac{x}{Jk} \right) $ and  $F:\R^+ \longrightarrow \R$, with $F(u) = u+\sqrt{u}$.
\end{proposition}
\noindent The proof of Proposition \ref{prop.ConvAAlpha} is postponed to Section \ref{part.proofConvAAlpha}. The following lemma is a direct consequence of Bernstein's inequality for bounded random variables, see e.g. Proposition 2.9 in \cite{MR2319879}.

\begin{lemma}\label{lemme.bernstein}
Suppose that Assumption \ref{hyp.Theta} holds and that the random variables $(a^*_j,\alpha^*_j)$, $j=1,\ldots,J$ have zero expectation in $[-A,A]\times[-\AA,\AA]$. Then, for any $x>0$, we have
$$
\P\left( \frac{1}{J}\snorm{(\Aa^*_{\bTheta_0}, \balpha^*_{\bTheta_0})-(\Aa^*, \balpha^*)}^2_{\R^{2J}} \geq C(A,\AA)\Big(\sqrt{\tfrac{2x}{J}} + \tfrac{x}{3J}\Big)^2\right) \leq 4e^{-x},
$$
where $C(A,\AA) = 4\max\{A^2,\AA^2\}$.\qed
\end{lemma}

\subsection{Proof of Theorem \ref{th:shift}}

The proof of Theorem \ref{th:shift} follows the same guideline as the proof of Theorem \ref{th:rotscal}. Consider the   inequality
$$
\frac{1}{J} \snorm{\hat \b\vphantom{\b}^{\lambda} - \b^* }^2_{\R^{2J}} \leq \frac{2}{J}\snorm{\hat\b\vphantom{\b}^\lambda - \b^*_{\bTheta_0} }^2_{\R^{2J}} + \frac{2}{J}\snorm{\b^*_{\bTheta_0} - \b^*  }^2_{\R^{2J}}.
$$
Theorem \ref{th:shift} is now a direct consequence of Proposition \ref{prop.ConvShift} and Lemma \ref{lemme.bernstein2}.

\begin{proposition}\label{prop.ConvShift}
%Suppose that Assumptions \ref{hyp.zeta} and \ref{hyp.f} hold and that $A,\AA<\delta\leq 0.1$. If $\lambda= k^{\frac{1}{2s+1}}$ then there is a constant 
Under the hypothesis of Proposition \ref{prop.ConvAAlpha}, there exists a constant $C(L,s,A,\AA,\linebreak[1]B,\f)>0$ such that for all $x>0$,
$$
\P \left(\frac{2}{J}\snorm{\hat\b\vphantom{\b}^\lambda-\b_{\bTheta_0}^*}^2_{\R^{2J}} \geq C(L,s,A,\AA,B,\f) \Big( F(k^{-\frac{-2s}{2s+1}}) + F(V_1(k,J,x)) + V_2(k,J,x) \Big) \right) \leq 5e^{-x},
$$
where $ V_1 (k,J,x )= 3\gamma_{\max}(k) k^{-\frac{2s}{2s+1}}\left( 1 + \sqrt{2\frac{x}{Jk}} + 2\frac{x}{Jk} \right)$, $ V_2(k,J,x)=\frac{4}{k}\gamma_{\max}(k) \big(1 +  \sqrt{2\frac{x}{J}} + \frac{x}{J}\big)$ and $F(u) = u +\sqrt{u}$, $u\geq 0$.
\end{proposition}
\noindent The proof of Proposition \ref{prop.ConvShift} is postponed to Section \ref{part.proofprop.ConvShift}.
\begin{lemma}\label{lemme.bernstein2}
Suppose that Assumption \ref{hyp.Theta} holds with  $\AA<\frac{\pi}{4}$ and that the random variables $(a^*_j,\alpha^*_j,b^*_j)$ $j=1,\ldots,J$ have zero expectation in $\Theta\subset\R^{4J}$.  For any $x>0$ , we have
$$
\P\Big( \frac{2}{J}\snorm{\b^*_{\bTheta_0} - \b^* }^2_{\R^{2J}} \geq C(A,\AA,B)\Big(\sqrt{\tfrac{2x}{J}} + \tfrac{x}{3J}\Big)^2\Big ) \leq 4 e^{-x},
$$
where $C(A,\AA,B) = 8B^2 e^{4A}(\cos \AA- \sin \AA)^{-2} $.
\end{lemma}

\begin{proof}
This result is a consequence of the Bernstein's inequality for bounded random variable. We have
$$
\snorm{\b^*_{\bTheta_0} - \b^* }_{\R^{2J}}  = \sqrt{\sum_{j=1}^J \snorm{e^{a_j^*} \bar\b^* (\overline{e^{\Aa}R_{\balpha}})^{-1} R_{\alpha^*_j} }^2_{\R^2} }  
\leq e^A \sqrt{J} \snorm{\bar\b^* (\overline{e^{\Aa^*}R_{\balpha^*}})^{-1} }_{\R^{2}},
$$
where $\overline{e^{\Aa^*}R_{\balpha^*}} = \frac{1}{J}\sum_{j=1}^J e^{a^*_j} R_{\alpha_j^*}$ is an invertible $2\times 2$ matrix whose smallest eigenvalue is greater than $e^{-A} (\cos \AA- \sin \AA )>0$ as $\AA<\frac{\pi}{4}$. To see this, remark that the eigenvalues of $\overline{e^{\Aa^*}R_{\balpha^*}} $ are $\frac{1}{J}\sum_{j=1}^J e^{a^*_j}(\cos\alpha^*_j \pm i \sin\alpha^*_j)$ and we have $\big|\frac{1}{J}\sum_{j=1}^J e^{a^*_j}(\cos\alpha^*_j \pm i \sin\alpha^*_j)\big| \geq e^{-A} (\cos \AA- \sin \AA )>0$. We now have
$$
\frac{1}{\sqrt{J}}\snorm{\b^*_{\bTheta_0} - \b^* }_{\R^{2J}} \leq C(A,\AA) \snorm{\bar\b^*}_{\R^{2}},
$$
where $C(A,\AA) = e^{2A}(\cos \AA- \sin \AA)^{-1} $. Finally, for all $u>0$ we have $\P( \frac{1}{\sqrt{J}}\snorm{\b^*_{\bTheta_0} - \b^* }_{\R^{2J}} \geq u) \leq \P( C(A,\AA) \snorm{\bar \b^*}_{\R^2} \geq u)$ and a Bernstein type inequality (see e.g. Proposition 2.9 in \cite{MR2319879}) gives us
$
\P\Big (\snorm{\bar \b^*}_{\R^2} \geq  2B\big(\sqrt{\tfrac{2x}{J}} + \tfrac{x}{3J}\big) \Big) \leq 4 e^{-x}
$ which yields
$$
\P\Big( \frac{1}{J}\snorm{\b^*_{\bTheta_0} - \b^* }^2_{\R^{2J}} \geq C(A,\AA,B)\Big(\sqrt{\tfrac{2x}{J}} + \tfrac{x}{3J}\Big)^2\Big ) \leq 4 e^{-x},
$$
where $C(A,\AA,B) = 4B^2 e^{4A}(\cos \AA- \sin \AA)^{-2} $.
\end{proof}

\subsection{Proof of Theorem \ref{th:shape}}

Recall the notations introduced Section \ref{part.SIM}: for $\hat g^\lambda_j = (\hat a^\lambda_j,\hat\alpha_j^\lambda,\hat b_j^\lambda) $ we have $\hat g^\lambda_j.\f = e^{\hat a^\lambda_j} \f R_{\hat\alpha_j^\lambda} + \hat b_j^\lambda$. Then, we have 
{\small
\begin{align*}
\frac{1}{k}\Big\|\hat\f\vphantom{\f}^{\lambda} -\f\Big\|^2_{\Sh} & = \frac{1}{k}\Big\|\frac{1}{J}\sum_{j=1}^J (\hat g^{\lambda}_j)^{-1}.(  \A^{\lambda} \Y_j) - \f \Big\|^2_{\Sh}\\
%%& \leq \frac{2}{kJ}\sum_{j=1}^J\snorm{[\hat\btheta^{\blambda}_j]^{-1}.( A_{\blambda}\Y_j)-A_{\blambda}([\btheta^*_{\bTheta_0}]_j^{-1}) . \btheta^*_j \f^* }^2 + f\snorm{A_{\blambda}([\btheta^*_{\bTheta_0}]_j^{-1}) . \btheta^*_j \f^*  -([\btheta^*_{\bTheta_0}]_j^{-1}) . \btheta^*_j \f^*}^2 \\
& \leq \frac{2}{kJ}\sum_{j=1}^J\Big\|\A^{\lambda}\Big ((\hat g^{\lambda}_j)^{-1}. \Y_j-\f\Big) \Big\|^2_{\Sh} +\frac{2}{k}\Big\|\A^{\lambda} \f  -\f\Big\|^2_{\Sh} \\
%& \leq \frac{2}{kJ}\sum_{j=1}^J\Big\|A^{\lambda} \Big ( (\hat g^{\lambda}_j)^{-1}.g^*_j.(\f^*+\bzeta_j)- ([g^*_{\bTheta_0}]_j)^{-1}.g^*_j . \f^* \Big) \Big\|^2 \\ &\qquad+\frac{2}{k}\Big\|e^{\bar \Aa^*}(A^{\lambda}\f^*  -\f^*)R_{\bar\balpha*}\Big\|^2 \\
%%& \leq \underbrace{ \frac{4}{kJ}\sum_{j=1}^J \Big\|(g^{\lambda}_j)^{-1}.g^*_j. \f^* - g_0^{-1}.  \f^*  \Big\|^2 + \Big\|  e^{a_j^* -\hat a_j^{\lambda}} A^{\lambda}\bzeta_j R_{\alpha_j^*-\hat\alpha_j^{\lambda}}\Big\|^2 }_{\mathbf V} \\ &\qquad+\underbrace{\frac{2}{k} \Big\| e^{\bar \Aa^*}(A^{\lambda}\f^*  -\f^*)R_{\bar\balpha*}\Big\|^2}_{\mathbf B}.
& \leq \underbrace{ \frac{4}{kJ}\sum_{j=1}^J \Big\|\A^{\lambda} \Big ((g^{\lambda}_j)^{-1}.g^*_j. \f -  \f\Big) \Big\|^2 _{\Sh}+ \Big\|  e^{a_j^* -\hat a_j^{\lambda}} \A^{\lambda}\bzeta_j R_{\alpha_j^*-\hat\alpha_j^{\lambda}}\Big\|^2_{\Sh} }_{\mathbf V} 
%\\ &\qquad \qquad \qquad 
+\underbrace{\vphantom{\frac{4}{kJ}\sum_{j=1}^J}\frac{2}{k} \Big\| \A^{\lambda}\f  -\f\Big\|^2_{\Sh}}_{\mathbf B}.
%%& \leq \underbrace{ \frac{4}{kJ}\sum_{j=1}^J \Big\|A^{\lambda} \Big ((g^{\lambda}_j)^{-1}.g^*_j. \f^* - e^{-\bar\Aa^*}  \f^* R_{-\bar\balpha^*} + \bar\b^* (\overline{e^{\Aa^*}R_{\balpha^*}})^{-1}\Big) \Big\|^2 + \Big\|  e^{a_j^* -\hat a_j^{\lambda}} A^{\lambda}\bzeta_j R_{\alpha_j^*-\hat\alpha_j^{\lambda}}\Big\|^2 }_{\mathbf V} \\ &\qquad+\underbrace{\frac{2}{k} \Big\| e^{\bar \Aa^*}(A^{\lambda}\f^*  -\f^*)R_{\bar\balpha*}\Big\|^2}_{\mathbf B}.
\end{align*}
}
The rest of the proof is devoted to control the terms $\mathbf{B}$ and $\mathbf{V}$. The term $\mathbf B$ is controlled by the bias of the low pass filter. According to Lemma \ref{lemme.biais}, if Assumption \ref{ass.f} holds and by choosing the optimal frequency cutoff $\lambda = \lambda(k) = k^{\frac{1}{2s+1}}$, there exists a constant $ C(L,s) >0$ such that
% \begin{equation}\label{eq.BiasVar}
%$\frac{1}{k}\snorm{ \A^{\lambda} \f  - \f}^2 \leq C(L) (k^{2-2s} + k^{-\frac{2s}{2s+1}}) \leq C(L) k^{-\frac{2s}{2s+1}}$  % \qquad \text{ and } \qquad \frac{1}{k}\E\big[\snorm{ \A^{\lambda} \bzeta}^2\big] \leq 3\gamma_{\max}(k) k^{-\frac{2s}{2s+1}},
%\end{equation}
\begin{equation}\label{eq.conB}
\mathbf B = \frac{2}{k} \| \A^{\lambda}\f  -\f\|^2_{\Sh}% \leq \frac{2}{k} e^{2A}  \snorm{A^{\lambda}\f -\f}^2 
\leq C(L,s) k^{-\frac{2s}{2s+1}}.
\end{equation}
%The last inequality is true if, for example, $s\geq\frac{3}{2}$.

%where $ C(L,A) =2e^{2A}  \kappa(L)$.
The term $\mathbf V$ contains two expressions. To bound the first one we use Bessel's inequality and Lemma \ref{lemme.AssOp}. More precisely, we have for all $j=1,\ldots,J$
\begin{align*}
\frac{1}{k}\big\|\A^{\lambda} \big ((g^{\lambda}_j)^{-1}.g^*_j. \f - \f \big) \big\|^2_{\Sh} & \leq \frac{1}{k}\big\| (g^{\lambda}_j)^{-1}.g^*_j. \f-\f \big\|^2_{\Sh} \\ 
& \leq C(A,\f) \big\| (a^*_j - \hat a^\lambda_j, \alpha^*_j- \hat \alpha^\lambda_j, e^{\hat a^\lambda_j}( b^*_j - \hat b\vphantom{b}^\lambda_j) R_{- \hat \alpha_j^\lambda} )\big \|^2_{\R^{4}} \\
& \leq  C(A,\f) \big\| (a^*_j - \hat a^\lambda_j, \alpha^*_j- \hat \alpha^\lambda_j)\big\|^2_{\R^{2}} + C(A,\f) \big\| b^*_j - \hat b\vphantom{\b}^\lambda_j \big\|^2_{\R^{2}}
\end{align*}
We can now use Theorem \ref{th:rotscal} and \ref{th:shift} to derive the following concentration inequality,
\begin{equation}\label{eq.temp2}
\begin{split}
\P\bigg(\frac{4}{kJ}\sum_{j=1}^J  \big\|\A^{\lambda} \big ((&g^{\lambda}_j)^{-1}.g^*_j. \f - \f \big) \big\|^2_{\Sh}  \\ & \geq  C(L,s,A,\AA,B,\f)\Big( A_1(k,J,x) + A_3(k,J,x) + A_2(J,x) \Big)\bigg)\leq 13e^{-x}
\end{split}
\end{equation}
where $C(L,s,A,\AA,B,\f)>0$ is a constant independent of $k$ and $J$ and $A_1, A_2, A_3$ are defined in the statement of Theorem \ref{th:rotscal} and \ref{th:shift}.
%Remark also that $(\bar\Aa^*,\bar\balpha^* , e^{\bar\Aa} (\bar\b^*(\overline{e^{\Aa^*}R_{\balpha}})^{-1})R_{\bar\balpha^*}) = (a_0, \alpha_0,b_0)^{-1}$ defined in equation \eqref{eq.g0}. 
%The first part of the term $\mathbf V$ is treated Lemma \ref{lemme.AssOp} and for all $j=1,\ldots,J$ we have $ \frac{1}{k}\snorm{([\btheta^{\blambda}_j]^{-1}\btheta^*_j). \f^* - \btheta_0.  \f^* }^2 \leq C(A,\f^*) \snorm{[\btheta^{\blambda}_j]^{-1}\btheta^*_j -  \btheta_0}^2$ for some constant $ C(A,\f^*) >0$. 
The second term contained in $\mathbf V$ is treated by equation \eqref{eq:conNoise} below. Hence, formulas  \eqref{eq.temp2} and \eqref{eq:conNoise} yield 
\begin{equation}\label{eq.conV}
\P\bigg( \mathbf{V}\geq  C(L,s,A,\AA,B,\f)\Big( A_1(k,J,x) + A_3(k,J,x) + A_2(J,x) \Big) \bigg) \leq 14e^{-x},
\end{equation}
for some constant $C(L,s,A,\AA,B,\f)>0$. Putting together equations \eqref{eq.conB} and \eqref{eq.conV} gives
$$
\P\bigg( \frac{1}{k}\snorm{\hat \f\vphantom{\f}^\lambda - \f}^2_{\Sh}\geq  C(L,s,A,\AA,B,\f)\Big( A_1(k,J,x) + A_3(k,J,x) + A_2(J,x) \Big) \bigg) \leq 14e^{-x}.
$$
for some constant $C(L,s,A,\AA,B,\f)>0$. The proof of Theorem \ref{th:shape} is completed. \qed

\subsection{Proof of Proposition \ref{prop.ConvAAlpha}}
\label{part.proofConvAAlpha}

The mean pattern $\f$ can be decomposed as $\f = \bar \f + \f_0 \in \one_{k\times 2}^{\vphantom{\perp}} \oplus \one_{k\times 2}^\perp$. Then, $\f_0$ is the centered version of $\f$ and we can consider the criterion,
$$
D_{0} (\Aa,\balpha)=  \frac{1}{Jk}\sum_{j=1}^J \bnorm{ e^{a^*_{j}-a_{j}} \f_0  R_{\alpha^*_{j}-\alpha_{j}} - \frac{1}{J} \sum_{j'=1}^J  e^{a^*_{j'}-a_{j'}} \f_0  R_{\alpha^*_{j'} - \alpha_{j'} } }^2_{\Sh}.
$$
We now have,
$$
(\hat \Aa^{\lambda} , \hat \balpha^{\lambda}) = \argmin_{(\Aa,\balpha)\in \bTheta_0^{\Aa,\balpha}} M_0^{\lambda}(\Aa,\balpha)\qquad\text{ and }\qquad (\Aa^{*}_{\bTheta_0},\balpha^{*}_{\bTheta_0}) = \argmin_{(\Aa,\balpha)\in \bTheta_0^{\Aa,\balpha}} D_0^{}(\Aa,\balpha).
$$ 
Then, the convergence of $(\hat \Aa^{\lambda} , \hat \balpha^{\lambda})$  to $(\Aa^{*}_{\bTheta_0} ,\balpha^{*}_{\bTheta_0})$ is guaranteed if $(\Aa^{*}_{\bTheta_0} ,\balpha^{*}_{\bTheta_0})$ is uniquely defined and if there is a uniform convergence in probability of $M_0^{\lambda} $ to $D_0$, see e.g. \cite{VdW}. This is the aim of Lemmas \ref{lemme.unic} and  \ref{lemme.unif} below.

\begin{lemma}\label{lemme.unic}
Let $\f$ be a non-degenerated configuration in $\R^{k \times 2}$, \ie $\f \notin \one_{k\times 2}$. Then, the argmin of $D_0$ on $\bTheta_0^{\Aa,\balpha}$ is unique and denoted by $(\Aa_{\bTheta_0}^*,\balpha_{\bTheta_0}^*) = (a_j - \bar\Aa^*,\alpha_j -\bar\balpha^*)$, where $\bar{\Aa}^* = \frac{1}{J}\sum_{j=1}^J a_j^* $, $\bar \balpha^* = \frac{1}{J}\sum_{j=1}^J \alpha_j^*$.
\end{lemma}
\begin{proof}
As $\f$ is a non-degenerated configuration, we have $\f_0 \neq 0$. Thus, the stabilizer of $\f$ is reduced to the identity, see Section \ref{part.SIM}. Then $D_0(\Aa,\balpha)=0$  if and only if there exists $(a_0,\alpha_0)\in\R^2$ such that $(\Aa,\balpha) = (\Aa^*,\balpha^*)*(a_0,\alpha_0) = (a_1^* +a_0,\alpha^*_1 +\alpha_0, \ldots,a_J+a_0,\alpha_J^*+\alpha_0)$. By choosing $(a_0,\alpha_0) = (-\bar{\Aa}^* ,- \bar \balpha^*)$ we have $\sum_{j=1}^J (a^*_j,\alpha^*_j)*(a_0,\alpha_0) = 0 $. That is  $(\Aa^*_{\bTheta_0},\balpha^*_{\bTheta_0})= (\Aa^*,\balpha^*)*(a_0,\alpha_0)\in \bTheta_0^{\Aa,\balpha}$.
\end{proof}
\noindent We now show the uniform convergence in probability, 
\begin{lemma}\label{lemme.unif}
Suppose that Assumptions  \ref{ass.f}, \ref{ass.zeta} and \ref{hyp.Theta} hold and let $F:\R \longrightarrow \R$, with $F(u) = u+\sqrt{u}$. For any $x>0$ we have
$$
\P\bigg( \sup_{(\Aa,\balpha)\in \bTheta_0^{\Aa,\balpha}} \big |M^{\lambda}_0(\Aa,\balpha) - D_0(\Aa,\balpha)\big|\geq C(L,s,A,\f_0) \left (F\big( k^{-\frac{2s}{2s+1}}  \big) + F\big(V(k,J,x)\big)\right)\bigg) \leq 2e^{-x}
$$
where $C(L,s,A,\f_0) = e^{2A} \max\left\{ \frac{2}{\sqrt{k}}\snorm{\f_0}_{\Sh}, \frac{\sqrt{2}e^A}{\sqrt{k}}\snorm{\f_0}_{\Sh},2  \right\} $ and $V (k,J,x )= 3\gamma_{\max}(k) k^{-\frac{2s}{2s+1}} \big( 1 + F \left(\frac{2x}{Jk}\right) \big) $.
%$C(L,A,\f) = 2\kappa(L)e^{-2A}\max\left\{1, \tfrac{2}{\sqrt k}\|\f^*_0\|\right\} $  and $V (k,J,x )= \gamma_{\max}(k) k^{-\frac{2s}{2s+1}} \left( 1 + F \left(2\frac{x}{Jk}\right)\right)$.%$V (k,J,x )= \gamma_{\max}(k) \frac{2\lambda_k+1}{k}\left( 1 + \sqrt{2\frac{x}{Jk}} + 2\frac{x}{Jk} \right) $. 
\end{lemma}

\begin{proof}
Let us write the following decomposition,
 \begin{align}
&M^{\lambda}_{0}(\Aa,\balpha) \nonumber\\
& =\frac{1}{Jk} \sum_{j=1}^J \bnorm{ e^{a^*_{j} -a_{j}}\f_0 R_{\alpha^*_{j} - \alpha_{j}} - \frac{1}{J} \sum_{j'=1}^J e^{a^*_{j'} -a_{j'}} \f_0 R_{\alpha^*_{j'} - \alpha_{j'}}}^2_{\Sh}\label{eq:D0} \\
&\qquad\qquad +\frac{1}{Jk} \sum_{j=1}^J \bnorm{ e^{a^*_{j} -a_{j}}(A^{\lambda}_0 \f -\f_0  ) R_{\alpha^*_{j} - \alpha_{j}}  - \frac{1}{J} \sum_{j'=1}^J e^{a^*_{j'} -a_{j'}} (A^{\lambda}_0 \f-\f_0  ) R_{\alpha^*_{j'} - \alpha_{j'}} }^2_{\Sh} \label{eq:QB}\\
&\begin{aligned}\qquad\qquad + \frac{2}{Jk} \sum_{j=1}^J \bigg\langle &e^{a^*_{j} -a_{j}}\f_0 R_{\alpha^*_{j} - \alpha_{j}}  - \frac{1}{J} \sum_{j'=1}^J e^{a^*_{j'} -a_{j'}} \f_0 R_{\alpha^*_{j'} - \alpha_{j'}}  ,\\&e^{a^*_{j} -a_{j}}(A^{\lambda}_0 \f -\f_0) R_{\alpha^*_{j} - \alpha_{j}}  - \frac{1}{J} \sum_{j'=1}^J e^{a^*_{j'} -a_{j'}}(A^{\lambda}_0 \f-\f_0  ) R_{\alpha^*_{j'} - \alpha_{j'}} \bigg\rangle_{\Sh}\end{aligned}\label{eq:LB}\\
&\qquad\qquad +\frac{1}{Jk} \sum_{j=1}^J \bnorm{e^{a^*_{j} -a_{j}} A^{\lambda}_0 \bzeta_{j} R_{\alpha^*_{j} - \alpha_{j}}  - \frac{1}{J} \sum_{j'=1}^J e^{a^*_{j'} -a_{j'}}A^{\lambda}_0 \bzeta_{j'}R_{\alpha^*_{j'} - \alpha_{j'}} }^2_{\Sh} \label{eq:QZ} \\
&\begin{aligned}\qquad\qquad + \frac{2}{Jk} \sum_{j=1}^J \bigg\langle & e^{a^*_{j} -a_{j}} A^{\lambda}_0  \f R_{\alpha^*_{j}- \alpha_{j}}- \frac{1}{J} \sum_{j'=1}^J e^{a^*_{j'} -a_{j'}}A^{\lambda}_0  \f R_{\alpha^*_{j'} - \alpha_{j'}} , \\&e^{a^*_{j} -a_{j}} A^{\lambda}_0 \bzeta_{j} R_{\alpha^*_{j} - \alpha_{j}}  - \frac{1}{J} \sum_{j'=1}^J e^{a^*_{j'} -a_{j'}}A^{\lambda}_0 \bzeta_{j'} R_{\alpha^*_{j'} - \alpha_{j'}} \bigg\rangle_{\Sh}\end{aligned}\label{eq:LZ}
\end{align}
Then, criterion $M^{\lambda}_0$ is viewed as a perturbed version of the criterion $D_0(\Aa,\balpha) = \eqref{eq:D0}$, 
$$
M^{\lambda}_0(\Aa,\balpha) = D_0(\Aa,\balpha) + \mathbf B^{\lambda}_0(\Aa,\balpha) + \mathbf V^{\lambda}_0(\Aa,\balpha)
$$
where the bias term is $\mathbf B^{\lambda}_0(\Aa,\balpha) = \eqref{eq:QB} + \eqref{eq:LB}$  %will be controlled in term of the bias of the smoother 
and the variance term is $\mathbf V^{\lambda}_0(\Aa,\balpha) = \eqref{eq:QZ} + \eqref{eq:LZ} $.

\paragraph*{The bias term. } We have for any $(\Aa,\balpha)\in\bTheta_0^{\Aa,\balpha}$, 
\begin{align*}
 \eqref{eq:QB} & = \frac{1}{Jk} \sum_{j=1}^J \bnorm{  e^{a^*_{j} -a_{j}}(A^{\lambda}_0 \f -\f_0  ) R_{\alpha^*_{j} - \alpha_{j}} }^2_{\Sh}  -\bnorm{ \frac{1}{J} \sum_{j'=1}^J e^{a^*_{j'} -a_{j'}}(A^{\lambda}_0 \f -\f_0  ) R_{\alpha^*_{j'} - \alpha_{j'}}  }^2_{\Sh} \\& \leq e^{2A}\frac{1}{k}\bnorm{A^{\lambda}_0 \f- \f_0 }^2_{\Sh}.
\end{align*}
A double application of Cauchy-Schwarz inequality implies that,
\begin{align*} 
&\begin{aligned} |\eqref{eq:LB}| \leq \frac{2}{Jk} \sum_{j=1}^J  \bigg\|e^{a^*_{j} -a_{j}}\f_0 &R_{\alpha^*_{j} - \alpha_{j}}  - \frac{1}{J} \sum_{j'=1}^J e^{a^*_{j'} -a_{j'}} \f_0 R_{\alpha^*_{j'} - \alpha_{j'}}   \bigg\|_{\Sh}\\ 
 &\bnorm{e^{a^*_{j} -a_{j}}(\A^{\lambda}_0 \f -\f_0) R_{\alpha^*_{j} - \alpha_{j}}  - \frac{1}{J} \sum_{j'=1}^J e^{a^*_{j'} -a_{j'}}(\A^{\lambda}_0 \f-\f_0  ) R_{\alpha^*_{j'} - \alpha_{j'}}}_{\Sh} \end{aligned}\\
%& \leq \frac{2}{J}\left( \ssum_{j=1}^J \bnorm{e^{a^*_{j} -a_{j}}\f^*_0 R_{\alpha^*_{j} - \alpha_{j}}  - \frac{1}{J} \ssum_{j'=1}^J e^{a^*_{j'} -a_{j'}} \f^*_0 R_{\alpha^*_{j'} - \alpha_{j'}} }^2 \right)^{\tfrac{1}{2}} \left(\ssum_{j=1}^J\bnorm{e^{a^*_{j} -a_{j}}(A_{\lambda} \f^*_0 -\f^*_0) R_{\alpha^*_{j} - \alpha_{j}}  - \frac{1}{J} \ssum_{j'=1}^J e^{a^*_{j'} -a_{j'}}(A_{\lambda} \f^*_0-\f^*_0  ) R_{\alpha^*_{j'} - \alpha_{j'}}}^2 \right)^{\tfrac{1}{2}}\\
& \phantom{\eqref{eq:LB}} \leq \frac{2}{k} \bigg(\frac{1}{J} \sum_{j=1}^J e^{2A}\snorm{\f_0 R_{\alpha^*_{j} - \alpha_{j}}}^2_{\Sh} \bigg)^{\tfrac{1}{2}} \bigg(\frac{1}{J}\sum_{j=1}^Je^{2A}\snorm{(A^{\lambda}_0 \f -\f_0  ) R_{\alpha^*_{j} - \alpha_{j}}}^2_{\Sh} \bigg)^{\tfrac{1}{2}}\\
&\phantom{\eqref{eq:LB}}  = 2e^{2A}\frac{1}{k}\snorm{\f_0  }_{\Sh}\snorm{A^{\lambda}_0 \f -\f_0}_{\Sh}.
\end{align*} 
Finally, by using Lemma \ref{lemme.biais} we have,
$$
\sup_{(\Aa,\balpha)\in\bTheta^{\Aa,\balpha}} |\mathbf B^{\lambda}_0(\Aa,\balpha) |\leq C_1(L,s,A,\f) \left( k^{-\frac{2s}{2s+1}} +k^{-\frac{s}{2s+1}}\right).
$$
where $ C_1(L,s,A,\f) = C(L,s) e^{2A} \max\left\{2\tfrac{1}{\sqrt k}\snorm{\f_0  }_{\Sh},1\right\} $.

\paragraph*{The variance term.}  First, the term $\eqref{eq:LZ}$ is by the Cauchy-Schwarz inequality controlled by $2e^{2A}\frac{1}{\sqrt{k}}\linebreak[1] \snorm{\f_0  }_{\Sh}\sqrt{\eqref{eq:QZ}}$. The term \eqref{eq:QZ} is bounded from above by $\frac{1}{Jk}\sum_{j=1}^{J} e^{2A} \snorm{A^{\lambda}_0 \bzeta_j }^2_{\Sh}$. To derive an upper bound in probability, note that we have the following equality in  law,
$$
\sum_{j=1}^{J} \snorm{A^{\lambda}_0 \bzeta_j}^2_{\Sh} =\xi'\,\mathbf{B}\,\xi,
$$
with $\mathbf{B} = \big[Id_J \otimes \mathbf \Sigma^{\frac{1}{2}}\big]\big[Id_{2J} \otimes (A^{\lambda}_0)'A^{\lambda}_0\big]\big [Id_J \otimes \mathbf\Sigma^{\frac{1}{2}}\big] \in \R^{2Jk\times 2Jk}$  and $\xi=(\xi_1,\ldots,\xi_{2Jk})'$ is a centered Gaussian vector of variance $Id_{2Jk}$. We have $\tr (Id_{J} \otimes\mathbf\Sigma )\leq 2Jk  \gamma_{\max}(k)$ and  $\tr ((A^{\lambda}_0)' A^{\lambda}_0) = \frac{2\lambda+1}{k}$. Using a classical concentration inequality for quadratic form of multivariate Gaussian random variables, see e.g. \cite{LauMa} Lemma 1, we have for all $x>0$, %\frac{e^{2A}}{Jk}
$
\P\Big( \xi'\mathbf{B}\xi \geq 2Jk \gamma_{\max}(k) \frac{2\lambda+1}{k} +  2\gamma_{\max}(k) \frac{2\lambda+1}{k}\sqrt{x 2Jk} + 4x \gamma_{\max}(k) \frac{2\lambda+1}{k} \Big) \leq e^{-x},
$
which yields together with formula \eqref{eq.conB},
\begin{align}\label{eq:conNoise}
\P\bigg( \frac{1}{Jk}\sum_{j=1}^{J} e^{2A} \snorm{A^{\lambda}_0 \bzeta_j }^2_{\Sh}\geq 6e^{2A} \gamma_{\max}(k) k^{-\frac{2s}{2s+1}} \Big( 1 + \sqrt{2 \frac{x}{Jk}} + 2\frac{x}{Jk}\Big) \bigg) \leq e^{-x}.
\end{align}
Hence we have, 
\begin{align*}
\P\bigg( \sup_{(\Aa,\balpha)\in\bTheta^{\Aa,\balpha}} &|\mathbf V^{\lambda}_0(\Aa,\balpha) +\mathbf B^{\lambda}_0 (\Aa,\balpha)| \\&\geq C(L,s,A,\f_0)\left(k^{-\frac{2s}{2s+1}} + k^{-\frac{s}{2s+1}}+V_1  (k,J,x )+ \sqrt{V_1 (k,J,x )}  \right) \bigg)  \leq 2e^{-x},
%\P\left( \sup_{\btheta\in\bTheta} \mathbf V^{\lambda}_0(\btheta) + \sup_{\btheta\in\bTheta} \mathbf B^{\lambda}_0 (\btheta) \geq e^{2A}\max\left\{\tfrac{2}{\sqrt k}\snorm{\f^*_0  },1\right\} ( C_1(A) k^{-\frac{2s}{2s+1}} +\sqrt{C_1(A)} k^{-\frac{2s}{4s+2}} +V  (k,J,x )+ \sqrt{V (k,J,x )}  ) \right)  \leq e^{-x},
%\P\left( \sup_{\btheta\in\bTheta} \mathbf V^{\lambda}_0(\btheta) + \sup_{\btheta\in\bTheta} \mathbf B^{\lambda}_0 (\btheta) \geq e^{2A}\max\{2\snorm{\f^*_0} ,1\} (V + \sqrt{V} + B(\lambda) + \sqrt{B(\lambda)}) \right)  \leq e^{-x},
\end{align*}
where $C(L,s,A,\f_0) = e^{2A} \max\big\{ \frac{2}{\sqrt{k}}\snorm{\f_0}_{\Sh}, \frac{1}{\sqrt{k}}\snorm{\f_0}_{\Sh}\sqrt{2}e^A,2  \big\} $ and $V_1 (k,J,x )= 3\frac{\gamma_{\max}(k)}{k^{\frac{2s}{2s+1}}} \big( 1 + \sqrt{2\frac{x}{Jk}} + 2\frac{x}{Jk} \big) $.
\end{proof}

For a fixed $J\in\N$, the convergence of the M-estimator $(\hat \Aa^{\lambda},\hat \balpha^{\lambda})$ to $(\Aa_{\bTheta_0}^*,\balpha_{\bTheta_0}^*)$ when $k\to\infty$ is guaranteed by Lemma \ref{lemme.unic} and \ref{lemme.unif}, see e.g. \cite{VdW}. Nevertheless, we are able to  give a rate of convergence and non-asymptotic bounds in $k$ and $J$ by using the classical inequality,
$$
 |D_0(\hat \Aa^{\lambda}, \hat \balpha^{\lambda}) - D_0(\Aa_{\bTheta_0}^*,\balpha_{\bTheta_0}^*)| \leq 2 \sup_{(\Aa,\balpha)\in\bTheta_0^{\Aa,\balpha}} | D_0(\Aa,\balpha) - M_0^{\lambda}(\Aa,\balpha)|.
$$
This, together with Lemma \ref{lemme.MinQuad} below will prove Proposition \ref{prop.ConvAAlpha}.
\begin{lemma}\label{lemme.MinQuad}
Assume that $A,\AA<0.1$. There exists a constant $C(A,\AA)>0$ independent of $J$ such that for all $(\Aa,\balpha) \in \bTheta^{\Aa,\balpha}$ we have
$$
\abs{D_0(\Aa,\balpha) - D_0(\Aa_{\bTheta_0}^*,\balpha_{\bTheta_0}^*)}\geq C(A,\AA,\f_0)\frac{1}{J}\norm{ (\Aa-\Aa_{\bTheta_0}^* ,\balpha-\balpha_{\bTheta_0}^*)}^2_{\R^{2J}},
$$
where $C(A,\AA,\f_0)=C(A,\AA)\frac{1}{k}\snorm{\f_0}^2_{\Sh}$.
\end{lemma}

\begin{proof}
By definition, given a $(\Aa^*,\balpha^*)\in[-A,A]^J\times[-\AA,\AA]^J$, the point $(\Aa_{\bTheta_0}^*,\balpha_{\bTheta_0}^*)$ is the unique minimum of $D_0$ on $\bTheta_0^{\Aa,\balpha} = [-A,A]^J\cap \1_J^\perp \times[-\AA,\AA]^J\cap\1_J^\perp$. %Then, for all $\delta>0 $ and $(\Aa,\balpha) \in \mathcal{N}_{(\Aa_{\bTheta_0}^*,\balpha_{\bTheta_0}^*)}(\delta)$, there exists a $\mathbf c = \mathbf c(\Aa,\balpha)\in\mathcal{N}_{(\Aa_{\bTheta_0}^*,\balpha_{\bTheta_0}^*)}(\delta) $ such that  the  Taylor expansion of $D_0$ at $(\Aa_{\bTheta_0}^*,\balpha_{\bTheta_0}^*)$ can be written, 
Then, for all $(\Aa,\balpha) \in\bTheta_0^{\Aa,\balpha}$, there exists a $\mathbf c = \mathbf c(\Aa,\balpha)\in\bTheta_0^{\Aa,\balpha}  $ such that  the  Taylor expansion of $D_0$ at $(\Aa_{\bTheta_0}^*,\balpha_{\bTheta_0}^*)$ can be written, 
\begin{align*}
D_0(\Aa,\balpha) - D_0(\Aa_{\bTheta_0}^*,\balpha_{\bTheta_0}^*) & = \frac{1}{2} (\Aa-\Aa_{\bTheta_0}^* ,\balpha-\balpha_{\bTheta_0}^*)' [\nabla^{2} D_0(\Aa_{\bTheta_0}^*,\balpha_{\bTheta_0}^*)](\Aa-\Aa_{\bTheta_0}^* ,\balpha-\balpha_{\bTheta_0}^*)  \\ & \qquad\qquad + \frac{1}{6} [ \nabla^3 D_0(\mathbf c)] (\Aa-\Aa_{\bTheta_0}^* ,\balpha-\balpha_{\bTheta_0}^*).
\end{align*}
Let $\delta = \max\{A,\AA\}$. This, together with Lemma \ref{lemme.eigmin} and \ref{lem:D3} imply that 
\begin{align*}
D_0(\Aa,\balpha) - D_0(\Aa_{\bTheta_0}^*,\balpha_{\bTheta_0}^*) &\geq\frac{1}{2} (\Aa-\Aa_{\bTheta_0}^* ,\balpha-\balpha_{\bTheta_0}^*)' [\nabla^{2} D_0(\Aa_{\bTheta_0}^*,\balpha_{\bTheta_0}^*)](\Aa-\Aa_{\bTheta_0}^* ,\balpha-\balpha_{\bTheta_0}^*)\\ & \qquad\qquad  -   \delta\frac{40}{6} e^{2A} \snorm{\f_0}_{\Sh}^2 \frac{1}{Jk}\norm{(\Aa-\Aa_{\bTheta_0}^* ,\balpha-\balpha_{\bTheta_0}^*)}^2_{\R^{2J}}\\
& \geq  \snorm{\f_0}^2_{\Sh}\frac{1}{Jk}\norm{(\Aa-\Aa_{\bTheta_0}^* ,\balpha-\balpha_{\bTheta_0}^*)}^2_{\R^{2J}} \left( e^{-2A}-  \delta\frac{40}{6}e^{2A}  \right).
\end{align*}
 Hence, one can choose $\delta >0$ sufficiently small such that $\left( e^{-2A}-  \delta\frac{40}{6}e^{2A} \right)$ is strictly positive for all $J$ and $k$. For example, we have $\left( e^{-2\delta}-  \delta\frac{40}{6}e^{2\delta} \right)>0$, if $\delta<0.1$. Then, using such a $\delta$ it follows that for all $(\Aa,\balpha) \in \bTheta^{\Aa,\balpha}_0$,
\begin{equation*}
\abs{D_0(\Aa,\balpha) - D_0(\Aa_{\bTheta_0}^*,\balpha_{\bTheta_0}^*)}\geq C(A,\AA)\frac{1}{k}\snorm{\f_0}^2_{\Sh}\frac{1}{J}\norm{ (\Aa-\Aa_{\bTheta_0}^* ,\balpha-\balpha_{\bTheta_0}^*)}^2_{\R^{2J}}.\tag*{\qedhere}
\end{equation*}
\end{proof}

The proof of Proposition \ref{prop.ConvAAlpha} is almost done. Remark that Lemma \ref{lemme.MinQuad} ensures that for all $u\geq 0$ we have
$
\P\big( \frac{1}{J}\snorm{(\hat \Aa^{\lambda}, \hat \balpha^{\lambda}) -(\Aa^*_{\bTheta_0}, \balpha^*_{\bTheta_0})}^2_{\R^{2J}} \geq u \big)  \leq \P\big( \frac{2}{C(A,\AA,\f_0)} \sup_{(\Aa,\balpha)\in\bTheta_0^{\Aa,\balpha} } \abs{M^{\lambda}_0(\Aa,\balpha) - D_0(\Aa,\balpha)}\geq u\big)
$. 
Lemma \ref{lemme.unif} ensures that there is a constant $C(L,s,A,\AA,\f_0)$ such that for all $x>0$,
$$
\P\Big( \frac{1}{J}\snorm{(\hat \Aa^{\lambda}, \hat \balpha^{\lambda}) -(\Aa^*_{\bTheta_0}, \balpha^*_{\bTheta_0})}^2_{\R^{2J}} \geq  C(L,s,A,\AA,\f_0)(F(k^{-\frac{2s}{2s+1}} ) + F(V_1(k,J,x)) \Big) \leq 2e^{-x}.
$$
where $V_1 (k,J,x )= \gamma_{\max}(k)k^{-\frac{2s}{2s+1}} \left( 1 + \sqrt{2\frac{x}{Jk}} + 2\frac{x}{Jk} \right) $ and  $F:\R \longrightarrow \R$, with $F(u) = u+\sqrt{u}$. \qed

\subsection{Proof of Proposition \ref{prop.ConvShift}}
\label{part.proofprop.ConvShift}

First of all remark that, thanks to formulas \eqref{eq.g0} and  \eqref{eq.blambda}, we have  explicit expressions of $\b^*_{\bTheta_0} = \argmin_{\b\in\bTheta_0^{\b}} \bar D(\Aa^*,\balpha^*,\b) $ and of  $ \hat \b\vphantom{\b}^\lambda_{} = \argmin_{\b\in\bTheta_0^{\b}} \bar M(\hat \Aa^\lambda,\hat \balpha^\lambda,\b) $. Indeed, we have,
$$ 
\b^*_{\bTheta_0} = (b_1^* -e^{a_1^*} \bar\b^* (\overline{e^{\Aa^*} R_{\balpha^*}})^{-1} R_{\alpha_1^*},\ldots, b^*_J-e^{a_J^*} \bar\b^* (\overline{e^{\Aa^*} R_{\balpha^*}})^{-1} R_{\alpha_J^*}),
$$
where  $\overline{e^{\Aa^*} R_{\balpha^*}} = \frac{1}{J} \sum_{j=1}^J e^{a_j^*}R_{\alpha_j^*} \in\R^{2\times 2}$, $\bar\b^* = \frac{1}{J} \sum_{j=1}^J b_j^*\in\R^2$ and
$$
\hat \b\vphantom{\b}^{\lambda} = (\bar\Y_1-e^{\hat a_1^{\lambda}} \bar\Y (\overline{e^{\hat \Aa^{\lambda}} R_{\hat \balpha^{\lambda}}})^{-1} R_{\hat \alpha_1^{\lambda}},\ldots,\bar\Y_J-e^{\hat a_J^{\lambda}} \bar\Y (\overline{e^{\hat \Aa^{\lambda}} R_{\hat \balpha^{\lambda}}})^{-1} R_{\hat \alpha_J^{\lambda}})
$$
where $\overline{e^{\hat \Aa^{\lambda}} R_{\hat \balpha^{\lambda}}} = \frac{1}{J} \sum_{j=1}^J e^{\hat a_j^\lambda}R_{\hat \alpha_j^\lambda} \in\R^{2\times 2}$,  
$
\bar\Y_j   = \frac{1}{k}\sum_{\ell=1}^k[\Y_j]_\ell %= \frac{1}{k}\sum_{\ell=1}^k e^{a_j^*} (\f_\ell + \bzeta_\ell)R_{\alpha^*_j} +  b_j^*
 = e^{a^*_j} \bar\f R_{\alpha^*_j} + b_j^* + e^{a^*_j} \bar\bzeta_j R_{\alpha^*_j} \in\R^2
$
where $\bar\f =\frac{1}{k}\sum_{\ell=1}^k  \f_\ell \in\R^2$, $\bar\bzeta_j =\frac{1}{k}\sum_{\ell=1}^k  [\bzeta_j]_\ell  \in\R^2 $ and  
$
\bar \Y = \frac{1}{J}\sum_{j=1}^J \bar\Y_j = \bar\f(\overline{e^{\Aa^*} R_{\balpha^*}}) + \bar\b^* +\frac{1}{J}\sum_{j=1}^J e^{a_j^*} \bar\bzeta_j R_{\alpha_j^*} \in\R^2.
$
Thence, we have
\begin{align}
\frac{1}{J}\snorm{\b^*_{\bTheta_0} - \hat\b\vphantom{\b}^{\lambda}}^2_{\R^{2J}} &\leq \frac{1}{J} \sum_{j=1}^J \snorm{\bar\f (e^{a^*_j} R_{\alpha_j^*} - e^{\hat a_j^\lambda} (\overline{e^{\vphantom{\hat \Aa^{\lambda}}\Aa^{*}} R_{\balpha^*}})(\overline{e^{\hat \Aa^{\lambda}} R_{\hat \balpha^{\lambda}}})^{-1}R_{\hat\alpha_j^\lambda}  )  }^2_{\R^2} \label{eq.term1} \\
&\qquad\qquad + \snorm{\bar\b^* (e^{a^*_j} (\overline{e^{\vphantom{\hat \Aa^{\lambda}}\Aa^{*}} R_{\balpha^*}})^{-1}R_{\alpha_j^*} - e^{\hat a_j^\lambda} (\overline{e^{\hat \Aa^{\lambda}} R_{\hat \balpha^{\lambda}}})^{-1}R_{\hat\alpha_j^\lambda}  )  }^2_{\R^2}  \label{eq.term2}\\
&\qquad\qquad\qquad + \snorm{e^{a^*_j} \bar\bzeta_j R_{\alpha^*_j} }^2_{\R^2}  + \snorm{e^{\hat a_j^\lambda} (\overline{e^{a^*\vphantom{\hat \Aa^{\lambda}}} \bzeta R_{\alpha^*}})(\overline{e^{\hat \Aa^{\lambda}} R_{\hat \balpha^{\lambda}}})^{-1} R_{\hat\alpha_j^\lambda} }^2_{\R^2}   \label{eq.term3}
\end{align} 
where $(\overline{e^{a^*} \bzeta R_{\alpha^*}}) = \frac{1}{J}\sum_{j=1}^J e^{a^*_j}\bar \bzeta_jR_{\alpha^*_j} \in\R^2$. The rest of the proof is devoted to the control of the terms \eqref{eq.term1}, \eqref{eq.term2} and \eqref{eq.term3}.

In this section, we denotes by $\| \cdot \|_{op}$ the operator norm of a $2\times 2$ matrix, \ie $\| A\|_{op} = |\gamma_{\max}(A)|$ where $\gamma_{\max}(A) $ denotes the largest eigenvalue of a matrix $A\in\R^{2\times 2}$. Note by the way that the eigenvalues of the matrix $\frac{1}{J}\sum_{j=1}^J e^{a_j}R_{\alpha_j}\in\R^{2\times 2}$  are $\frac{1}{J}\sum_{j=1}^J e^{a_j}( \cos(\alpha_j) \pm i \sin(\alpha_j)) $ for any $J$ and $(a_1,\ldots,a_J,\alpha_1,\ldots,\alpha_J)\in [-A,A]^J\times [-\AA,\AA]^J$. It yields  $\snorm{ e^{- a_j}R_{-\alpha_j}}_{op} \leq e^A$ and $\snorm{(\overline{e^{\Aa} R_{\balpha}})^{-1} }_{op} \leq e^A(\cos(\AA) - \sin(\AA))^{-1}$ which is a positive real number since by hypothesis we have $\AA<\frac{\pi}{4}$. We are now able to derive an upper bound for \eqref{eq.term1},
\begin{align*}
\eqref{eq.term1} & \leq \snorm{\bar\f}_{\R^2}^2 \frac{1}{J}\sum_{j=1}^J\snorm{ e^{a^*_j} R_{\alpha_j^*} - e^{\hat a_j^\lambda} (\overline{e^{\vphantom{\hat \Aa^{\lambda}}\Aa^{*}} R_{\balpha^*}})(\overline{e^{\hat \Aa^{\lambda}} R_{\hat \balpha^{\lambda}}})^{-1}R_{\hat\alpha_j^\lambda} }_{op}^2\\
&\leq e^{2A} \snorm{\bar\f}_{\R^2}^2 \frac{1}{J}\sum_{j=1}^J\snorm{ e^{a^*_j} R_{\alpha_j^*} (e^{-\hat a_j^\lambda}R_{-\hat\alpha_j^\lambda}) -  (\overline{e^{\vphantom{\hat \Aa^{\lambda}}\Aa^{*}} R_{\balpha^*}})(\overline{e^{\hat \Aa^{\lambda}} R_{\hat \balpha^{\lambda}}})^{-1}}_{op}^2 \\
& \leq 2e^{2A} \snorm{\bar\f}_{\R^2}^2 \frac{1}{J}\sum_{j=1}^J \snorm{ e^{-\hat a_j^\lambda}R_{-\hat\alpha_j^\lambda}}_{op}^2 \snorm{ e^{a^*_j} R_{\alpha_j^*} -e^{\hat a_j^\lambda}R_{\hat\alpha_j^\lambda} }_{op}^2 \\
& \qquad\qquad +2e^{2A} \snorm{\bar\f}_{\R^2}^2\snorm{(\overline{e^{\hat \Aa^{\lambda}} R_{\hat \balpha^{\lambda}}})^{-1} }_{op}^2 \snorm{  (\overline{e^{\vphantom{\hat \Aa^{\lambda}}\Aa^{*}} R_{\balpha^*}}) - (\overline{e^{\hat \Aa^{\lambda}} R_{\hat \balpha^{\lambda}}})}_{op}^2 
\end{align*}
%Recall the proof of Lemma \ref{lemme.bernstein2} where it was shown that when $\AA<\frac{\pi}{4}$ there exists a constant $C(A,\AA) \geq \snorm{(\overline{e^{\hat \Aa^{\lambda}} R_{\hat \balpha^{\lambda}}})^{-1} }_{op}$. It yields,
%$$
%\eqref{eq.term1}  \leq C(A,\AA)\snorm{\bar\f}_{\R^2}(\snorm{ e^{a^*_j} R_{\alpha_j^*} -e^{\hat a_j^\lambda}R_{\hat\alpha_j^\lambda} }_{op}  )
%$$
%for some constant $C(A,\AA)$. 
It is now easy to show that there  exists  a constant $C(A,\AA,\f)>0$ independent of $k$ and $J$ such that
\begin{equation}\label{eq.conterm1}
\eqref{eq.term1} \leq C(A,\AA,\f) \frac{1}{J}\sum_{j=1}^J\snorm{(a^*_j - \hat a_j^\lambda , \alpha^*_j - \hat \alpha_j^\lambda )}_{\R^2}^2.
\end{equation} 
The term \eqref{eq.term2} is very similar to the term \eqref{eq.term1} and we have 
\begin{equation}\label{eq.conterm2}
\eqref{eq.term2} \leq C(A,\AA,B) \frac{1}{J}\sum_{j=1}^J\snorm{(a^*_j - \hat a_j^\lambda , \alpha^*_j - \hat \alpha_j^\lambda )}_{\R^2}^2,
\end{equation} 
for some constant $C(A,\AA,B)>0$ independent of $k$ and $J$. By using formula \eqref{eq.conterm1}, \eqref{eq.conterm2} and Proposition \ref{prop.ConvAAlpha} together, there exists a constant $C(L,s,A,\AA,B,\f)$ independent of $k$ and $J$ such that for all $x>0$,
\begin{equation}\label{eq.con1}
\P\Big(\eqref{eq.term1} + \eqref{eq.term2} \geq C(L,s,A,\AA,B,\f) \big(F(k^{-\frac{2s}{2s+1}}) + F(V_1(k,J,x)) \big) \Big) \leq 4e^{-x},
\end{equation}
where $V_1(k,J,x)$ is defined in statement of proposition \ref{prop.ConvAAlpha} and $F(u) = u + \sqrt{u}$ for $u\geq 0$.

The term \eqref{eq.term3} can be bounded in probability as follows. First of all, remark that there exists a constant $C(A,\AA)>0$ such that 
$$
\eqref{eq.term3} \leq C(A,\AA) \frac{1}{J}\sum_{j=1}^J \snorm{\bar\bzeta_j}^2_{\R^2}.
$$
Then, for all $j=1,\ldots,J$, the random variable $\bar\bzeta_j =\frac{1}{k}\sum_{\ell=1}^k ([\bzeta_j^{(1)}]_\ell,[\bzeta_j^{(2)}]_\ell ) \in\R^{2}$ can be written 
$$
(\bar \bzeta^1_j,\bar\bzeta^2_j)' = \frac{1}{k} \begin{pmatrix}
\1_{ k}' & 0_{ k}' \\
0_{k}' & \1_{ k}'
\end{pmatrix} ([\bzeta^{(1)}_j]_1,\ldots,[\bzeta^{(1)}_j]_k,[\bzeta^{(2)}_j]_1,\ldots, [\bzeta^{(2)}_j]_k)' 
$$
where $ 0_{k}$ is a column vector of $k$ zeros. The random vector $(\bar \bzeta^1_j,\bar\bzeta^2_j)'$ is a two dimensional centered Gaussian vector of variance  $
\mathbf V = \frac{1}{k^2}\begin{pmatrix}
\1_{k}' & 0_{k}' \\
0_{k}' & \1_{k}'
\end{pmatrix} 
\mathbf \Sigma 
\begin{pmatrix}
\1_{k} & 0_{ k} \\
0_{ k} & \1_{k}
\end{pmatrix}
$. The $2\times 2$ matrix  $\mathbf V$ is  of trace less or equal to $\frac{2}{k}\gamma_{\max}(k)$. Hence, the random variable $\frac{1}{J}\sum_{j=1}^J \snorm{ \bar\bzeta_j}^2_{\R^2}  = \frac{1}{J}\sum_{j=1}^J (\bar \bzeta^{(1)}_j,\bar\bzeta^{(2)}_j)(\bar \bzeta^{(1)}_j,\bar\bzeta^{(2)}_j)'$ has the same probability distribution as $\frac{1}{J} \xi' [Id_J \otimes \mathbf{V}] \xi $ where $\xi$ is a centered Gaussian vector of variance $Id_{2J}$. %square of its norm follows a $\chi^2$ distribution and we have the following 
A standard concentration inequality of $\chi^2$ distribution (see e.g. \cite{LauMa} Lemma 1) is $\P\left(  \xi' [Id_J \otimes \mathbf{V}] \xi\geq  J\frac{4}{k}\gamma_{\max}(k) +   \frac{4}{k}\gamma_{\max}(k)\sqrt{2J x} +  x \frac{4}{k}\gamma_{\max}(k) \right) \leq e^{-x} $ for any $x>0$. It yields that there exists a constant $C(A,\AA)>0$ such that for all $x>0$ 
\begin{equation}\label{eq.conterm3}
\P\left( \eqref{eq.term3}  \geq  C(A,\AA)\frac{\gamma_{\max}(k)}{k} \left(1 +  \sqrt{2\frac{x}{J}} + \frac{x}{J}\right)\right) \leq e^{-x}.
\end{equation}

To end the proof remark that for all $u\geq 0$ we have $\P(\frac{1}{J}\snorm{\b^*_{\bTheta_0} - \hat\b\vphantom{\b}^{\lambda}}^2_{\R^{2J}}\geq u ) \leq \P(\eqref{eq.term1} + \eqref{eq.term2} + \eqref{eq.term3} \geq u )$. This together with \eqref{eq.con1} and  \eqref{eq.conterm3} yield that there exists a constant $C(L,s,A,\AA,B,\linebreak[1]\f)>0$ independent of $k$ and $J$ such that for all $x>0$ we have 
$$
\P\Big(\frac{1}{J}\snorm{\b^*_{\bTheta_0} - \hat\b\vphantom{\b}^{\lambda}}^2_{\R^{2J}}\geq C(L,s,A,\AA,B,\f) \big( F(k^{-\frac{-2s}{2s+1}}) + F(V_1(k,J,x)) + V_2(k,J,x) \big) \Big)  \leq 5 e^{-x},
$$
where $V_2(k,J,x) =\frac{\gamma_{\max}(k)}{k} \left(1 +  \sqrt{2\frac{x}{J}} + \frac{x}{J}\right) $ and the proof of Proposition \ref{prop.ConvShift} is completed. \qed

%%%%%%%%%%%%%%%%%%%% TECHNICAL LEMMAS %%%%%%%%%%%%%%%%%%%%%%%%%%%%%%%

\section{Technical Lemma}

\begin{lemma}\label{lemme.biais}
Assume that Assumption \ref{ass.f} holds, \ie $f\in H_s(L)$ with $s>\frac{3}{2}$  (see \eqref{eq:H}) and $\f = (f(\frac{\ell}{k}))_{\ell=1}^k \in\Sh$. If $\lambda(k) = k^{\frac{1}{2s+1}} $ then there exists a constant $C(L,s)$ such that for all $f\in H_s(L)$ we have
$$
\frac{1}{k}\snorm{\A^\lambda \f -\f}_{\Sh}^2 \leq C(L,s) k^{-\frac{2s}{2s+1}},
$$ 
where $\A^\lambda$ is the projection matrix defined in \eqref{eq:decomp}.
\end{lemma}

\begin{proof}
%The matrix $\A^{\lambda}$ corresponds to a low pass filter of frequency cutoff $\lambda\in\N$. 
Recall the notations introduced Sections \ref{part.Intro.1} and \ref{part:RedDim} :  $c_m(f) = (c_m(f^{(1)}),c_m(f^{(2)})) \in\C^2$ is the $m$-th Fourier coefficient of $f\in L^2([0,1],\R^2)$ and $c_m(\f)= (c^{(1)}_m(\f),c^{(2)}_m(\f)) \in\C^2$ is the $m$-th discrete Fourier coefficient of $\f\in\Sh$. Thus, we have by Parseval's equality
$$
\frac{1}{k}\bnorm{\A^{\lambda}\f-\f}^2_{\Sh} =\frac{1}{k}\bnorm{ \bigg( \frac{1}{k} \sum_{ |m| >\lambda}  c_m(\f) e^{i2\pi m \frac{\ell}{k} } \bigg)_{\ell=1}^{k} }^2_{\Sh}   = \frac{1}{k^2} \sum_{\abs{m}>\lambda} \bnorm{c_m(\f)}_{\C^2}^2.
$$
where for all $c = (c^{(1)},c^{(2)})\in\C^2$, $\| c \|_{\C^2}^2 =  |c^{(1)}|^2 +|c^{(2)}|^2$.
It yields
$$
\frac{1}{k}\snorm{\A^{\lambda}\f-\f}^2_{\Sh} \leq 2 \sum_{\abs{m}>\lambda} \snorm{\frac{1}{k} c_m(\f) - c_m(f)}^2_{\C^2} + 2\sum_{\abs{m}>\lambda}\snorm{c_m(f)}^2_{\C^2}.
$$
Firstly, Lemma 1.10 in \cite{Tsy} ensures that if $s>\frac{1}{2}$ then $\abs{\frac{1}{k} c_m(\f^{(i)}) - c_m(f^{(i)})}\leq C(L,s) k^{-s+\frac{1}{2}}$ for any $m\in\N$ and $i=1,2$. Secondly, equation (1.87) in \cite{Tsy} ensure that if $\lambda(k) = k^{\frac{1}{2s+1}}$ then  $\sum_{\abs{m}>\lambda}  \sabs{ c_m(f^{(i)})}^2 \leq C(L,s) k^{-\frac{2s}{2s+1}}$. Thence, there is a constant $C(L,s)$ independent of $k$ such that
$$
\frac{1}{k}\snorm{\A^{\lambda}\f-\f}^2_{\Sh} \leq C(L,s) (k^{2-2s} + k^{-\frac{2s}{2s+1}}).
$$
If $s\geq \frac{3}{2}$ then we have 
\begin{equation*}
\frac{1}{k}\snorm{\A^{\lambda}\f-\f}^2_{\Sh} \leq C(L,s) k^{-\frac{2s}{2s+1}}.\tag*{\qedhere}
\end{equation*} 
\end{proof}

We now give a general expression of the gradient and of the Hessian of the criterion $D$. Recall that we have %defined in \eqref{eq.critD} and that can be written,
{\small
$$
D(\Aa,\balpha,\b) = \frac{1}{J}\sum_{j=1}^J\bnorm{e^{-a_{j}} (e^{a_{j}^*}\f R_{\alpha_{j}^*} + \1_k\otimes ( b^*_j -b_j) )R_{- \alpha_{j}} + \frac{1}{J}\sum_{j'=1}^Je^{-a_{j'}} (e^{a_{j'}^*}\f R_{\alpha_{j'}^*} + \1_k\otimes ( b^*_{j'} -b_{j'}) )R_{- \alpha_{j'}} }^2_{\Sh}
$$}where $(\Aa,\balpha,\b) = (a_1,\ldots,a_J,\alpha_1,\ldots,\alpha_J,b_1,\ldots,b_J)\in\R^J\times[-\pi,\pi[^J\times\R^{2J}$. To shorten the formulas below we note $g_j  =(g_j^{(1)},g_j^{(2)}, g_j^{(3)},g_j^{(4)})=(a_j,\alpha_j,b_j)$, that is $g_j^{(1)} = a_j$, $g_j^{(2)} =\alpha_j$, $g^{(3)}_j = b_j^{(1)}$ and $g^{(4)}_j = b_j^{(2)}$. Let $\f_{g_j}=g^{-1}.g^*_j.\f = e^{-a_{j}} (e^{a_{j}^*}\f R_{\alpha_{j}^*} + \1_k\otimes ( b^*_j -b_j) )R_{- \alpha_{j}}$, and  for all $j_1=1,\ldots,J$ and $p_1 = 1,\ldots,4$,
%where $(\Aa,\balpha) = (a_1,\ldots,a_J,\alpha_1,\ldots,\alpha_J)\in\R^J\times[-\pi,\pi[^J$. To shorten the formulas below we note $\btheta_j  =(\theta_j^{(1)},\theta_j^{(2)})=(a_j,\alpha_j)$, that is $\theta_j^{(1)} = a_j$, $\theta_j^{(2)} =\alpha_j$. Let $\f_{\btheta_j}= e^{a_{j}^*-a_{j}} \f R_{\alpha_{j}^* - \alpha_{j}}$, and  for all $j_1=1,\ldots,J$ and $p_1 = 1,\ldots,4$,
\begin{align}\label{eq:GradD}
\partial_{g_{j_1}^{(p_1)}} D(\Aa,\balpha,\b)& = \frac{2}{Jk}\bprs{\partial_{g^{(p_1)}_{j_1}}   \f_{g_{j_1}},  \f_{g_{j_1}}-\frac{1}{J}\sum_{j'=1}^J \f_{g_{j'}}}_{\Sh}.
\end{align}
The second order derivatives are
\begin{align}
\partial_{g_{j_2}^{(p_2)}} \partial_{g_{j_1}^{(p_1)}}D(\Aa,\balpha,\b) & = -\frac{2}{J^2k}\bprs{ \partial_{g^{(p_1)}_{j_1}} \f_{g_{j_1}}, \partial_{g^{(p_2)}_{j_2}} \f_{g_{j_2}}}_{\Sh},\qquad \qquad \text{ if } j_1 \neq j_2, \label{eq:hessD1}\\
\partial_{g_{j_1}^{(p_2)}} \partial_{g_{j_1}^{(p_1)}} D(\Aa,\balpha,\b) & =\frac{2}{Jk} \bprs{ \partial_{g^{(p_2)}_{j_1}} \partial_{g^{(p_1)}_{j_1}}\f_{g_{j_1}},\bigg(\f_{g_{j_1}}  -\frac{1}{J} \sum_{j'=1}^J \f_{g_{j'}}\bigg)}_{\Sh}\nonumber\\ &\qquad\qquad\qquad+\left( \frac{2}{Jk}-\frac{2}{J^2k}\right)\bprs{ \partial_{g^{(p_1)}_{j_1}} \f_{g_{j_1}} ,\partial_{g^{(p_2)}_{j_1}}\f_{g_{j_1}}}_{\Sh}. \label{eq:hessD2}
\end{align}

The expressions of the gradient and the Hessian of $D$ simplify on the set $(\Aa^*,\balpha^*,\b^*) * \G$, see Lemma \ref{lemma.zeros}. For any $g_0 =(a_0,\alpha_0,b_0)\in\R\times[-\pi,\pi[\times\R^2$, %each coordinate of  $(\Aa^*,\balpha^*,\b^*) * (a_0,\alpha_0,b_0)$ can be written $\btheta_{j_1}^*.\btheta_0 = (a^*_{j_1},\alpha^*_{j_1},b^*_{j_1}).(a_0,\alpha_0,b_0) $ and we have %by definition, for each $J=1,\ldots,J$, $\btheta_j = =( $ %= \btheta^* * \btheta_0 = (\btheta^*_1.\btheta_0,\ldots, \btheta^*_J.\btheta_0 )$, 
we have $\f_{g_{j_1}^*.g_0} - \frac{1}{J}\sum_{j'=1}^J \f_{g_{j'}^*.g_0} = e^{-a_0} (\f - \1_k \otimes b_0) R_{-\alpha_0} -\frac{1}{J}\sum_{j'=1}^J e^{-a_0} (\f - \1_k \otimes b_0) R_{-\alpha_0} =0$. 
It yields that for all $(a_0,\alpha_0,b_0)\in\R\times[-\pi,\pi[\times\R^2$ we have  $\nabla D((\Aa^*,\balpha^*,\b^*) * (a_0,\alpha_0,b_0)) = 0$, and 
\begin{equation}
\partial_{g_{j_2}^{(p_2)}} \partial_{g_{j_1}^{(p_1)}}D((\Aa^*,\balpha^*,\b^*) * (a_0,\alpha_0,b_0)) =  \begin{cases}-\frac{2}{J^2k}\bprs{ \partial_{g^{(p_1)}_{j_1}} \f_{g_{j_1}^*. g_0}, \partial_{g^{(p_2)}_{j_2}} \f_{g_{j_2}^*. g_0}}_{\Sh}, & \text{ if } j_1\neq j_2, \\ \left( \frac{2}{Jk}-\frac{2}{J^2k}\right)\bprs{ \partial_{g^{(p_1)}_{j_1}} \f_{g_{j_1}^*. g_0}, \partial_{g^{(p_2)}_{j_2}} \f_{g_{j_1}^*. g_0}}_{\Sh}.\end{cases}\label{eq.hes0}
\end{equation}

\begin{lemma}\label{lemme.eigmin}
The smallest eigenvalue of $\nabla^2 D_0(\Aa_{\bTheta_{0}}^*, \balpha_{\bTheta_0}^*)$ restricted to the subset $\bTheta_0$ is greater than $e^{-2A}\frac{2}{Jk}\norm{\f_0}^2_{\Sh}$.
\end{lemma}

\begin{proof}
In this proof, $\bh = (\Aa,\balpha)\in \R^J\times[-\pi,\pi[^J$ % that is $ \bh_j^{(1)} = a_j$ the parameter of scaling and $ \bh_j^{(2)} = \alpha_j $ the parameter of rotation. Moreover 
and $\f^0_{h_j}=%\btheta_{j}^{-1}.\btheta_{j}^*.\f  = 
e^{a_{j}^*-a_{j}} \f_0 R_{\alpha_j^* - \alpha_j}$. We have  $\partial_{a_{j_1}} \f^0_{h_{j_1}} = - \f^0_{h_{j_1}} $ and $\partial_{\alpha_{j_1}} \f^0_{h_{j_1}} = \f^0_{h_{j_1}} R_{-\frac{\pi}{2} }$ for all $ j_1=1,\ldots,J$. By using Formulas \eqref{eq.hes0}, the Hessian of $D_0$ at the point $ (\Aa^*_{\bTheta_0},\balpha_{\bTheta_0}) = (\Aa^* - \bar \Aa^*, \balpha^* -\bar  \balpha^*)\in\bTheta^{\Aa,\balpha}_0$ is given by, 
\begin{align*}
\partial_{\alpha_{j_1}} \partial_{\alpha_{j_2}} D_0(\Aa^* - \bar \Aa^*, \balpha^* - \bar \balpha^*)  = \partial_{a_{j_1}} \partial_{a_{j_2}} D_0(\Aa^* - \bar \Aa^*, \balpha^* - \bar \balpha^*)  = \begin{cases} - \frac{2}{J^2k} \snorm{e^{\bar \Aa^*} \f_0}^2_{\Sh} %\frac{2}{J^2k}\prs{ e^{a_{j_1}^* - a_{j_1}} \f^* R_{\alpha_{j_1}^* -\alpha_{j_1} }  , e^{a_{j_2}^* - a_{j_2}} \f^* R_{\alpha_{j_2}^* -\alpha_{j_2} }}
, \quad \text{ if } j_1\neq j_2,  \\ \left( \frac{2}{Jk}-\frac{2}{J^2k}\right) \snorm{e^{\bar \Aa^*} \f_0}^2_{\Sh}%\frac{2}{Jk} \prs{ e^{a_{j_1}^* - a_{j_1}} \f^* R_{\alpha_{j_1}^* -\alpha_{j_1}}  ,  e^{a_{j_1}^* - a_{j_1}} \f^* R_{\alpha_{j_1}^* -\alpha_{j_1}} - \frac{1}{J}\sum_{j=1}^{J} e^{a_{j}^* - a_{j}} \f^* R_{\alpha_{j}^* -\alpha_{j} }}  +\left( \frac{2}{Jk}-\frac{2}{J^2k}\right)\prs{,  }
\end{cases}
\end{align*}
and the second order cross derivatives are 0. % By reordering the terms $(a_1,\ldots,a_J,\alpha_1,\ldots,\alpha_J) \in\R^{2J}$ t
The Hessian of $D_0$ at $(\Aa^* - \bar \Aa^*, \balpha^* -\bar \balpha^*)$ can be written,
\begin{equation}\label{eq.HessD0}
\nabla^2 D_0(\Aa^* - \bar \Aa^*, \balpha^* -\bar \balpha^*)= \frac{2}{J^2k} \snorm{e^{\bar\Aa^*} \f_0}^2_{\Sh}\begin{pmatrix}
J Id_J - \1_{J\times J} & 0 \\ 
0 & J Id_J - \1_{J\times J} 
\end{pmatrix},
\end{equation}
where $Id_J$ is the identity in $\R^J$ and $\1_J$ is the $J\times J$ matrix of ones. The eigenvalues of $J Id_J - \1_{J\times J}$ are 0 with eigenvector $\1_J$ and $J$ with eigenspace $\1_J^{\perp}$. It yields that on $\bTheta^{\Aa,\balpha}_0$ the smallest eigenvalue of  $\nabla^2 D_0(\Aa^* - \bar \Aa^*, \balpha^* -\bar \balpha^*)$ is $e^{2\bar\Aa^*}\frac{2}{Jk}\snorm{\f_0}^2_{\Sh}$. To finish the proof, remark that $\bar \Aa^* \leq A$.
\end{proof}

\begin{lemma} \label{lem:D3}
%Let $\btheta = (\Aa,\balpha)\in\R^{2J}$. 
Let  $\delta = \max\{A,\AA\}$. For all $\mathbf c = (a_1,\ldots,a_J,\alpha_1,\ldots,\alpha_J)\in\bTheta_0^{\Aa,\balpha}$ and $(\Aa,\balpha) \in  \R^{2J}$ we have% satisfying $|a_j|, |\alpha_j|\leq \delta$ for each $j=1,\ldots,J$, then 
$$
\abs{\big[\nabla^3 D_0(\mathbf{c})\big](\Aa,\balpha) } \leq 40 \delta e^{2A} \snorm{\f_0}^2_{\Sh} \frac{1}{Jk}\norm{(\Aa,\balpha)}^2_{\R^{2J}}.
$$
%where $A,\AA<\delta$.
\end{lemma}
\begin{proof}
In this proof, $\bh = (\Aa,\balpha)\in\R^{J}\times[-\pi,\pi[^{J}$, that is $ h_j^{(1)} = a_j$ is the parameter of scaling and $ h_j^{(2)} = \alpha_j $ is the parameter of rotation. We have $\f^0_{h_j} = e^{a_{j}^*-a_{j}} \f_0 R_{\alpha_j^* - \alpha_j}$. Then, from equations \eqref{eq:hessD1} and \eqref{eq:hessD2}, it follows that  for all $j_1,j_2,j_3 = 1,\ldots,J$ and $p_1,p_2,p_3 =1,\ldots,2$,
\begin{align*}
\partial_{h_{j_3}^{(p_3)}} \partial_{h_{j_2}^{(p_2)}} \partial_{h_{j_1}^{(p_1)}}  D_0(\Aa,\balpha) & = 0, \qquad   \text{ if } j_1 \neq j_2 \text{ and } j_2 \neq j_3 \text{ and } j_1 \neq j_3,\\
 \partial_{h_{j_2}^{(p_3)}} \partial_{h_{j_1}^{(p_2)}} \partial_{h_{j_1}^{(p_1)}} D_0(\Aa,\balpha) & = -\frac{2}{J^2k}\bprs{\partial_{h_{j_1}^{(p_2)}}\partial_{h_{j_1}^{(p_1)}} \f^0_{h_{j_1}},\partial_{h_{j_2}^{(p_3)}} \f^0_{h_{j_2}}}_{\Sh},  \qquad \text{ if } j_1 \neq j_2, \\
 \partial_{h_{j_1}^{(p_3)}}\partial_{h_{j_1}^{(p_2)}}\partial_{h_{j_1}^{(p_1)}} D_0(\Aa,\balpha) & = \frac{2}{Jk} \bprs{\partial_{h_{j_1}^{(p_3)}}\partial_{h_{j_1}^{(p_2)}}\partial_{h_{j_1}^{(p_1)}} \f^0_{h_{j_1}} ,\bigg(\f^0_{h_{j_1}}-\frac{1}{J} \sum_{j'=1}^J \f^0_{h_{j'}}\bigg) }_{\Sh} \\ &\quad + \left(\frac{2}{Jk} -\frac{2}{J^2k} \right) \bigg( \bprs{ \partial_{h^{(p_2)}_{j_1}}\partial_{h^{(p_1)}_{j_1}} \f^0_{h_{j_1}},\partial_{h_{j_1}^{(p_3)}} \f^0_{h_{j_1}}}_{\Sh}+\bigg\langle \partial_{h^{(p_3)}_{j_1}}\partial_{h^{(p_1)}_{j_1}} \f_{h_{j_1}}^0, \\ & \hspace*{-4.41069pt}\hspace*{4cm}\partial_{h_{j_1}^{(p_2)}} \f^0_{h_{j_1}}\bigg\rangle_{\Sh}+\bprs{ \partial_{h^{(p_3)}_{j_1}}\partial_{h^{(p_2)}_{j_1}} \f^0_{h_{j_1}},\partial_{h_{j_1}^{(p_1)}} \f^0_{h_{j_1}}}_{\Sh} \bigg)% \text{ if } i = 1,\ldots,J.
\end{align*}
By Cauchy-Schwarz inequality we have,
\begin{align}\label{eq.bound1}
\left| \prs{\partial_{h_{j_1}^{(p_2)}}\partial_{h_{j_1}^{(p_1)}} \f^0_{h_{j_1}},\partial_{h_{j_2}^{(p_3)}} \f^0_{h_{j_2}}}_{\Sh} \right| \leq \snorm{\partial_{h_{j_1}^{(p_2)}}\partial_{h_{j_1}^{(p_1)}} \f^0_{h_{j_1}}}_{\Sh}\snorm{\partial_{h_{j_2}^{(p_3)}} \f^0_{h_{j_2}}}_{\Sh} \leq e^{2A} \snorm{\f_0}^2_{\Sh}.
\end{align}
and
\begin{align}\label{eq.bound2}
\bigg| \bigg\langle\partial_{h_{j_1}^{(p_3)}}\partial_{h_{j_1}^{(p_2)}}\partial_{h_{j_1}^{(p_1)}} &\f^0_{h_{j_1}} ,\bigg(\f^0_{h_{j_1}}-\frac{1}{J} \sum_{j'=1}^J \f^0_{h_{j'}}\bigg) \bigg\rangle_{\Sh}  \bigg| \nonumber\\
&\leq \norm{\partial_{h_{j_1}^{(p_3)}}\partial_{h_{j_1}^{(p_2)}}\partial_{h_{j_1}^{(p_1)}} \f^0_{h_{j_1}}}_{\Sh} \Big\| \f^0_{h_{j_1}}-\frac{1}{J} \sum_{j'=1}^J \f^0_{h_{j'}} \Big\|_{\Sh} \leq 2e^{2A}\snorm{\f_0}^2_{\Sh} %\leq 2C(\Theta,\F,f^*) \norm{f^*}_{L^2} \leq C(\Theta,\F) 
\end{align}
For $\kappa = (\kappa_1,\ldots,\kappa_{2J})\in\N^{2J}$, denote by $\abs{\kappa} = \kappa_1+\ldots+\kappa_{2J}$ and $$(\partial_{\bh})^{\kappa} = (\partial_{a_{1}})^{\kappa_1} (\partial_{\alpha_{1}})^{\kappa_2} \ldots (\partial_{ a_{J}})^{\kappa_{2J-1}} (\partial_{\alpha_{J}})^{\kappa_{2J}}.$$ Then, the differential of order 3 of $D_0$ at $\mathbf c \in \bTheta^{\Aa,\balpha}_0$ applied at $(\Aa,\balpha) \in \R^{2J}$ writes as
$$
\big[\nabla^3 D_0(\mathbf{c})\big](\Aa,\balpha)  = \sum_{\abs{\kappa}=3} (\partial_{\bh})^{\kappa} D_0(\mathbf c) \bh^{\kappa} 
$$
where $\bh^{\kappa} = a_{1}^{\kappa_1}\alpha_{1}^{\kappa_2} \ldots a_{J}^{\kappa_{2J-1}} \alpha_{J}^{\kappa_{2J}} $. This formula together with equations \eqref{eq.bound1} and \eqref{eq.bound2} give,
\begin{align*}
\babs{\big[\nabla^3 D_0(\mathbf{c})\big](\Aa,\balpha)} & = \bigg | \sum_{p_1,p_2,p_3=1}^2 \sum_{j_1=1}^J \partial_{h_{j_1}^{(p_3)}}\partial_{h_{j_1}^{(p_2)}}\partial_{h_{j_1}^{(p_1)}} D_0(\mathbf c)  h_{j_1}^{(p_1)}h_{j_1}^{(p_2)}h_{j_1}^{(p_3)} \\ & \hspace{15em}+ 3\sum_{j_1\neq j_2 = 1}^J \partial_{h_{j_2}^{(p_3)}}\partial_{h_{j_1}^{(p_2)}}\partial_{h_{j_1}^{(p_1)}} D_0(\mathbf c)h_{j_1}^{(p_1)}h_{j_1}^{(p_2)}h_{j_2}^{(p_3)}\bigg |\\
%& \leq C(h,\F,f^*) \sum_{p_1,p_2,p_3=1}^p \abs{(\frac{4}{J}-\frac{2}{J^2} )\sum_{j_1=1}^J \abs{\btheta_{j_1}^{p_1}\btheta_{j_1}^{p_2}\btheta_{j_1}^{p_3}} -\frac{6}{J^2} \sum_{j_1\neq j_2 = 1}^J \abs{\btheta_{j_1}^{p_1}\btheta_{j_1}^{p_2}\btheta_{j_2}^{p_3}}}\\
& \leq 2 e^{2A} \frac{1}{k}\snorm{\f_0 }^2_{\Sh} \sum_{p_1,p_2,p_3=1}^2 \bigg(\frac{4}{J}\sum_{j_1=1}^J  \sabs{h_{j_1}^{(p_1)}h_{j_1}^{(p_2)}h_{j_1}^{(p_3)}} +\frac{6}{J^2} \sum_{j_1\neq j_2 = 1}^J \sabs{h_{j_1}^{(p_1)}h_{j_1}^{(p_2)}h_{j_2}^{(p_3)}}\bigg)\\
&\leq  2 \delta e^{2A}\frac{1}{k} \snorm{\f_0 }^2_{\Sh}\sum_{p_1,p_2=1}^2 \bigg(\frac{4}{J}\sum_{j_1=1}^J \sabs{h_{j_1}^{(p_1)}h_{j_1}^{(p_2)}} +\frac{6(J-1)}{J^2} \sum_{j_1= 1}^J \sabs{h_{j_1}^{(p_1)}h_{j_1}^{(p_2)}}\bigg)\\
&\leq  20 \delta e^{2A} \frac{1}{k}\snorm{\f_0}^2_{\Sh}\frac{1}{J}\sum_{j=1}^J\sum_{p_1,p_2=1}^2 \sabs{h_{j}^{(p_1)}h_{j}^{(p_2)}} \\
&=  20 \delta e^{2A} \frac{1}{k}\snorm{\f_0}^2_{\Sh}\frac{1}{J}\sum_{j=1}^J\bigg(\sum_{p_1=1}^2\sabs{h_{j}^{(p_1)}}\bigg)^2 \leq 40 \delta e^{2A} \frac{1}{k}\snorm{\f_0 }^2_{\Sh} \frac{1}{J}\sum_{j=1}^J\sum_{p_1=1}^2\sabs{h_{j}^{(p_1)}}^2. %\leq \delta C(h,\F)\frac{1}{J}\norm{\btheta}^2.
\end{align*}
And the claim is now proved.
\end{proof}

\begin{lemma}\label{lemme.AssOp}
For all $\f\in\R^{k\times 2}$ and $(a_1,\alpha_1,b_1),(a_2,\alpha_2,b_2) \in [-A,A]\times [-\AA,\AA]\times [-B,B]^2$, let $g_i.\f = e^{a_i}\f R_{\alpha_i} +\1_k\otimes b_i$, $i=1,2$. Then, we have
$$
 \frac{1}{k}\snorm{g_1.\f - g_2.\f}^2_{\Sh} \leq C(A,\f) \snorm{(a_1,\alpha_1,b_1)-(a_2,\alpha_2,b_2)}^2_{\R^{4}},
$$ 
where $C(A,\f)=2\max\{4e^{4A}\frac{1}{k}\snorm{\f}^2_{\Sh},1\}$.
\end{lemma}

\begin{proof}
We have 
\begin{equation}\label{eq.maaaj}
\frac{1}{k}\snorm{g_1.\f - g_2.\f}^2_{\Sh} %=  \frac{1}{k}\snorm{e^{a_1}\f^*R_{\alpha_1} +\b_1 -e^{a_2}\f^*R_{\alpha_2} +\b_2}^2 
\leq 2e^{2A} \frac{1}{k}\snorm{e^{a_1-a_2}\f R_{\alpha_1-\alpha_2} - \f}^2_{\Sh }+2 \frac{1}{k}\snorm{\1_k\otimes ( b_1 - b_2)}^2_{\Sh}
\end{equation}
Let now $F(a,\alpha) = \frac{1}{\sqrt{k}}\snorm{e^{a}\f R_{\alpha} - \f}_{\Sh }$. We have 
$
|\partial_{a} F(a,\alpha)|= \frac{1}{\sqrt{k}\snorm{e^{a}\f R_{\alpha} - \f}}_{\Sh }   |\langle e^a\f R_{\alpha} ,\linebreak[1] e^{a}\f R_{\alpha} - \f \rangle_{\Sh }| \leq  \frac{e^{A}}{\sqrt k}\snorm{\f}_{\Sh } \linebreak[1]\text{ and }\linebreak[1] |\partial_{\alpha} F(a,\alpha)| = \frac{1}{\sqrt{k}\snorm{e^{a}\f R_{\alpha} - \f}}_{\Sh } |\langle e^a\f R_{\alpha+\frac{\pi}{2}} ,\linebreak[1] e^{a}\f R_{\alpha} - \f \rangle_{\Sh }| \leq \frac{e^{A}}{\sqrt k}\snorm{\f}_{\Sh }.
$
The Euclidean norm of the gradient of $F$ satisfies 
$$
\snorm{\nabla F (a,\alpha)}_{\R^2} = \sqrt{|\partial_{a} F(a,\alpha)|^2 + |\partial_{\alpha} F(a,\alpha)|^2 }\leq \sqrt 2 e^{A}\frac{1}{\sqrt k}\snorm{\f}_{\Sh}.
$$
Since we have $\sabs{F(a,\alpha) } = \sabs{F(a,\alpha) -F(0,0)} \leq \sqrt{2} e^{A}\frac{1}{\sqrt k}\snorm{\f}_{\Sh}\snorm{(a,\alpha)}_{\R^2}$, equation \eqref{eq.maaaj} yields
$$
\frac{1}{k}\snorm{ g_1.\f - g_2.\f}^2_{\Sh} \leq 2\max\big \{4e^{4A}\tfrac{1}{k}\snorm{\f}^2_{\Sh},1\big\}\big (\sabs{a_1 - a_2}^2 + \sabs{\alpha_1 - \alpha_2}^2 + \sabs{b_1^{(1)} -  b_2^{(1)}}^2+ \sabs{b_1^{(2)} -  b_2^{(2)}}^2 \big ),
$$ 
which concludes the proof.
\end{proof}

\bibliographystyle{apalike2}%{alpha}
\bibliography{procrustean_mean2}

\begin{thebibliography}{}

\bibitem[Bhattacharya \& Patrangenaru, 2003]{bhat}
Bhattacharya, R. \& Patrangenaru, V. (2003).
\newblock Large sample theory of intrinsic and extrinsic sample means on
  manifolds. i.
\newblock {\em The Annals of Statistics}, 31(1), pp. 1--29.

\bibitem[Bigot \& Charlier, 2011]{BC11}
Bigot, J. \& Charlier, B. (2011).
\newblock On the consistency of fréchet means in deformable models for curve
  and image analysis.
\newblock {\em Electronic Journal of Statistics}, To appear.

\bibitem[Dryden \& Mardia, 1998]{MR1646114}
Dryden, I.~L. \& Mardia, K.~V. (1998).
\newblock {\em Statistical shape analysis}.
\newblock Wiley Series in Probability and Statistics: Probability and
  Statistics. Chichester: John Wiley \& Sons Ltd.

\bibitem[Goodall, 1991]{MR1108330}
Goodall, C. (1991).
\newblock Procrustes methods in the statistical analysis of shape.
\newblock {\em J. Roy. Statist. Soc. Ser. B}, 53(2), 285--339.

\bibitem[Gray, 2006]{Toep}
Gray, M. (2006).
\newblock {\em Toeplitz and circulant matrices: a review}.
\newblock Now Publishers Inc.

\bibitem[Horn \& Johnson, 1990]{MR1084815}
Horn, R. \& Johnson, C. (1990).
\newblock {\em Matrix analysis}.
\newblock Cambridge: Cambridge University Press.

\bibitem[Huckemann, 2011]{HuckemannSJS}
Huckemann, S. (2011).
\newblock Inference on 3d procrustes means: Tree bole growth, rank deficient
  diffusion tensors and perturbation models.
\newblock {\em Scand. J. Statist.}, To appear.

\bibitem[Kendall, 1984]{MR737237}
Kendall, D.~G. (1984).
\newblock Shape manifolds, {P}rocrustean metrics, and complex projective
  spaces.
\newblock {\em Bull. London Math. Soc.}, 16(2), 81--121.

\bibitem[Kendall et~al., 1999]{MR1891212}
Kendall, D.~G., Barden, D., Carne, T.~K., \& Le, H. (1999).
\newblock {\em Shape and shape theory}.
\newblock Wiley Series in Probability and Statistics. Chichester: John Wiley \&
  Sons Ltd.

\bibitem[Kent, 1992]{MR1175661}
Kent, J.~T. (1992).
\newblock New directions in shape analysis.
\newblock In {\em The art of statistical science}, Wiley Ser. Probab. Math.
  Statist. Probab. Math. Statist.  (pp.\ 115--127). Chichester: Wiley.

\bibitem[Kent \& Mardia, 1997]{MR1436569}
Kent, J.~T. \& Mardia, K.~V. (1997).
\newblock Consistency of {P}rocrustes estimators.
\newblock {\em J. Roy. Statist. Soc. Ser. B}, 59(1), 281--290.

\bibitem[Laurent \& Massart, 2000]{LauMa}
Laurent, B. \& Massart, P. (2000).
\newblock {Adaptive estimation of a quadratic functional by model selection.}
\newblock {\em Ann. Stat.}, 28(5), 1302--1338.

\bibitem[Le, 1998]{MR1618880}
Le, H. (1998).
\newblock On the consistency of procrustean mean shapes.
\newblock {\em Adv. in Appl. Probab.}, 30(1), 53--63.

\bibitem[Lele, 1993]{MR1225013}
Lele, S. (1993).
\newblock Euclidean distance matrix analysis ({EDMA}): estimation of mean form
  and mean form difference.
\newblock {\em Math. Geol.}, 25(5), 573--602.

\bibitem[Massart, 2007]{MR2319879}
Massart, P. (2007).
\newblock {\em Concentration inequalities and model selection}, volume 1896 of
  {\em Lecture Notes in Mathematics}.
\newblock Berlin: Springer.

\bibitem[Tsybakov, 2009]{Tsy}
Tsybakov, A. (2009).
\newblock {\em {Introduction to nonparametric estimation.}}
\newblock {Springer Series in Statistics. New York, NY: Springer. xii, 214~p.}

\bibitem[van~der Vaart, 1998]{VdW}
van~der Vaart, A.~W. (1998).
\newblock {\em Asymptotic statistics}, volume~3 of {\em Cambridge Series in
  Statistical and Probabilistic Mathematics}.
\newblock Cambridge: Cambridge University Press.

\end{thebibliography}

\end{document}